\documentclass{article}
\usepackage{amsmath,amsthm,amssymb}
\usepackage{mathpazo}
\usepackage{color}
\usepackage[colorlinks,
            linkcolor=blue,
            anchorcolor=blue,
            citecolor=blue
            ]{hyperref}
\usepackage{multirow,booktabs,graphicx,appendix}
\usepackage{geometry}
\usepackage{subfigure}
\usepackage{threeparttable}
\usepackage{natbib}

\usepackage{enumerate}
\usepackage{array}

\newtheorem{Theorem}{Theorem}[section]

\newtheorem{Proposition}{Proposition}[section]
\newtheorem{Lemma}{Lemma}[section]
\newtheorem{Remark}{Remark}[section]
\newtheorem{Corollary}{Corollary}[section]
\newtheorem{Example}{Example}[section]
\newenvironment{proofoftheorem}[1]{\noindent{\bf Proof of Theorem #1. }}{ \qed }
\newenvironment{proofofproposition}[1]{\noindent{\bf Proof of Proposition #1. }}{ \qed }
\newenvironment{proofofcorollary}[1]{\noindent{\bf Proof of Corollary #1. }}{ \qed }

\newcommand{\tabincell}[2]{\begin{tabular}{@{}#1@{}}#2\end{tabular}}

\begin{document}

\thispagestyle{empty}
\pagestyle{plain}
\title{Common Decomposition of Correlated Brownian Motions and its Financial Applications}
\author{Tianyao Chen, Xue Cheng, Jingping Yang}
\maketitle

\begin{abstract}
\noindent
In this paper, we develop a theory of \emph{common decomposition} for two correlated Brownian motions, in which, by using change of time method, the correlated Brownian motions are represented by a triplet of processes, $(X,Y,T)$, where $X$ and $Y$ are independent Brownian motions. We show the equivalent conditions for the triplet being independent. We discuss the connection and difference of the common decomposition with the local correlation model. Indicated by the discussion, we propose a new method for constructing correlated Brownian motions which performs very well in simulation. For applications, we use these very general results for pricing two-factor financial derivatives whose payoffs rely very much on the correlations of underlyings. And in addition, with the help of numerical method, we also make a discussion of the pricing deviation when substituting a constant correlation model for a general one.
\end{abstract}
\section{Introduction}\label{introduction}

The correlation between assets plays an important role in finance. Whenever we meet a problem involving two stochastic factors, the correlation risk is unavoidable. The problem may be from areas of asset allocation, pairs trading, risk management and typically, multi-assets derivative's pricing. In financial derivatives' pricing, there are quite a lot chances to meet with the situation of handling two stochastic factors. For example, in stochastic volatility models, the risky price and the stochastic volatility are two factors; in cross-currency derivatives, the evolution of two currencies are driven by different stochastic factors; in two-asset or multi-asset derivatives, the price movements may be modeled by two stochastic processes, etc. Generally speaking, 
 there are two methods in financial modelling to induce dependence between assets, one is by \emph{copula}, the other is in SDE models by assuming a correlation structure for processes driving the model. Hence, modeling the stochastic factors by two Brownian motions has been a common-used method, see, among others, \cite{heston1993closed},\cite{dai2004lookback} and \cite{hurd2010fourier}. In most situations, from a practical aspect, the two stochastic factors (hence the two Brownian motions) should be correlated with each other. Since Brownian motion is the most commonly used driving process stemming from Bachelier, correlation between Brownian motions is crucially important in the latter.

To formulate correlated Brownian motions, many models adopt the constant local correlation assumption, i.e., $d[B,W]_t=\rho dt$ or 
conventionally, $dB_tdW_t=\rho dt$, for Brownian motions $B$ and $W$ and a constant $\rho\in \mathbb R$. However, more and more empirical works proved that the dependence between financial factors varies over time and depending on the economic status, e.g., \cite{bahmani2015relation} for cross-currency derivatives, \cite{engle2001theoretical} for multi-asset and \cite{benhamou2010time} for stochastic volatility models. Other empirical evidences are as follows, \cite{Chiang2007Dynamic} found a significant increasing for correlations between Asian market after the crisis, \cite{syllignakis2011dynamic} and \cite{junior2012correlation} getting similar results for the European and global markets, \cite{xiong2018time} discovered time-varying correlation between policy index and stock return in China and \cite{balcilar2018dynamic} found dynamic correlation between oil price and inflation in South Africa.

Probably for this reason, there is a growing literature in recent years applying dynamic local correlation for financial problems. Since the value of local correlation, i.e., $\rho$ introduced above, must be in $[-1,1]$, these literatures adopted various techniques to assure this. \cite{osajima2007asymptotic} and \cite{fernandez2013static} modeled $\rho$ as a bounded deterministic function of time $t$  for SABR model while \cite{teng2015pricing} adopted the same idea in geometric Brownian motion model and applied it to pricing Quanto option. Note that in these models, $\rho$ is dynamic but nonstochastic. For stochastic $\rho$, \cite{van2006modelling}, \cite{langnau2010dynamic}, \cite{teng2016versatile} and \cite{carr2017bounded} expressed $\rho$ as a bounded function of some stochastic state processes and applied it in derivatives' pricing problems. And some literatures modeled $\rho$ directly by a bounded stochastic process. For example, bounded Jacobi process is a kind of bounded diffusion process driven by Brownian motion and was introduced to model $\rho$ with applications in option pricing and assets management, including vanilla option \citep{teng2016heston}, correlation swap \citep{meissner2016correlation}, Quanto \citep{Ma2009Pricing} and multi-asset option \citep{ma2009stochastic}, and in portfolio selection and risk management \citep{buraschi2010correlation}. \cite{hull2010valuation} modeled the local correlation as a step process where each step is a beta-distributional random variable. \cite{markus2019comparison} made a comparison of several stochastic local correlation models. Moreover, regime switching model is a well used model in finance where all the parameters, including $\rho$, could be driven by a common continuous-time finite-state stationary Markov process, and thus provide another way to model stochastic local correlation, e.g. \cite{zhou2003markowitz}. Wishart process can establish stochastic covariance directly, and the local correlation obtained from covariance matrix is stochastic as well. \cite{da2007option} discussed the Wishart process for multi-asset option pricing and found that there is a correlation leverage affect in call on max style option. Double Heston model also allows a special kind of local correlation between asset and stochastic volatility, see \cite{costabile2012pricing} and \cite{christoffersen2009shape} for more details.

Except correlated Brownian motions, there are also other ways to construct correlated stochastic processes. \cite{wang2009multivariate} obtained correlated variance gamma processes by Brownian motions with constant correlation compound with time changes. \cite{mendoza2016multivariate} and \cite{barndorff2001multivariate} describe correlated stochastic processes by independent background stochastic processes with dependent L\'evy subordinators. \cite{ballotta2016multivariate} proposed factor model for L\'evy process, each asset is governed by a systematic component and a specific component.

The main focus of this paper is on proposing a new method which we call \emph{Common Decomposition} for formulation and analysis of the dependency structure for general correlated Brownian motions. By introducing a time change process, the two correlated Brownian motions can be decomposed as two independent Brownian motions, where the two independent Brownian motions characterize the common and counter movements of the original two correlated Brownian motions. Hence, the key point of dependency structure of two original Brownian motions is the time change process. Comparing with the local correlation, an important advantage of common decomposition is that time change process is observable while the local correlation is usually unobservable. Time change is a developed technique to construct stochastic processes \citep{barndorff2015change}, and is widely applied to mathematical finance \citep{carr2003stochastic,geman2001time}. However, as far as we know, there are few works apply time change technique into modeling correlated Brownian motions. An interesting thing is that we find common decomposition is invariance after change of measure under proper conditions.

Conversely, we also consider how to construct correlated Brownian motions by common decomposition. Comparing with the Euler-Maruyama method \citep{kloeden2013numerical} of \emph{Local Correlation} model, we find that common decomposition method simulate the correlated Brownian motions much faster. Under some conditions, there is no simulation error in common decomposition method which is impossible for Euler-Maruyama method of local correlation model.

After construct correlated Brownian motions, we apply our method into financial derivatives pricing, such as Quantos, covariance and correlation swap, 2-assets option, etc. For 2-assets option, it is hard to obtain closed form directly, hence we provide a analytical solution based on Fourier transform. Fourier transform method in option pricing is proposed by \cite{carr1999option}, more recent papers studied Fourier transform method to price multi-asset options, e.g. \cite{hurd2010fourier} for spread option, \cite{wang2009multivariate} for rainbow options and \cite{leentvaar2008multi} gave a numerical method for multi-asset options without explicit expression. Through Fourier method, we find a unified analytical tractable expression of prices of 2-assets options.

We investigate the pricing error between constant correlation model and stochastic correlation model for 2-assets option by numerical experiments. The numerical results shows that for most out-of-the-money options, the constant correlation model perform poorly while the constant correlation model perform well for at-the-money and in-the-money options.

This paper is organized as follows. In Section \ref{Dependency Structure of Two Correlated Brownian Motions}, we give the definition of common decomposition and discuss the independency properties of stochastic processes obtained from common decomposition. Besides, we consider the relationship between common decomposition method and local correlation model. In Section \ref{Construction of Two Correlated Brownian Motions}, we provide a sufficient condition for constructing correlated Brownian motions and compare the simulation efficiency between common decomposition and traditional method. Financial applications for derivatives pricing are given in Section \ref{Pricing Financial Derivatives by Decomposition of Two Correlated Brownian Motions}. Numerical results are shown in Section \ref{Numerical Results}. Proofs of this paper are given in Section \ref{Some Proofs}.

\section{Common Decomposition of Two Correlated Brownian Motions}
\label{Dependency Structure of Two Correlated Brownian Motions}
In this section, we consider the new method which is called the \emph{common decomposition} of two correlated Brownian motions. Firstly, we propose the definition of common decomposition of two correlated Brownian motions and give some notations. Secondly, we investigate the distribution and independency property of stochastic processes obtained from the common decomposition. Finally, we study the connection of the common decomposition and \emph{local correlation} of two correlated Brownian motions.

In the financial market, if the time interval of observing asset price tends to $0$, then the realized variance of observed asset price tends to the quadratic variation of asset price. Note that the quadratic variation $[\cdot,\cdot]$ of continuous local martingale is same as the predictable quadratic variation $\langle\cdot,\cdot\rangle$ \citep[Chapter IV, Theorem 1.8]{revuz2013continuous}, and the stochastic process involved in this paper are all continuous local martingales, hence we replace $\langle\cdot,\cdot\rangle$ with $[\cdot,\cdot]$ if there is no confusion.
\subsection{Definition of Common Decomposition}\label{model setup}

On a complete probability space $\left(\Omega, \mathcal F, P\right)$, we consider two correlated Brownian motions, $\{B_t\}_{t\ge0}$ and $\{W_t\}_{t\ge0}$, with respect to the same filtration $\mathbb F=\{\mathcal F_t\}_{t\geq0}$ which is assumed to satisfy the usual conditions.

Define
\begin{equation}T_t\triangleq\frac{t+[B,W]_t}2,\ S_t\triangleq\frac{t-[B,W]_t}2,\label{definition of T}\end{equation}
where $[B,W]_t$ denotes the cross variation of $B$ and $W$. Note that $B$ and $W$ are Brownian motions, hence $[B,B]_t=[W,W]_t=t$. Consequently $[\frac{B+W}2,\frac{B+W}2]_t=\frac14([B,B]_t+[W,W]_t+2[B,W]_t)=T_t$, which implies $T$ is quadratic variation of $\frac{B+W}2$. Similarly, $S$ is quadratic variation of $\frac{B-W}2$.
By immediate calculation, when $s<t$ we have
\begin{align}\nonumber-t+s=&\frac{-[B,B]_t-[W,W]_t+[B,B]_s+[W,W]_s}{2}\\
\nonumber\le&[B,W]_t-[B,W]_s\\
\le&\frac{[B,B]_t+[W,W]_t-[B,B]_s-[W,W]_s}2=t-s,\label{[B,W]_t-[B,W]_s}\end{align}
hence
\begin{equation}0\le T_t-T_s\le t-s,\ 0\le S_t-S_s\le t-s.\label{T_t-T_s}\end{equation}
Consequently, $T_t$ and $S_t$ are increasing processes with $T_t+S_t=t$ and thus they are both absolutely continuous with respect to $t$. Then by Radon-Nikodym theorem, $T_t$ and $S_t$ are derivable with respect to $t$.
\begin{Example}If the correlation coefficient $\rho$ of $B$ and $W$ is constant, i.e., $[B,W]_t=\rho t$, then $T_t=\frac{1+\rho}2t$ and $S_t=\frac{1-\rho}2t$. Particularly,
\begin{itemize}
\item when $B$ and $W$ are completely positive correlated, then $[B,W]_t=t$, $T_t=t$ and $S_t=0$;
\item when $B$ and $W$ are completely negative correlated, then $[B,W]_t=-t$, $T_t=0$ and $S_t=t$;
\item when $B$ and $W$ are independent with each other, then $T_t=S_t=\frac t2$.
\end{itemize}
\end{Example}

%

We will explain in Section \ref{Comparison with Local Correlation Model} that $T$ and $S$ could be regarded as special ``timers" that records the time with special correlation information. Next, let
\begin{equation}\label{tau,varsigma}
  {\tau_t}=\inf\{u:T_u>t\},\ {\varsigma_t}=\inf\{u:S_u>t\},\forall t\ge0.
\end{equation}
By definition, $\{{\tau_t}\}_{t\ge0}$ and $\{\varsigma_t\}_{t\ge0}$ are time changes\footnote{A time change $C$ is a family $C_s,s\ge0,$ of stopping times such that the map $s\rightarrow C_s$ are a.s. increasing and right-continuous (\cite{revuz2013continuous},Chapter V, Definition 1.2).} of filtration $\mathbb F$, and on the contrary, $T$ is a time change of $\{\mathcal F_{\tau_t}\}_{t\geq 0}$ and $S$ is a time change of $\{\mathcal F_{\varsigma_t}\}_{t\geq 0}$.

When $\tau_t<\infty$ and $\varsigma_t<\infty$ for any $t>0$, the so-called \emph{common decomposition} in this article could be given through time-changed processes. Let
\begin{equation}X_t\triangleq \frac{B_{\tau_t}+W_{\tau_t}}{2},\  Y_t\triangleq\frac{B_{\varsigma_t}-W_{\varsigma_t}}{2}.\label{definition of X and Y}\end{equation}
If $\tau_{T_t}=t$,it is evident that $X_{T_t}=\frac{B_t+W_t}2$ according to \eqref{definition of X and Y}. If $t<\tau_{T_t}<\infty$, for any $u\in[t,\tau_{T_t}]$, we have $T_u=T_t$ by the continuity of $T$ and the definition of $\tau$. Note that $T$ is the quadratic variation of $\frac{B+W}2$, hence $\frac{B_u+W_u}2=\frac{B_t+W_t}2$ for any $u\in[t,\tau_{T_t}]$ according to \citet{revuz2013continuous}. Consequently, $X_{T_t}=\frac{B_{\tau_{T_t}}+W_{\tau_{T_t}}}2=\frac{B_t+W_t}2$. If $\varsigma_{S_t}<\infty$, by the similar approach, we have $Y_{S_t}=\frac{B_t-W_t}2$. In summary, if $\tau_{T_t}<\infty$ and $\varsigma_{S_t}<\infty$, we have
\begin{equation}B_t=X_{T_t}+Y_{S_t},\ W_t=X_{T_t}-Y_{S_t}.\label{decomposition equation}\end{equation}
Thus we obtained a representation of $(B,W)$ through the three new-defined processes $X,\ Y,\ \mbox{and}\ T$ (it always holds that $S_t=t-T_t$). We call \eqref{decomposition equation} the \emph{common decomposition of $(B,W)$} and the triplet of common decomposition is denoted by $(X,Y,T)$. Note that the concept of common decomposition was first proposed by \cite{chen2018decomposing} for the correlated random walks. In this article, we focus on the common decomposition in correlated Brownian motions.

Given $\omega\in\Omega$,
$$T_{\infty}(\omega)\triangleq\lim_{u\to\infty}T_u(\omega),\ S_\infty(\omega)\triangleq\lim_{u\to\infty}S_u(\omega).$$
Whenever $T_{\infty}(\omega)$ is finite, $X_t(\omega)$ is not well-defined for $t\geq T_\infty(\omega)$. For example, if $B$ and $W$ are completely negative correlated, then $[B,W]_t=-t$, $T_t=0$ for any $t\geq 0$, and ${\tau_t}=\infty$. The same happens to $S$ and $Y$. In order to overcome this limitation, we apply the similar method as in  \citet[Chapter V]{revuz2013continuous} to modify the definition of $X$ and $Y$. We assume the probability space $(\Omega,\mathcal F, P)$ are rich enough to support Brownian motions  that are independent of known Brownian motions and $\mathcal F_{\infty}$. 

By \citet[Chapter V, Proposition 1.8]{revuz2013continuous},
$X_{\infty}\triangleq\lim_{t\to\infty}\frac{B_{t}+W_{t}}{2}$ exists on $\{T_{\infty}<\infty\}$; Similarly,$Y_{\infty}\triangleq\lim_{t\to\infty}\frac{B_{t}-W_{t}}{2}$ exists on $\{S_{\infty}<\infty\}$.
Suppose $\{\tilde X_t,\tilde Y_t\}_{t\ge0}$ is a 2-dimensional Brownian motion independent from $\mathcal F_{\infty}$. We modify the definition of $\{X_{t}\}_{t\ge0}$ and $\{Y_{t}\}_{t\ge0}$ as follows:
\begin{equation}
X(t,\tau_t)\triangleq\begin{cases}\frac{B_{\tau_t}+W_{\tau_t}}{2},\quad&\mbox{if }t<T_\infty\\ X_{\infty}+ \tilde X_{t-T_{\infty}},\quad&\mbox{if }t\geq T_\infty\end{cases},\ \ Y(t,\varsigma_t)\triangleq\begin{cases}\frac{B_{\varsigma_t}-W_{\varsigma_t}}{2},\quad&\mbox{if }t<S_\infty\\ Y_{\infty}+\tilde Y_{t-S_{\infty}},\quad&\mbox{if }t\ge S_\infty\end{cases}.\label{definition of X}\end{equation}
In the following, $X(t,\tau_t)$ and $Y(t,\varsigma_t)$ will be abbreviated as $X_t$ and $Y_t$ when there is no confusion. If $T_t<T_\infty$, note that $\{T_t<T_\infty\}=\{\tau_{T_t}<\infty\}$, hence we have  $X_{T_t}=\frac{B_{\tau_{T_t}}+W_{\tau_{T_t}}}2=\frac{B_t+W_t}2$ from the previous discussion; if $T_t=T_\infty$, since $T$ is quadratic variation of $\frac{B+W}2$, we have $X_{T_t}=X_{\infty}=\frac{B_t+W_t}2$ according to \eqref{definition of X} and \citet[Chapter IV, Proposition 1.13]{revuz2013continuous}. Because $T_t\le T_\infty$ for any $t\ge0$, we have $X_{T_t}=\frac{B_t+W_t}2$ for any $t\ge0$. With the similar proof, we have $Y_{S_t}=\frac{B_t-W_t}2$ for any $t\ge0$. As a consequence, after modifying the definition of $X$ and $Y$, the common decomposition \eqref{decomposition equation} holds for any $t\ge0$.
%

\begin{Remark}
 The choice of $(\tilde X,\tilde Y)$ can only influence the definition of $(X,Y)$, but has no influence on the decomposition of $B$ and $W$. To be more specific, for $\forall t\geq 0$, if $T_t<T_{\infty}$, by definition, $X_{T_t}$ does not depend on $\tilde{X}$; if $T_t=T_\infty$, then $X_{T_t}=X_{T_\infty}=X_{\infty}$, does not depend on $\tilde X$, either. The same is true for $Y$.
\end{Remark}

In the following, suppose $(X,Y,T)$ and $(X^\prime,Y^\prime,T^\prime)$ both satisfy \eqref{decomposition equation}, then
$$T_t=[X_T,X_T]_t=[\frac{B+W}2,\frac{B+W}2]_t=[X^\prime_{T^\prime},X^\prime_{T^\prime}]_t=T^\prime_t\quad a.s.,$$
which implies $T$ is unique in the sense of almost sure. Thus, by the definition of $\tau$, if $\tau_t<\infty$,
$$X_t=\frac{B_{\tau_t}+W_{\tau_t}}2=X^\prime_t,\quad a.s..$$
Note that $\{\tau_t<\infty\}=\{T_{\infty}>t\}$, which indicates $X$ is unique in the interval $[0,T_{\infty})$. Similarly, $Y$ is unique in $[0,S_{\infty})$. In the common decomposition \eqref{decomposition equation}, $X$ and $Y$ are only related with the values in the time interval $[0,T_{\infty})$ and $[0,S_{\infty})$ respectively, hence the common decomposition is unique.

For convenience, we introduce some notations here:
\begin{itemize}
\item $\mathcal F^X_t$: natural filtration of stochastic process $\{X_t\}_{t\ge0}$.
\item $A\perp B|C$: $A$ and $B$ are conditional independent given $C$.
\end{itemize}
It is remarkable that $\mathcal F^T_t=\sigma(T_u:u\le t)=\sigma(S_u:u\le t)=\mathcal F^S_t$.
\subsection{Main Theories of the Common Decomposition}\label{Main Result}

In the previous section, we introduced the so called common decomposition $(X,Y,T)$ of Brownian motions $B$ and $W$. 
In this part we give some basic properties of the decomposition. Proofs can be found in Section \ref{Some Proofs}.

Our first result illustrates the distribution of $X$, $Y$ and the path property of $T$.
\begin{Theorem}\label{decompose BM}
 Given Brownian motions $\{B_t\}_{t\ge0}$ and $\{W_t\}_{t\ge0}$ with respect to $\mathbb F$ and their common decomposition is denoted as $(X,Y,T)$, the following statements hold.
 \begin{description}
 \item[(\romannumeral1)] $\{X_t\}_{t\ge0}$ is a Brownian motion of the filtration $\{\mathcal F^\tau_t\}_{t\geq 0}$, $\{Y_t\}_{t\ge0}$ is a Brownian motion of the filtration $\{\mathcal F^\varsigma_t\}_{t\geq 0}$, $X$ and $Y$ are independent;
 \item[(\romannumeral2)] $\{[B,W]_t\}_{t\ge0}$ is derivable with respect to $t$, and
 $$T_t=\frac{t+\int_0^t\rho_udu}2,\ S_t=\frac{t-\int_0^t\rho_udu}2,$$
 where
 \begin{equation}\rho_t\triangleq\frac{d[B,W]_t}{dt}.\label{def rho}\end{equation}
 \end{description}
\end{Theorem}
In \eqref{def rho}, $\rho$ is called the \emph{local correlation process} of $B$ and $W$. Further discussion of local correlation and common decomposition can be found in Section \ref{Comparison with Local Correlation Model}.

From Theorem \ref{decompose BM}, the common decomposition represents $B$ (resp. $W$) as the sum (resp. difference) of two time-changed Brownian motions. The dependency structure of $B$ and $W$ is embodied in $T$ as well as in the dependencies between $T$ and the two new-defined Brownian motions $X$ and $Y$. Hence for clarity and convenience, the independency of $X$, $Y$ and $T$ is worth studying. In the following theorem, a sufficient and necessary condition is given for mutual independency of them.

\begin{Theorem}\label{condition c1}
Under the conditions and notations as in Theorem \ref{decompose BM}, the common decomposition triplet $X$, $Y$ and $T$ are mutually independent if and only if:
\begin{description}
\item[(C1)]$\mathcal F^B_{\infty}\perp\mathcal F^T_{\infty}|\mathcal F_t^{B,W}$ and $\mathcal F^W_{\infty}\perp\mathcal F^T_{\infty}|\mathcal F_t^{B,W}$.
\end{description}
\end{Theorem}
As an example to understand the condition, when $B$ and $W$ has a constant correlation say, $\rho$, (C1) is satisfied since $T_t=\frac{1+\rho}2t$ and $\mathcal F^T_{\infty}$ is a trivial $\sigma$-algebra. More general cases will be discussed later.

The above two theorems give a more visual interpretion of the common decomposition. During the two Brownian motions' movings, sometimes they move as if with positive correlation and sometimes quite the contrary. These ``common" or ``opposite" moving times are picked out to form new ``clocks" $T_t$ or $S_t$. And their revolutions are decomposed thereupon according to the new clocks. By Theorem \ref{decompose BM}, under the new clocks, they keep their Brownian-motion features and these features are independent under the two clocks. Thus dependency structures and Brownian features of the original correlated Brownian motions are separated. By Theorem \ref{condition c1}, if they satisfy the condition (C1), their dependency information is only contained in $T$, the decomposition is quite complete and clear. In this case, we can focus on the process $T$ in common decomposition if we want to study the dependency structure of two correlated Brownian motions.

The following proposition gives an equivalent condition of (C1) from another aspect.

\begin{Proposition}\label{girsanov}
Suppose the assumptions in Theorem \ref{decompose BM} hold. Then the condition (C1) is equivalent with the following statement.
\begin{description}
\item[(C2)] Given two processes $\{\phi_t^1\}_{t\ge0}$ and $\{\phi_t^2\}_{t\ge0}$, which are progressively measurable with $\{\mathcal F^T_t\}_{t\geq 0}$ and satisfy
    \begin{equation}E\left[\exp\left(\frac12\int_0^t(\phi_{u}^1)^2dT_u+\frac12\int_0^t(\phi_{u}^2)^2dS_u\right)\right]<\infty,\forall t.\label{phi in condtion c2}\end{equation}
Let
$$D_t^{\phi}\triangleq\exp\left(\int_0^t\phi_u^1dX_{T_u}+\int_0^t
\phi_u^2dY_{S_u}-\frac12\int_0^t(\phi_u^1)^2dT_u-\frac12\int_0^t
(\phi_u^2)^2dS_u\right),\forall t\ge0,$$
then $D^{\phi}$ is a martingale and $\frac{dQ}{dP}|_{\mathcal F_t}=D_t^{\phi}$ defines a probability measure such that
\begin{equation}(X^{\phi}_T,Y^{\phi}_S)_{Q}\overset{d}{=}
(X_{T},Y_{S})_P,\label{identical distribution of X_T}\end{equation}
where $X^{\phi}_{T_t}=X_{T_t}-\int_0^t\phi_u^1dT_u,Y^{\phi}_{S_t}=Y_{S_t}-\int_0^t\phi_u^2dS_u.$
\end{description}
\end{Proposition}

This proposition link the independency of the decomposition triplet with conditions similar to Girsanov theorem. Undoubtedly it may attract our attention to consider its application in financial modelling.
\begin{Example}\label{example of girsanove}
In financial models, the Girsanov transform is typically used to change the drift parts of diffusions that modelling the prices. Consider two drifted Brownian motions,
$$\int_0^t\theta^1_udu+B_t,\quad\int_0^t\theta^2_udu+W_t,$$
where $\theta^i,\ i=1,2$ are bounded, progressively measurable with $\{\mathcal F^T_{t}\}_{t\ge0}$. According to Theorem \ref{decompose BM}, $B$ and $W$ can be decomposed into $(X,Y,T)$. And consequently the two drifted Brownian motions can be represented as
$$\int_0^t\theta^1_udu+X_{T_t}+Y_{S_t},\quad\int_0^t\theta^2_udu+X_{T_t}-Y_{S_t}.$$
Let $\lambda\ \mbox{and}\ \mu$ denote the densities of $T$ and $S$,
$$\lambda_t\triangleq\frac{dT_t}{dt},\quad\mu_t\triangleq\frac{dS_t}{dt},$$
and suppose that $\inf\{\lambda_t(\omega),\mu_t(\omega)|t\ge0,\omega\in\Omega\}>0$. Then $\phi=(\phi^1,\phi^2)$ satisfy \eqref{phi in condtion c2}, where
$$\phi_t^1=\frac{\theta_t^1+\theta_t^2}{2\lambda_t},\ \phi_t^2=\frac{\theta_t^1-\theta_t^2}{2\mu_t}.$$

If $(B,W)$ satisfies the condition (C1), then from Proposition \ref{girsanov}, the two drifted Brownian motions can be transformed to
\begin{equation}\label{definition of tilde B & W}\int_0^t\theta^1_udu+B_t=X^{\phi}_{T_t}+Y^{\phi}_{S_t}:=B^\phi_t,\ \int_0^t\theta^2_udu+W_t=X^{\phi}_{T_t}-Y^{\phi}_{S_t}:= W^\phi_t.\end{equation}
Under the probability $Q$ as defined in Proposition \ref{girsanov}, it is notable that by \eqref{identical distribution of X_T},
\begin{equation}\label{identical distribution}(B,W)_P\overset{d}{=}(B^\phi,W^\phi)_Q,\end{equation}
thus the drift parts vanish after change of probability measure.
Consider the common decomposition of $(B^\phi, W^\phi)$, denoted by $(X^\phi,Y^\phi,T^\phi)$. From \eqref{definition of tilde B & W}, we have
\begin{equation}T_t=T^\phi_t,\quad\forall t\ge0.\label{identical of T}\end{equation}
Moreover, from \eqref{identical distribution of X_T} we have
\begin{equation}\left(\{T_t\}_{t\ge0}\right)_{P}\overset{d}{=}\left(\{T_t\}_{t\ge0}\right)_{Q}.\label{identical distribution of T}\end{equation}
\end{Example}
\begin{Remark}\label{estimate rho from real probability measure}
Equation \eqref{identical of T} and \eqref{identical distribution of T} reveal the invariance property of $T$ under change of measure. From the application point of view, this implies that in financial modelling after change of numeraire, the common decomposition method is still valid. And from empirical view, we can estimate parameters from real probability measure and apply to risk neutral measure directly. For example, \cite{ballotta2016multivariate} bring correlation matrix estimated from observed asset prices (real probability measure) into option pricing model (risk neutral probability measure) directly, and we show the theoretical foundation of such operation. This is quite convenient for derivatives pricing which are lack of public data.

This also shows that we can simplify two correlated Brownian motions with drifts by changing of measure, and keep the dependency structure of original processes.
\end{Remark}

\subsection{Common Decomposition and Local Correlation Model}\label{Comparison with Local Correlation Model}


In this section, we take a new look at the common decomposition via the local correlation process. We consider the difference and connection between the common decomposition method and the local correlation model. As before, the proofs can be found in Section \ref{Some Proofs}.

\subsubsection{Relationship Between Common Decomposition and Local Correlation Model}\label{Different Decomposing Methods for correlated Brownian Motions}

Let us first recall a well used decomposition method representing correlated Brownian motions as linear combinations of independent Brownian motions based on $\rho$. Suppose $\tilde Z$ is a Brownian motion independent of $\mathcal F_{\infty}$ and $(\tilde X,\tilde Y)$, then we define

\begin{equation}\label{definition of z}
Z_t\triangleq\int_0^t\frac{1_{\{\rho_u\neq\pm1\}}}{\sqrt{1-\rho_u^2}}(dW_u-\rho_udB_u)+\int_0^t1_{\{\rho_u=\pm1\}}d\tilde Z_u.\footnote{
The second part of $Z$ is obviously well defined. Consider the first part $\int_0^t\frac{1_{\{\rho_u\neq\pm1\}}}{\sqrt{1-\rho_u^2}}(dW_u-\rho_udB_u)$,
let $M_t\triangleq W_t-\int_0^t\rho_udB_u,$ $[M]_t=t+\int_0^t\rho_u^2du-2\int_0^t\rho_ud[B,W]_u
=\int_0^t(1-\rho_u^2)du.$
For any $0<t<\infty$,
$\int_0^t\frac{1_{\{\rho_u\neq\pm1\}}}{1-\rho_u^2}d[M]_u\le\int_0^tdu=t<\infty,$
which implies the first part is well defined, too.}\end{equation}
Particularly, if $\forall t$, $\rho_t\neq\pm1,$ a.s., then
$$Z_t=\int_0^t\frac1{\sqrt{1-\rho_u^2}}dW_u-
\int_0^t\frac{\rho_u}{\sqrt{1-\rho_u^2}}dB_u.$$
It is not difficult to verify that $[B,Z]_t=0,\forall t\ge0$, hence $\{Z_t\}_{t\ge0}$ is a Brownian motion independent of $\{B_t\}_{t\ge0}$, and that the local correlation of $Z$ and $W$ is $\sqrt{1-\rho_t^2}$.

By definition of $Z_t$, we have the local-correlation based decomposition of $(B,W)$,
\begin{equation}\label{local-corr}
 (B_t,W_t)=(B_t,\int_0^t\rho_sdB_s+\int_0^t\sqrt{1-\rho_s^2}dZ_s).
\end{equation}
If we start from the right side of the equation, i.e., starting from independent Brownian motions $B,\, Z$ and local correlation process $\rho$, we have got a commonly used model for constructing correlated Brownian motions $(B,W)$.

As a comparison, by the common decomposition in the current paper, $(B,W)$ has the representation
$$(B_t,W_t)=(X_{T_t}+Y_{S_t},X_{T_t}-Y_{S_t}).$$
Similarly, if we start from the right side, i.e., from independent Brownian motions $X,\,Y$ and time-change process $T$, and make the construction, then $(B,W)$ are correlated Brownian motions under some conditions. Following the procedure, we can get a new construction method of $(B, W)$. We will make further discussions of this new construction method of correlated Brownian motions in Section \ref{Construction of Two Correlated Brownian Motions}.

\begin{Remark}
The different ideas behind the two methods look clear from the above comparison: the local-correlation method characterize dependency of the Brownian motions from a spatial perspective while the common-decomposition method from a temporal perspective. And $\rho_t$ characterized the correlation between $B$ and $W$ at time $t$, but $T_t$ characterized the correlation in the time period $[0,t]$. Namely, $\rho_t$ represent the correlation locally, but $T_t$ characterize the correlation in the whole time period $[0,t]$.
\end{Remark}

The next proposition gives a connection between local-correlation based decomposition and common decomposition. The two method would share the same equivalent conditions when considering completely-independent decomposition.
\begin{Proposition}\label{independent local correlation imply independent time change}
Under the conditions stated in Theorem \ref{decompose BM}, $X$, $Y$ and $T$ are mutually independent if and only if 
the following condition holds:
\begin{description}
\item[(C3)] $\rho$, $B$ and $Z$ in local-correlation model (\ref{local-corr}) are mutually independent.
\end{description}
\end{Proposition}

\begin{Remark}
Suppose $\{M_t\}_{t\ge0},\{N_t\}_{t\ge0}$ are two continuous local martingales with respect to $\mathcal F_t$ and $[M,M]_t=[N,N]_t,\forall t$, then Theorem \ref{decompose BM} 
can be generalized directly, where
$$T_t=\frac{[M,M]_t+[M,N]_t}2,S_t=\frac{[M,M]_t-[M,N]_t}2,$$
and $X_t,Y_t$ are defined similarly with Section \ref{model setup}. Theorem \ref{condition c1} and Proposition \ref{girsanov} remains valid if we replace $\mathcal F^T$ by $\mathcal F^{T,S}$ in condition the (C1) and (C2)\footnote{In Brownian motion case, $\mathcal F^T=\mathcal F^{T,S}$}.
Moreover, if $[M,M]_t=[N,N]_t$ is absolute continuous with respect to $t$, then according to martingale representation theorem, we can rewrite $(M_t,N_t)$ as
$$(M_t,N_t)=(\int_0^t\theta_udB_u,\int_0^t\xi_udB_u+\int_0^t\eta_udZ_u),$$
where $B$ and $Z$ are two independent Brownian motions and $\theta_u\ge0,\eta_u\ge0,\forall u$. It is evident that $\mathcal F^{T,S}=\mathcal F^{\theta,\xi,\eta}$. Hence, Proposition \ref{independent local correlation imply independent time change} is still correct if $\rho$ is replaced with $\mathcal F_{\infty}^{T,S}$ in the condition (C3).

Particularly, the equivalent condition of $\{X_t\}_{t\ge0}$, $\{Y_t\}_{t\ge0}$ and $\{T_t\}_{t\ge0}$ are mutually independent in 1-dimension situation, i.e. Ocone martingale, illustrated in \cite{kallsen2006didactic} and \cite{vostrikova2000some} is a special case of $M_t=N_t$. Ocone martingale has been widely used in financial mathematics, such as \cite{carr2005pricing} and \cite{geman2001asset}.
\end{Remark}

\subsubsection{Further Discussions for $T$ and $\rho$}

From the setup, we can see $T$ play an important role in the common decomposition. Since $X$ and $Y$ are independent, $T$ is relevant to the dependency structure of $(B,W)$ in the common decomposition triplet. Particularly, in the case of complete decomposition where $X$, $Y$ and $T$ are independent, $T$ contains all the dependency information. On the other hand, if we treat $T$ as a special timer, a "clock", it is obvious that this clock's movements are affected by the correlation of $(B,W)$. In this section, we make further discussions of $T$ via $\rho$ to get a better understanding of the common decomposition.

First, by Theorem \ref{decompose BM}, $(T,S)$ and $\rho$ are connected as follows: $$T_t=\int_0^t\frac{1+\rho_u}2du,\
S_t=\int_0^t\frac{1-\rho_u}2du,$$
in which, $1+\rho_t$ is in fact the distance between local correlation $\rho_t$ and $-1$, $1-\rho_t$ is the distance between $\rho_t$ and $+1$, and the denominator $2$ is the distance between $-1$ and $+1$. Thus the integrands could be regarded as normalizations of the deviation of $(B,W)$'s correlation from complete correlation. For instance, Figure \ref{example of rho} shows a path of $\rho$, in which the shadow part represents $S$ and the light part represents $T$. Think of the case when $\rho$ is always close to $1$ and far away from $-1$, then the "clock $T$" runs faster than $S$, and the "clock T" focuses on positive correlation.

\begin{figure}[htbp!]
  \centering
  \includegraphics[width=15cm]{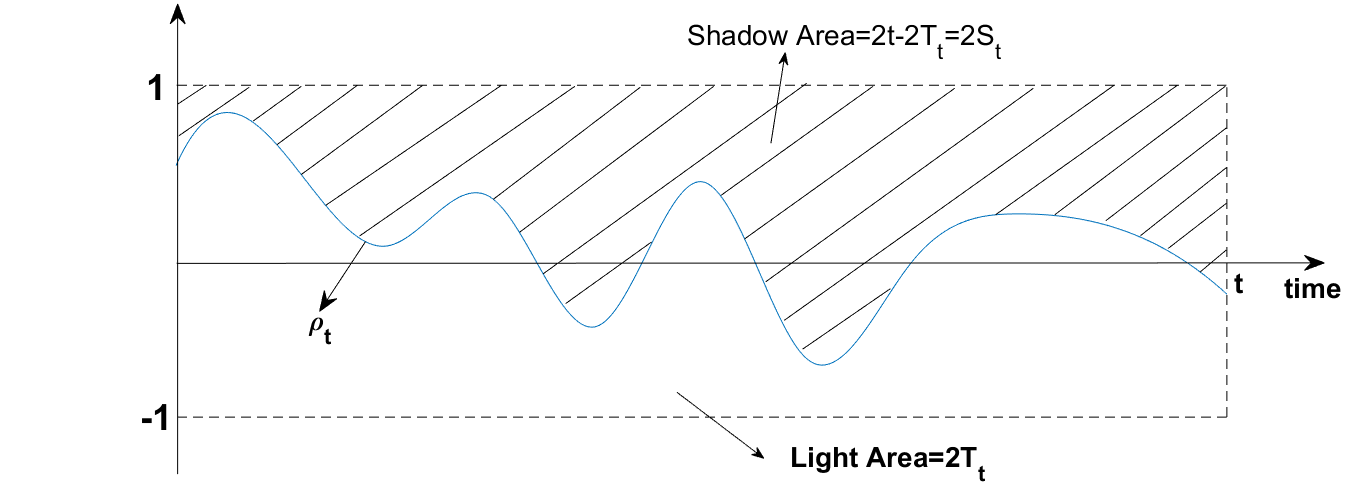}
  \caption{A Path of Local Correlation $\rho$}\label{example of rho}
\end{figure}


Consider the values of $T$ and $S$, at any time $t$, they satisfy
$$T_t+S_t=t,\ T_t-S_t=\int_0^t\rho_udu.$$
That is to say, the sum of the readings of two clocks represents the calender time, while the difference of them shows the cumulated correlation of $(B,W)$ till time $t$.

The average correlation coefficient process is defined as
\begin{equation}\label{average correlation}
 \bar{\rho}_t=\frac1t\int_0^{t}\rho_udu,
\end{equation}
and it could also be represented by $T$ and $S$,
$$\bar{\rho}_t=\frac{T_t-S_t}t=\frac{T_t-S_t}{T_t+S_t}.$$

Another main difference between $T$ and $\rho$ is observability. $T$ is always observable through quadratic covariation while $\rho$ is usually unobservable. In statistics, we can only estimate the correlation coefficient for a period of time, that is to say, the estimation of $\rho$ in statistics is actually $\bar\rho$ but not $\rho$ itself. Hence, if the local correlation is dynamic, statistics can help us to study $T$ well.


Consider the two-factor derivative's pricing in finance. When local correlation of the two factors varies
stochastically over time, it is always difficult to obtain the option prices. The average correlation coefficient process, $\bar\rho$, usually plays an important role under this circumstances. For example, the price of foreign equity option was approximated by the moments of $\bar\rho$ in \cite{Ma2009Pricing}, and \cite{van2006modelling} and \cite{teng2016versatile} show that the price of a Quanto is determined by the Laplace transform of $\bar\rho$. In our method, $\bar\rho_t=\frac{2T_t}{t}-1$, this is one of the reasons indicating the advantage of using common-decomposition method in financial modelling. We will discuss this further in Section \ref{Pricing Financial Derivatives by Decomposition of Two Correlated Brownian Motions}.

In the next part we use a simple example to reveal the concepts mentioned above.
\begin{Example}\label{constant correlation}
Suppose $B$ and $W$ are two Brownian motions with constant correlation $\rho\in(-1,1)$. Then by the local-correlation method,
$$(B_t,W_t)=(B_t,\rho B_t+\sqrt{1-\rho^2}Z_t),$$
where $Z$ has been defined in \eqref{definition of z}. In this case, the condition in Proposition \ref{independent local correlation imply independent time change} is satisfied, thus the processes of the common decomposition triplet, $X,Y$ and $T$, are mutually independent. And they can be calculated accurately,
$$T_t=\frac{1+\rho}2t,\ S_t=\frac{1-\rho}2t,$$
$$X_t=\frac12B_{\frac2{1+\rho}t}+\frac12W_{\frac2{1+\rho}t},\ 
Y_t=\frac12B_{\frac2{1-\rho}t}-\frac12W_{\frac2{1-\rho}t}
,$$
and the decomposition of $(B,W)$ is
$$(B_t,W_t)=(X_{\frac{1+\rho}2t}+Y_{\frac{1-\rho}2t},X_{\frac{1+\rho}2t}-Y_{\frac{1-\rho}2t}).$$

In this example, we summarize three statements as follows.
\begin{description}
  \item[(\romannumeral1)] $T$ and $S$ conform a decomposition of the ``calender time" in any time period. They are composed by special ``time points" picked out according to the correlation structure of $(B,W)$. They can be considered as special
      clocks that moves only at special time.
  \item[(\romannumeral2)] If $\rho>0$, the clock $T$ runs faster than the clock $S$, vice versa.
  \item[(\romannumeral3)] Consider $\mathcal C=\{\alpha B+(1-\alpha)W|\alpha\in\mathbb R\}$, the family of generalized convex combinations of $B$ and $W$. The correlation coefficient of every two processes in $\mathcal C$ with parameters $\alpha$ and $\beta$ is $$\rho_{\alpha,\beta}=(1-\rho)[(2\alpha-1)\beta-\alpha]+1.$$ If $\alpha=\frac{1}{2}$, $\rho_{\alpha,\beta}=\frac{\rho+1}{2}>0, \forall \beta\in\mathbb R$. Otherwise, we have $\rho_{\alpha,\beta}\leq 0$ if $\alpha>\frac12$, $\beta\le(\alpha-\frac1{1-\rho})/(2\alpha-1)$ or $\alpha<\frac12$, $\beta\ge(\alpha-\frac1{1-\rho})/(2\alpha-1)$. In other words, $\frac{B+W}{2}$ is the only process in $\mathcal C$ that is strictly positive correlated with any other process in $\mathcal C$. Note that this process is in fact $X$ under clock $T$, thus $X$ represents the common structures in $B$ and $W$. Similarly, $Y$ represents the common structures in $B$ and $-W$. Namely, $X$ and $Y$ are two extreme cases, and they are taken from $B$ and $W$ by the common decomposition. In fact, the background of the set $\mathcal C$ is from a financial example. Suppose $B$ and $W$ represent returns of two assets, then $\alpha B+(1-\alpha) W$ represents the return of portfolio on these two assets. And $\alpha<0$ or $\alpha>1$ represents the short selling of assets.
\end{description}
\end{Example}
\begin{Remark}
Actually, if the local correlation of $B$ and $W$ is not constant, the three statements for Example \ref{constant correlation} remain valid. For \textbf{(\romannumeral1)} and \textbf{(\romannumeral2)}, the results remain the same. For \textbf{(\romannumeral3)}, we can prove
$$\frac1tCov\left(\alpha B_t+(1-\alpha)W_t,\beta B_t+(1-\beta)W_t\right)=(1-Corr(B_t,W_t))[(2\alpha-1)\beta-\alpha]+1,$$
where $Cov$ and $Corr$ denote covariance and correlation respectively. With the similar discussion, $\frac{B+W}{2}$ is the only process in $\mathcal C$ that is strictly positive correlated with any convex combination of $B$ and $W$.
\end{Remark}

\subsection{Illustration of the Common Decomposition via Discretization}

The example in previous section demonstrated what the processes
in common decomposition look like and how to construct the clock $T$ when $\rho_t\equiv \rho\in (0,1)$.
In this section, similar analysis is carried out from a distributional
 aspect for general cases by discretizing $\rho$. In this part, we also start with two correlated Brownian motions $B$ and $W$ with local correlation process $\rho$, and all the other notations defined in previous sections are followed.

Given $t\geq 0$, let $\Pi$ be a partition
of $[0,t]$:$$0=t_0<t_1<t_2<\dots<t_n=t,$$
and write
$||\Pi||=\max\{t_i-t_{i-1}:i=1,\dots,n\}$. Given $\rho$, for $\forall \omega\in \Omega$, define\footnote{The choice of $A^\Pi(\omega)$ is not unique. $A^\Pi(\omega)$ can be any Borel set as long as $m\left(A^\Pi(\omega)\cap(t_i,t_{i+1}]\right)=\frac{1+\rho_{t_i}(\omega)}2\Delta t_i,\forall i$, where $m(\cdot)$ denotes the Lebesgue measure on $\mathbb R$.}
$$A^\Pi(\omega)=\bigcup_{i=0}^{n-1}
(t_i,\, t_i+\frac{1+\rho_{t_i}(\omega)}2\Delta t_i].$$
Note that by the construction of $A^\Pi$, the stochastic processes $\{1_{\{u\in A^\Pi\}}\}_{0\leq u\leq t}$ and $\{1_{\{u\notin A^\Pi\}}\}_{0\leq u\leq t}$ are predictable.

Set
$$\tilde X^{\Pi}_{s}\triangleq
\int_{0}^{s}1_{\{u\in A^\Pi\}}dB_u,\ \ \tilde Y^{\Pi}_{s}\triangleq\int_{0}^{s}1_{\{u\notin A^\Pi\}}dB_u,\quad\forall s\in[0,t],$$
i.e., $\tilde X^{\Pi}$ keeps in step with $B$ in $A^\Pi$ and stays stationary at other time while $\tilde Y^\Pi$ performs oppositely. Let
\begin{equation}\tilde W^\Pi_s\triangleq\tilde X^{\Pi}_{s}-\tilde Y^{\Pi}_{s}=
\int_{0}^{s}1_{\{u\in A^\Pi\}}dB_u-\int_{0}^{s}1_{\{u\notin A^\Pi\}}dB_u.\label{tilde W}\end{equation}
Then $\tilde W^\Pi$ is a Brownian motion moving commonly with $B$ in $A^\Pi$
 and oppositely in $[0,t]\setminus A^\Pi$. And $\tilde X^\Pi$ and $\tilde Y^\Pi$ represent the common movements and counter movements of $B$ and $\tilde W^\Pi$.

At any time $s\leq t$, the time period $[0,s]$ is divided into two parts: the commonly-moving period $A^\Pi\bigcap [0,s]$ and the oppositely-moving period $[0,s]\setminus A^\Pi$, whose total lengths could be calculated respectively as (suppose $t_i< s\leq t_{i+1}$)
 $$\tilde T_s^{\Pi}(\omega)\triangleq m\left([0,s]\cap A^\Pi(\omega)\right)=\frac{t_i+\sum_{k=0}^{i}\rho_{t_k}(\omega)\Delta t_k}{2}+m\left((t_i,s]\cap A^\Pi(\omega)\right),$$
$$\tilde S_s^{\Pi}(\omega)\triangleq m\left([0,s]\setminus A^\Pi(\omega)\right)=\frac{t_i-\sum_{k=0}^{i}\rho_{t_k}(\omega)\Delta t_k}{2}+m\left((t_i,s]\setminus A^\Pi(\omega)\right),$$
where $m(\cdot)$ denotes the Lebesgue measure on $\mathbb R$.
Obviously,
\begin{equation}\label{T of local correlation model}
\lim_{||\Pi||\to0}\tilde T_s^{\Pi}=\frac{s+\int_0^s\rho_udu}2
=T_s,\ \lim_{||\Pi||\to0}\tilde S_s^{\Pi}=\frac{s-\int_0^s\rho_udu}2=S_s,\forall s\in[0,t].
\end{equation}
The following proposition considers the limitation property of $\tilde W^\Pi$ in distribution.

\begin{Proposition}\label{convergence in distribution}Suppose the assumptions in model setup and the conditions in Proposition \ref{independent local correlation imply independent time change} hold. For any given $0\le u_1<u_2<\dots<u_K<\infty, 0\le v_1<v_2<\dots<v_L<\infty$, as $||\Pi||\to0$, we have
$$(B_{u_1},B_{u_2},\dots,B_{u_K},\tilde W^\Pi_{v_1},\tilde W^\Pi_{v_2},\dots,\tilde W^\Pi_{v_L})\xrightarrow[\quad]{d}(B_{u_1},B_{u_2},\dots,B_{u_K},W_{v_1},W_{v_2},\dots,W_{v_L}).$$
\end{Proposition}
Proposition \ref{convergence in distribution} guarantees that as $||\Pi||\to0$, any finite dimensional distribution of $(B,\tilde W^\Pi)$ converges to $(B,W)$ in the sense of distribution. For simplicity, we denote this finite dimensional distribution convergence of processes by "$\xrightarrow[\quad]{f.d.d.}$". Thus,
$$(B,\tilde W^\Pi)\xrightarrow[\quad]{f.d.d.}(B,W),$$
as a consequence,
 \begin{equation}(\tilde X^{\Pi},\tilde Y^{\Pi})=(\frac{B+\tilde W^\Pi}2,\frac{B-\tilde W^\Pi}2)\xrightarrow[||\Pi||\to0]{f.d.d.}(X_T,Y_S).\label{tilde X tilde Y convergence}
\end{equation}
 The convergence properties \eqref{tilde X tilde Y convergence} and
\eqref{T of local correlation model} reveal the connections of $X$ and $Y$ with common and counter movements of $(B,W)$ in some sense,  and give an intuitive  explanation for $T$ and $S$ to be considered as  clocks recording positive correlation and negative correlation of $(B,W)$.

\section{Construction and Simulation of Correlated Brownian Motions Based on the Common Decomposition}
\label{Construction of Two Correlated Brownian Motions}

In the previous section, the common decomposition of two correlated Brownian motions has been demonstrated. For any two Brownian motions $B$ and $W$, we can find a triplet $(X,Y,T)$ to represent them by change of time method. While in practice, a converse problem may also be worth concerning and studying. That is, is it possible to construct two Brownian motions with desired dependency structure from two independent Brownian motions by common decomposition method? In this section we will focus on this problem. Furthermore, the simulation method based on the common decomposition is also given.

\subsection{A New Method for Construction of Correlated Brownian Motions}


In this section, we construct correlated Brownian motions by common decomposition method under some conditions and give an example to show the application of this new construction method.


 \begin{Theorem}\label{combine BM}
Let $(X,Y)$ be a 2-dimensional standard Brownian motion and $\{T_t\}_{t\ge0}$, $\{S_t\}_{t\ge0}$ be time changes with respect to $\mathbb F$.
If $\mathcal{F}_t^{Y_S}\perp\mathcal{F}^{X_T}_{\infty}|\mathcal{F}_t^{X_T}$ and $\mathcal{F}_t^{X_T}\perp\mathcal{F}^{Y_S}_{\infty}|\mathcal{F}_t^{Y_S}$, then $\{X_{T_t}\}_{t\ge0}$ and $\{Y_{S_t}\}_{t\ge0}$ are martingales with respect to $\mathcal F^{X_T,Y_S}$. Furthermore, if $T$, $S$ are strictly increasing and $T_t+S_t=t,\forall t\ge 0$, then
$$B_t\triangleq X_{T_t}+Y_{S_t}\ \mbox{and}\ W_t\triangleq X_{T_t}-Y_{S_t}$$
are two correlated Brownian motions with respect to $\mathcal F^{B,W}$ and $[B,W]_t=T_t-S_t$.

\end{Theorem}



%

Immediately, we have a convenient way to construct correlated Brownian motions from Theorem \ref{combine BM}.

\begin{Corollary}\label{XYT independent}Suppose that $T,\ S$ are strictly increasing processes satisfying $T_t+S_t=t,\forall t\geq 0$, and $X,\ Y$ are independent Brownian motions. If $X,\ Y,\ T$ are mutually independent,
then
$$B_t\triangleq X_{T_t}+Y_{S_t}\ \mbox{and}\ W_t\triangleq X_{T_t}-Y_{S_t}$$
are two correlated Brownian motions with respect to $\mathcal F^{B,W}$ and $[B,W]_t=T_t-S_t$.
\end{Corollary}

In the following, we consider constructing correlated Brownian motions through common decomposition and regime switching model.Regime switching is a commonly used model in finance, and it fits financial data well. For example, \cite{schaller1997regime} found very strong evidence for state-dependent switching behaviour in stock market returns. Regime switching model for correlations in discrete time have been considered, e.g. \cite{casarin2018bayesian} and \cite{pelletier2006regime}. Hence, we consider regime switching model to construct correlated Brownian motions by common decomposition method in the next example.
\begin{Example}(Regime switching model)\label{Q process}
Suppose $\{Q_t\}_{t\ge0}$ is a continuous time stationary Markov process taking values in a finite state space $\{e_1,e_2,\dots,e_n\}$, where $$e_i=(\underbrace{0,\dots,0}_{i-1},1,\underbrace{0,\dots,0}_{n-i})$$ denotes the unit vector. The Markov process $\{Q_t\}_{t\ge0}$ has a stationary transition probability matrix $\boldsymbol P(t)=(p_{ij}(t))_{n\times n}$, where
$$p_{ij}(t)=P(Q_{t+s}=e_j|Q_s=e_i).$$
The homogeneous generator $\boldsymbol A=(a_{ij})_{n\times n}$ exists and is defined as
$$\boldsymbol A\triangleq\lim_{t\downarrow0}\frac{\boldsymbol{P}(t)-\boldsymbol I}t,$$
where $\boldsymbol I$ denotes the identity matrix. Then we have
$$\frac{d\boldsymbol P(t)}{dt}=\boldsymbol A\boldsymbol P(t)=\boldsymbol P(t)\boldsymbol A.$$
Solving this ODE we obtain
\begin{equation}\boldsymbol P(t)=e^{\boldsymbol At}.\label{P(t)}\end{equation}
Let $\boldsymbol{\alpha}=[\alpha_1,\alpha_2,\dots,\alpha_n]^T$, $\alpha_i\in(0,1),\forall i$ and
$$T_t=\int_0^t\boldsymbol\alpha^TQ_sds,S_t=t-T_t=\int_0^t(\boldsymbol{1}-\boldsymbol{\alpha})^TQ_sds.$$
Obviously, $\{T_t\}_{t\ge0},\{S_t\}_{t\ge0}$ are increasing processes. Let $\{X_t\}_{t\ge0},\{Y_t\}_{t\ge0}$ be 2-dimensional standard Brownian motion independent with $Q_t$. Then from Corollary \ref{XYT independent}, we have $\{X_{T_t}+Y_{S_t}\}_{t\ge0}$ and $\{X_{T_t}-Y_{S_t}\}_{t\ge0}$ are two correlated Brownian motions.
\end{Example}

\subsection{A New Method for Simulation of Correlated Brownian motions}

Simulation is also an important part of constructing correlated Brownian motions. In this section, a new way to simulate correlated Brownian motions is given by the common decomposition method. The local correlation model characterize the correlation in the micro view and only focus on the correlation at the moment; however, the common decomposition characterize the correlation over the entire period of time, which is from the macro view. This difference of two methods may bring advantages of the new simulation method compared with the simulation method from local correlation model.


One of the most common simulation method for local correlation model is Euler-Maruyama scheme, see \cite{kloeden2013numerical}. Firstly, given a partition $\Pi$ of $[0,t]$, let
\begin{equation}\label{simulation scheme}W_t^{\Pi}=\int_0^t\rho_u^{\Pi}dB_u+\int_0^t\sqrt{1-(\rho_u^{\Pi})^2}dZ_u
=\sum_{k=0}^{n-1}(\rho_{t_k}\Delta B_{t_k}+\sqrt{1-\rho_{t_k}^2}\Delta Z_{t_k}),\end{equation}
where $\Delta B_{t_k}=B_{t_{k+1}}-B_{t_k},\Delta Z_{t_k}=Z_{t_{k+1}}-Z_{t_k}$ and $\{\rho_u^\Pi\}_{0\le u\le t}$ is defined as
$$\rho^{\Pi}_u=\rho_{t_i},\ t_i\le u<t_{i+1}.$$
Secondly, simulate $(B,W)$ by applying \eqref{simulation scheme}. Thus, the simulation result is $(B,W^\Pi)$ eventually, and there always exist simulation errors.

Under the condition that $X$, $Y$ and $T$ are mutually independent, Table \ref{comparing simulation} and Table \ref{comparing simulate path} show the specific steps of simulation by common decomposition method when we do not have the explicit expression of $T$'s distribution. As a comparison, the Euler-Maruyama scheme of local correlation model is also shown in the second column of Table \ref{comparing simulation} and Table \ref{comparing simulate path}. The common decomposition of $(B,W^\Pi)$ is denoted as $(X^\Pi,Y^\Pi,T^\Pi)$.
From Table \ref{comparing simulation}, compared with Euler-Maruyama scheme, the differences and the advantages of common decomposition method are as follow:
\begin{itemize}
\item If the trajectory is not necessary, and we only need $B_t$ and $W_t$ at time $t$, common decomposition method can reduce the time of simulations. If $\rho_t$ is a stochastic process, we have to simulate $3n$ random numbers in Euler-Maruyama scheme, i.e. $\Delta B_{t_i}$, $\Delta Z_{t_i}$, $\rho_{t_i}$, $i=0,1,\dots,n-1$. However, in common decomposition method we only need to simulate $n+2$ random numbers, i.e. $T^{\Pi}_{t_1},T^{\Pi}_{t_2},\dots,T^{\Pi}_{t_n}$, $X_{T^{\Pi}_t}^\Pi$ and $Y_{S^{\Pi}_t}^\Pi$.
\item If we have the explicit expression of $T$'s distribution, we can simulate $T_t$ directly, then we only need to simulate $X_{T_t}$ and $Y_{S_t}$, hence simulation can be reduced to 3 times.
\item The simulation error can be controlled as long as the simulation error of $T_t^{\Pi}$ can be controlled, since
$$E|X_{T_t}-X_{T_t^{\Pi}}|^2=E|T_t-T_t^{\Pi}|\le (E|T_t-T_t^{\Pi}|^2)^{\frac12}.$$
Therefore, if the explicit expression of $T$'s distribution is obtained, one can simulate $T$ directly, and then simulate $X_T$ and $Y_S$ with the similar steps in Table \ref{comparing simulation} and Table \ref{comparing simulate path}. Then there is no simulation error for $T$, hence we can simulate $(B,W)$ accurately while this is impossible for local correlation model. If the explicit expression of $T$'s distribution is inexplicit, the simulation error of two methods is the same, because both the methods simulate $(B,W^\Pi)$.
\end{itemize}


From Table \ref{comparing simulate path} if the trajectory is needed, and we do not have the explicit expression of $T$'s distribution , there is little difference between the two simulation methods.

\begin{threeparttable}[htbp]
\caption{Simulate $(B_t,W_t)$ for a given $t$ (Explicit expression of $T$'s distribution is unobtained)}
\label{comparing simulation}
\begin{tabular}{p{1cm}<{\raggedright}p{7.6cm}<{\raggedright}p{7.6cm}<{\raggedright}}
  \toprule
   &Common decomposition method & Local correlation model (Euler-Maruyama scheme)\\
  \midrule
   Step 1&Simulate $T^{\Pi}_{t_1},T^{\Pi}_{t_2},\dots,T^{\Pi}_{t_n}$ in order\tnote{1} & Simulate $\rho_{t_0},\rho_{t_1},\dots,\rho_{t_{n-1}}$ in order\\
   Step 2&Simulate $X_{T^{\Pi}_t}^\Pi$ and $Y_{S^{\Pi}_t}^\Pi$\tnote{2} & Simulate $\Delta B_{t_0}$, $\Delta B_{t_1}$, $\dots$, $\Delta B_{t_{n-1}}$ and $\Delta Z_{t_0}$, $\Delta Z_{t_1}$, $\dots$, $\Delta Z_{t_{n-1}}$\\
   Step 3&Calculate $(B_t,W^{\Pi}_t)$ & Calculate $(B_t,W^{\Pi}_t)$\\
  \bottomrule
\end{tabular}
\begin{tablenotes}

\item[1]{According to $T^\Pi_t=\frac{t+\sum_{i=0}^{n-1}\rho_{t_i}\Delta t_i}{2}$, complexity of simulating $T^{\Pi}_{t_0},T^{\Pi}_{t_1},\dots,T^{\Pi}_{t_n}$ is equal to simulating $\rho_{t_0},\rho_{t_1},\dots,\rho_{t_{n-1}}$.}
\item[2]{Under the condition of $T_{t}^\Pi$, $X_{T^{\Pi}_t}^\Pi$ and $Y_{S^{\Pi}_t}^\Pi$ are independent normal distributions with mean zero and variance $T^\Pi_{t}$ and $S^\Pi_{t}$ respectively.}
\end{tablenotes}
\end{threeparttable}

\begin{threeparttable}[htbp]
\caption{Simulate trajectory of $(B,W)$ in $[0,t]$ (Explicit expression of $T$'s distribution is unobtained)}
\label{comparing simulate path}
\begin{tabular}{p{1cm}<{\raggedright}p{7.6cm}<{\raggedright}p{7.6cm}<{\raggedright}}
  \toprule
   &Common decomposition method  & Local correlation model (Euler-Maruyama scheme)\\
  \midrule
   Step 1&Simulate $T^{\Pi}_{t_1},T^{\Pi}_{t_2},\dots,T^{\Pi}_{t_n}$ in order & same with Step 1 in Table \ref{comparing simulation}\\
   Step 2&Simulate $\Delta X^{\Pi}_{T^\Pi_{t_0}}$, $\Delta X^{\Pi}_{T^\Pi_{t_1}}$, $\dots$, $\Delta X^{\Pi}_{T^\Pi_{t_{n-1}}}$ and $\Delta Y^{\Pi}_{S^\Pi_{t_0}}$, $\Delta Y^{\Pi}_{S^\Pi_{t_1}}$, $\dots$, $\Delta Y^{\Pi}_{S^\Pi_{t_{n-1}}}$\tnote{1} & same with Step 2 in Table \ref{comparing simulation}\\
   Step 3&Calculate $B_{t_1},B_{t_2},\dots,B_{t_n}$ and $W^{\Pi}_{t_1},W^{\Pi}_{t_2},\dots,W^{\Pi}_{t_n}$ & Calculate $B_{t_1},B_{t_2},\dots,B_{t_n}$ and $W^{\Pi}_{t_1},W^{\Pi}_{t_2},\dots,W^{\Pi}_{t_n}$\\
  \bottomrule
\end{tabular}
\begin{tablenotes}
\item[1]{Under the condition of $T_{t_0}^\Pi,T_{t_1}^\Pi,\dots,T_{t_{n-1}}^\Pi$, the random variables $\Delta X^\Pi_{T^\Pi_{t_0}},\Delta X^\Pi_{T^\Pi_{t_1}},\dots,\Delta X^\Pi_{T^\Pi_{t_{n-1}}},\Delta Y^\Pi_{S^\Pi_{t_0}},\Delta Y^\Pi_{S^\Pi_{t_1}},\dots,\Delta Y^\Pi_{S^\Pi_{t_{n-1}}}$ are independent normal distributions with mean zero and variance $\Delta T^\Pi_{t_0},\Delta T^\Pi_{t_1},\dots,\Delta T^\Pi_{t_{n-1}}$, $\Delta S^\Pi_{t_0},\Delta S^\Pi_{t_1},\dots,\Delta S^\Pi_{t_{n-1}}$ respectively.}
\end{tablenotes}
\end{threeparttable}

\begin{Example}Take parameters as follow,
\begin{equation}Q_0=[1,0,0]^T,\boldsymbol{\alpha} = [0.3,0.6,0.9]^T,\boldsymbol A=\begin{bmatrix}-1 & 0.8 & 0.2\\0.4&-1&0.6\\0.3&0.7&-1\end{bmatrix},t=1,\Delta t_i=0.01,\forall i.\label{parameters}\end{equation}
Figure \ref{simulate T}, Figure \ref{simulate X_T and Y_S}, Figure \ref{calculate B and W} display how we simulate the trajectory of $(B,W)$ in $[0,t]$ through common decomposition method (explicit expression of $T$'s distribution is unobtained) step by step
\begin{figure}[htbp]
\centering
\subfigure[Step 1: Simulate $T_t$]{
\label{simulate T}
\begin{minipage}{5cm}
\centering
\includegraphics[width=5cm,height=4cm]{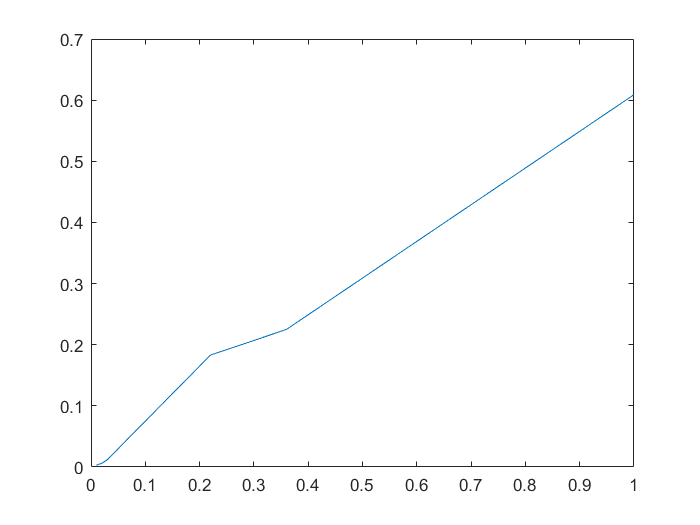}
\end{minipage}
}
\subfigure[Step 2: Simulate $X_{T_t}$ and $Y_{S_t}$]{
\label{simulate X_T and Y_S}
\begin{minipage}{5cm}
\centering
\includegraphics[width=5cm,height=4cm]{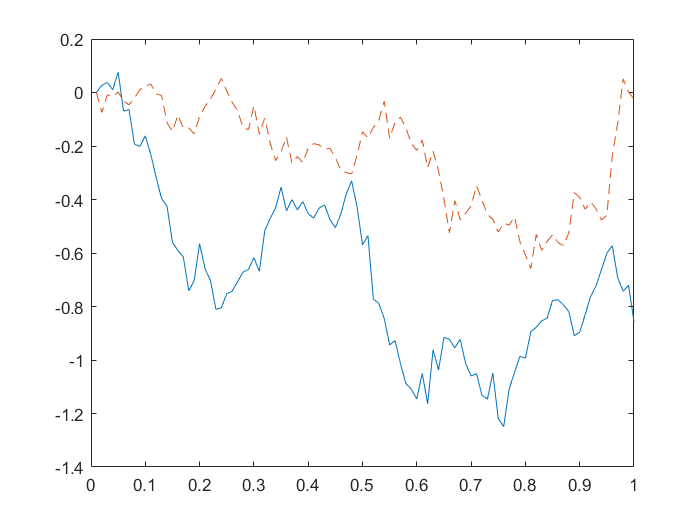}
\end{minipage}
}
\subfigure[Step 3: Calculate $B_t$ and $W_t$]{
\label{calculate B and W}
\begin{minipage}{5cm}
\centering
\includegraphics[width=5cm,height=4cm]{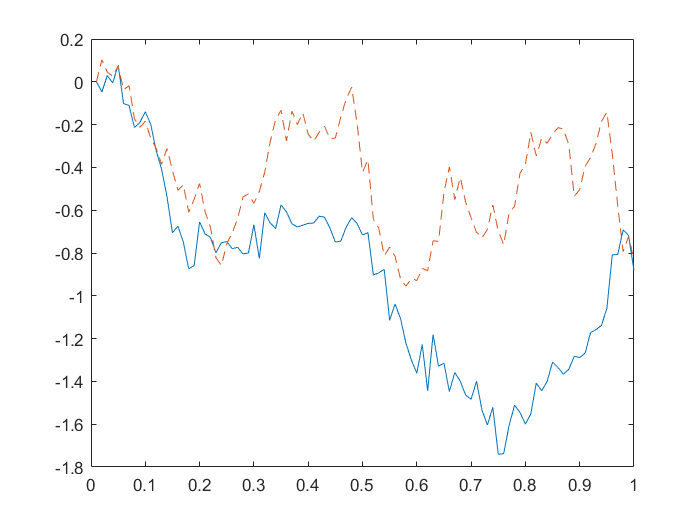}
\end{minipage}
}
\caption{Simulate $(B_t,W_t)$ by Common Decomposition Method (Explicit expression of $T$'s distribution is unobtained)}
\label{call on min calibration}
\end{figure}
\end{Example}

We consider the regime switching model in Example \ref{Q process}. Thanks to \eqref{P(t)}, simulation for regime switching model is feasible. Take the same parameters as in \eqref{parameters}, we calculate the expectation of $B_t+W_t$ by simulating $(B_t,W_t)$\footnote{Note that we do not need to simulate the trajectory here.} with $N=5000$ replications. We implement Monte Carlo methods by MATLAB2017b with a Core i7 2.8GHZ CPU.  

Table \ref{two simulation methods} shows that the standard deviation of two methods are very close, hence their simulation error are truly close. And common decomposition method runs much faster than local correlation model with Euler-Maruyama scheme.
\begin{table}[htbp]
\centering
\caption{Comparing two simulation methods}
\label{two simulation methods}
\begin{tabular}{cccc}
  \hline
   & $E(B_t+W_t)$ &Std Dev& Running time\\
  \tabincell{c}{Common decomposition method\\ (explicit expression of $T$'s distribution is unobtained)} &-0.0034 & $1.8593\times10^{-2}$ &3.1868 seconds\\
  \tabincell{c}{Local correlation model\\ (Euler-Maruyama scheme)} &0.0265 & $1.8561\times10^{-2}$ &11.4762 seconds\\
  \hline
\end{tabular}
\end{table}


\section{Financial Derivatives Valuation by Applying the Common Decomposition Method}\label{Pricing Financial Derivatives by Decomposition of Two Correlated Brownian Motions}

In the previous two sections, the common decomposition of two Brownian motions is considered, in which dependence structure could be very general. We showed how to decompose Brownian motions $(B, W)$ to a triplet $(X,Y,T)$, and we also answered how to construct two correlated Brownian motions from a given triplet $(X,Y,T)$. In this section, we will apply the common decomposition method to study the pricing problem of some typical two-factor derivatives that modeled by two correlated Brownian motions. We first give two examples showing direct usage of the common decomposition triplet $(X,Y,T)$ in pricing covariation swap, covariation option and Quanto option. And then we will focus on the pricing problem of two-color rainbow options. There are several typical examples for two-color rainbow options, one is given by option-bonds, see \cite{stulz1982options} for details; besides, a special kind of two-color rainbow options, spread options, are ubiquitous in financial markets, including equity, fixed income, foreign exchange, commodities and energy markets, \cite{carmona2003pricing} present a overview of examples and common features of spread options.


For simplicity, we assume that $X,\ Y$ and $T$ are mutually independent in this section, i.e., ${\rho}$ is independent from ($B$, $Z$) in the local correlation model by Proposition \ref{independent local correlation imply independent time change}.  This assumption is not so rigorous as to go against the reality. For example, in \cite{Ma2009Pricing}, when considering the pricing problem of foreign equity options with stochastic correlations, the author illustrated independency of $\rho$, $B$ and $Z$ from an empirical view.

\subsection{Pricing Covariance Swap and Covariance Option}\label{section-covariance-option}
Options which depend on exchange rate movements, such as those paying in a currency different from the underlying currency, have an exposure to the correlation between the asset and the exchange rate. This risk may be eliminated by two ways, a straightforward approach is Quanto option which will be discussed in Section \ref{Quanto option}; the other approach that we focus on this section is \emph{Covariance Options} or \emph{Correlation Options}, see \cite{swishchuk2016change} for more details. By combining variance and covariance options, the realised variance of return on a portfolio can be locked in. \cite{carr1999introducing} illustrated that the covariance swaps can be constructed by options and futures, in other words, options can be perfectly hedged by covariance swaps and futures. In the following part, we consider the so called {covariance options} which is designed to cope with the covariance risks of two underlying assets.

Suppose that the prices of the two assets, $(S^1,\,S^2)$, can be characterized as
\begin{equation}
\frac{dS^1_t}{S^1_t}=\mu_1dt+\sigma_1dB_t,\ \frac{dS^2_t}{S^2_t}=\mu_2dt+\sigma_2dW_t,\label{risky assets}\end{equation}
where the drifts $\mu_i,i=1,2$ and volatilities $\sigma_i,i=1,2$ of underlying assets are assumed to be constant.


\begin{Example}[Swap and Option on Realized Covariance of Returns]\label{example of covariance swap and option}
Consider two risky assets whose prices evolve as in (\ref{risky assets}). Then according to Example \ref{example of girsanove},
$(S^1,S^2)$ could be transformed to, under proper conditions,
$$\frac{dS_t^1}{S_t^1}=rdt+\sigma_1d\tilde B_t,\frac{dS_t^2}{S_t^2}=rdt+\sigma_2d\tilde W_t,$$
where $\tilde B$ and $\tilde W$ are Brownian motions under the risk neutral measure $Q$, and $r$ denotes the constant risk free interest rate. 

Continuously compounded rate returns of two assets are $\ln(S^1_t/S^1_0)$ and $\ln(S^2_t/S^2_0)$. Accordingly, the realized covariance of returns of two underlying assets is defined as the cross variation of $\ln(S^1_t/S^1_0)$ and $\ln(S^2_t/S^2_0)$
$$Cov_R(S^1_t,S^2_t)\triangleq[\ln\frac{S^1}{S^1_0},\ln\frac{S^2}{S^2_0}]_t,$$
then the payoff of covariance swap and covariance option of the underlying equity $S^1$ and $S^2$ at expiration is
$$Cov_R(S^1_t,S^2_t)-K,$$
and
$$max\{Cov_R(S^1_t,S^2_t)-K,0\},$$
where $K$ represent the strike price. Note that
$$Cov_R(S^1_t,S^2_t)=[\ln\frac{S^1}{S^1_0},\ln\frac{S^2}{S^2_0}]_t=\int_0^t\sigma_1\sigma_2d[B,W]_t=\sigma_1\sigma_2(T_t-S_t)=\sigma_1\sigma_2(2T_t-t),$$
the price of covariance swap and covariance option only depend on the expectation and distribution of $T_t$. Note that $T$ is observable in real probability measure, and the distribution of $T$ under real probability measure and risk neutral probability measure is coincident according to Example \ref{example of girsanove}, hence we can easily obtain the distribution and expectation of $T$ from historical data and then obtain the price of covariance swap and covariance option. The result of  correlation swap and correlation option is similar.
\end{Example}

\subsection{Pricing Quanto Option}\label{Quanto option}
Quanto option is a famous cross-currency financial product trading in organized exchanges as well as in OTC. Its payoff is calculated in one currency but is settled in another currency at a fixed exchange rate. It is designed to hedge the risks of delivering foreign investments to domestic currency. Hence the correlation between the underlying price and the exchange rate plays an ultimate role in pricing. Usually, this correlation structure is modeled by two correlated Brownian motons. In Section \ref{Dependency Structure of Two Correlated Brownian Motions}, we have showed that part of the dependency of two Brownian motions could be described by $T$ in common decomposition. In the following example, we will show the essential role of $T$ in the  pricing of an European-style Quanto.

\begin{Example}
Consider an European-style Quanto. Suppose the price of underlying equity $S$ in foreign currency and the exchange rate $R$ are modeled, under the risk neutral probability in the domestic currency, as follows:
$$dS_t=\mu_1S_tdt+\sigma_1S_tdB_t,\quad dR_t=\mu_2R_tdt+\sigma_2R_tdW_t,$$
and the payoff of a Quanto put option is
$$R_0\max(K-S_{t},0).$$
Let $r_1,r_2$ represent the risk free interest under domestic currency and foreign currency respectively. Under the arbitrage free assumption in domestic currency world, any discounted asset should be a martingale in risk neutral probability. Hence, consider the bank account and stock account in foreign currency, one can get
\begin{equation}R_0=\exp(-r_1t)E\left[\exp(r_2t)R_t\right],\label{bank account}\end{equation}
\begin{equation}S_0R_0=\exp(-r_1t)E\left[S_tR_t\right].\label{stock account}\end{equation}
Note that
$$E[R_t]=R_0\exp(\mu_2t),\quad E[S_tR_t]=S_0R_0\exp\left((\mu_1+\mu_2-\frac12\sigma_1^2-\frac12\sigma_2^2)t\right)E\left[\exp(\sigma_1B_t+\sigma_2W_t)\right],$$
and under the condition (C3), we have
$$E\left[\exp(\sigma_1B_t+\sigma_2W_t)\right]=\exp\left((\sigma_1^2+\sigma_2^2)\frac t2\right)E\left[\exp\left(\sigma_1\sigma_2\int_0^t\rho_udu\right)\right].$$
After simple calculations,
$$\mu_2=r_1-r_2,\quad \mu_1=r_1-\mu_2-\frac1{t}\ln E\left[\exp{(\sigma_1\sigma_2\int_0^{t}\rho_udu)}\right].$$ According to \cite{van2006modelling} and \cite{teng2016versatile}, Quantos' price is 
$$P_{Quanto}=R_0\left(Ke^{-r_1t}N(-d_2)-S_0e^{-(r_1t-r_2t+\ln E\left[\exp{(\sigma_1\sigma_2\int_0^{t}\rho_udu)}\right])}N(-d_1)\right),$$
where
$$d_1=\frac{\log(S_0/K)+(r_2+\sigma_1^2/2)t-\ln E\left[\exp{(\sigma_1\sigma_2\int_0^{t}\rho_udu)}\right]}{\sigma_1\sqrt{t}},d_2=d_1-\sigma_1\sqrt{t}.$$
Note that
$$\ln E\left[\exp{(\sigma_1\sigma_2\int_0^{t}\rho_udu)}\right]=\ln E\left[\exp\left(\sigma_1\sigma_2(2T_t-t)\right)\right],$$
then Quantos' price is actually determined by Laplace transform of $T_t$, similar with Example \ref{example of covariance swap and option}, we can obtain the Laplace transform of $T_t$ by the distribution of $T$ from historical data.
\end{Example}

\subsection{Pricing 2-Color Rainbow Options}\label{subsection_rainbow}
In this section, we focus on a class of multi-asset options, the 2-color rainbow option which is written on the maximum or minimum of two risky assets. This kind of option was first studied in \cite{margrabe1978value}, and in \cite{stulz1982options}, the author showed its extensive applications in valuing many financial instruments such as foreign currency bonds, option-bonds, risk-sharing contracts in corporate finance, secured debt, etc.

In this part we use the same asset-price models as in Section \ref{section-covariance-option}. We find an unified and analytical expression of the prices of different rainbow options.

The payoff of a rainbow option with maturity $\tau$ may have the forms listed in Table \ref{payoff} \citep{ouwehand2006pricing}. We will demonstrate that all these types of rainbow options could be valuated through a unified approach.
\begin{table}[!htbp]
\centering
\caption{Types of rainbow option}
\label{payoff}
\begin{tabular}{cc}
\toprule
Option Style & Payoff \\
  \midrule
 \emph{Best of assets or cash}& $\max(S_{\tau}^1,S^2_{\tau},K)$\\
 \emph{Put 2 and Call 1}& $\max(S_{\tau}^1-S_{\tau}^2,0)$ \\
 \emph{Call on max}&  $\max(\max(S_{\tau}^1,S_{\tau}^2)-K,0)$\\
 \emph{Call on min}& $\max(\min(S_{\tau}^1,S_{\tau}^2)-K,0)$\\
 \emph{Put on max}& $\max(K-\max(S_{\tau}^1,S_{\tau}^2),0)$\\
 \emph{Put on min}& $\max(K-\min(S_{\tau}^1,S_{\tau}^2),0)$ \\
  \bottomrule
\end{tabular}
\end{table}


Define a 2-dimensional process $\boldsymbol{M}_t=(X_{T_t},Y_{S_t})^\top$. Similar to the cases studied in \cite{carr2004time}, the payoffs in Table \ref{payoff} could be reformulated as
$$(a_1+b_1e^{\boldsymbol{\theta}_1^\top\boldsymbol{M}_{\tau}})1_{\{\boldsymbol{c}_1^\top\boldsymbol{M}_{\tau}\le k_1\}}1_{\{\boldsymbol{c}^\top\boldsymbol{M}_{\tau}\le k\}}+(a_2+b_2e^{\boldsymbol{\theta}_2^\top\boldsymbol{M}_{\tau}})1_{\{\boldsymbol{c}_2^\top\boldsymbol{M}_{\tau}\le k_2\}}1_{\{\boldsymbol{c}^\top\boldsymbol{M}_{\tau}\ge k\}},$$
with some proper parameters $a_i,b_i,\boldsymbol c_i,\boldsymbol\theta_i,k_i,i=1,2,$ and $k$.

For example, consider the Call-on-max option, whose payoff is $\max(\max(S_{\tau}^1,S_{\tau}^2)-K,0)$, the parameters are (for $i=1,2$)
$$a_i=-K,b_i=S^i_0e^{(r-\frac12\sigma_i^2)\tau}, \boldsymbol{\theta}_1=\begin{pmatrix}\sigma_1\\\sigma_1\end{pmatrix},
\boldsymbol{\theta}_2=\begin{pmatrix}\sigma_2\\ -\sigma_2\end{pmatrix},$$
$$\boldsymbol{c}_i=-\boldsymbol{\theta}_i, k_i=-\ln\frac K{b_i},
\boldsymbol{c}=\boldsymbol
{\theta}_2-\boldsymbol{\theta}_1, k=\ln\frac{b_1}{b_2}.$$
It is easy to check that
$${\{\boldsymbol{c}_i^\top\boldsymbol{M}_{\tau}\le k_i\}}=\{S^i_{\tau}\ge K\},\{\boldsymbol{c}^\top\boldsymbol{M}_{\tau}\le k\}=\{S^1_{\tau}\ge S^2_{\tau}\}.$$

Now we can present a unified valuation approach for options with payoffs in Table \ref{payoff} through process $M$. First, for given parameters $\gamma_1,\gamma_2\in\mathbb R,\gamma_3,\gamma_4,\gamma_5\in\mathbb R^2$, an intermediate valuation function $G:\,\mathbb R^2\rightarrow \mathbb R$ is defined as
\begin{equation}\label{G-func}
G(x_1,x;\gamma_1,\gamma_2,\boldsymbol{\gamma}_3,
\boldsymbol{\gamma}_4,\boldsymbol{\gamma}_5)\triangleq E^Q\left[(\gamma_1+\gamma_2e^{\boldsymbol{\gamma}_3^\top
\boldsymbol{M}_{\tau}})1_{\{\boldsymbol{\gamma}_4^\top
\boldsymbol{M}_{\tau}\le x_1\}}1_{\{\boldsymbol{\gamma}_5^\top
\boldsymbol{M}_{\tau}\le x\}}\right],
\end{equation}
where $E^Q$ indicates the expectation under the risk-neutral measure $Q$.
It is obvious that the initial price of a rainbow option could be given by $G$ as
\begin{equation}\label{prices by G}
  e^{-r\tau}\Big[G(k_1,k;a_1,b_1,\boldsymbol{\theta}_1,\boldsymbol{c}_1,
\boldsymbol{c})+G(k_2,-k;a_2,b_2,\boldsymbol{\theta}_2,\boldsymbol{c}_2,
-\boldsymbol{c})\Big].
\end{equation}
For simplicity, we omit the parameters $\gamma_i,\, i=1,\dots,5,$ in the function expressions when there is no confusion. The following proposition gives a general rule to calculate function $G$.
\begin{Proposition}\label{fourier transform}
Suppose $X,Y$ and $T$ are mutually independent. Let $G(x_1,x),\,(x_1,x)\in \mathbb R^2$ be given as in (\ref{G-func}), and $L_t$ represent the Laplace transform of $T_t$. Then the characteristic function of $\boldsymbol{M}_{\tau}$ is as follows,
\begin{equation}\label{Phi M}\Phi_{\boldsymbol{M}_{\tau}}(z_1,z_2)=e^{-\frac12\tau z_2^2}
L_{\tau}(-\frac12(z_1^2-z_2^2)).\end{equation}
Moreover, the generalized fourier transform of $G(x_1,x),$ denoted by $\hat G(\lambda_1,\lambda)$, is given as
\begin{equation}\label{hat G}\hat G(\lambda_1,\lambda)=-\frac{\gamma_1}{\lambda\lambda_1}\Phi_{\boldsymbol{M}_{\tau}}(\lambda_1\boldsymbol{\gamma}_4+\lambda\boldsymbol{\gamma}_5)-\frac{\gamma_2}{\lambda\lambda_1}\Phi_{\boldsymbol{M}_{\tau}}(\lambda_1\boldsymbol{\gamma}_4+\lambda\boldsymbol{\gamma}_5-i\boldsymbol{\gamma}_3),\end{equation}
where $Im\lambda, Im \lambda_1>0$. In particular, if $\rho_t=\rho$ is a constant, then $L_t(z)=\exp(\frac{1+\rho}2tz)$ and $\hat G(\lambda_1,\lambda)$ can be obtained from \eqref{Phi M} and \eqref{hat G}.
\end{Proposition}

Given Proposition \ref{fourier transform}, the function $G(x_1,x;\gamma_1,\gamma_2,\boldsymbol{\gamma}_3,
\boldsymbol{\gamma}_4,\boldsymbol{\gamma}_5)$ could be calculated by the inversion formula and numerical method, then the prices of rainbow options are obtained from (\ref{prices by G}).

\begin{Remark}
For general cases where the payoffs can not be represented as before, Proposition \ref{fourier transform} is un available. But we can still apply the  Fourier-transform method directly to pricing functionals.
For given parameters $(S_0,\tau,r,\sigma_1,\sigma_2)$, rewrite the option payoffs as  $V(y_1+B_{\tau},y_2+W_{\tau})$ , where $y_i=(\frac{r}{\sigma_i}-\frac{\sigma_i}{2})\tau, i=1,2$. Denote by $f(b,w)$ the joint probability density of $B_{\tau}$ and $W_{\tau}$ under $Q$, then the price of $V(y_1+B_{\tau},y_2+W_{\tau})$ is $$C(y_1,y_2)=\int_{-\infty}^{\infty}\int_{-\infty}^{\infty}
V(y_1+b,y_2+w)f(b,w)dbdw.$$

According to \cite{leentvaar2008multi}, the Fourier transform of $C(y_1,y_2)$ is
\begin{align*}\hat C(\lambda_1,\lambda_2)=&\int_{-\infty}^{\infty}\int_{-\infty}^{\infty}\int_{-\infty}^{\infty}\int_{-\infty}^{\infty}e^{i\lambda_1y_1+i\lambda_2y_2}V(y_1+b,y_2+w)f(b,w)dbdwdy_1dy_2\\
=&\hat V(\lambda_1,\lambda_2)E^Q\left[e^{-i\lambda_1
B_{\tau}-i\lambda_2W_{\tau}}\right],
\end{align*}
where $\hat V$ denotes the Fourier transform of $V$. In general, $\hat V$ has no explicit expression and thus usually be calculated numerically.

When the correlation coefficient of $B$ and $W$ is constant, $E^Q\left[e^{-i\lambda_1
B_{\tau}-i\lambda_2W_{\tau}}\right]$ could be calculated explicitly, $$E^Q\left[e^{-i\lambda_1
B_{\tau}-i\lambda_2W_{\tau}}\right]=
\exp\left(-(\lambda_1^2+\lambda_2^2+2\rho\lambda_1\lambda_2)\tau\right).$$
In this case, \cite{leentvaar2008multi} have put forward a numerical method to calculate $\hat V$.

When the correlation coefficient of $B$ and $W$ is not constant, we can still use similar approaches as in \cite{leentvaar2008multi} by means of common decomposition. Continuing to use the notions as before, we have
$$E^Q\left[e^{-i\lambda_1
B_{\tau}-i\lambda_2W_{\tau}}\right]=\Phi_{\boldsymbol{M}_{\tau}}
(-\lambda_1-\lambda_2,-\lambda_1+\lambda_2)=
e^{-(\lambda_1-\lambda_2)^2\tau}L_{\tau}(-2\lambda_1\lambda_2).$$
Consequently,
\begin{equation}\label{fourier transform for arbitrary payoff}\hat C(\lambda_1,\lambda_2)=\hat V(\lambda_1,\lambda_2)e^{-(\lambda_1-\lambda_2)^2
\tau}L_{\tau}(-2\lambda_1\lambda_2).\end{equation}
Hence when the Laplace transform $L_t$ of $T_t$ is known, the price can be obtained by inverse Fourier transform formula.
\end{Remark}

In the previous discussion, we considered how to calculate the price of a rainbow option. Actually, following similar approach outlined in Proposition \ref{fourier transform}, we could give a Fourier-transform method for calculating Greeks. The next corollary set forth an example of this.
\begin{Corollary}
Consider the Delta of $S^1$ for a Call-on-Max option listed in Table \ref{payoff}, which is denoted by $\Delta(S^1)$. After calculations, we have
\begin{align*}
\Delta(S^1)=&\frac1{S_0^1}\left(\frac{\partial G}{\partial x_1}(k_1,k;a_1,b_1,\boldsymbol{\theta}_1,\boldsymbol{c}_1,\boldsymbol{c})
\right.\left.+\frac{\partial G}{\partial x}(k_1,k;a_1,b_1,\boldsymbol{\theta}_1,
\boldsymbol{c}_1,\boldsymbol{c})-\frac{\partial G}{\partial x}(k_2,-k;a_2,b_2,\boldsymbol{\theta}_2,\boldsymbol{c}_2,
-\boldsymbol{c})\right)\\
&+e^{(r-\frac12\sigma_1^2)\tau}\frac{\partial G}{\partial \gamma_2}(k_1,k;a_1,b_1,\boldsymbol{\theta}_1,\boldsymbol{c}_1,
\boldsymbol{c})\\
:=&g_1(k_1,k)+g_2(k_2,-k),
\end{align*}
where
\begin{align*}g_1(k_1,k)=&\left(\frac1{S_0^1}(\frac{\partial G}{\partial x_1}+\frac{\partial G}{\partial x})+e^{(r-\frac12\sigma_1^2)\tau}\frac{\partial G}{\partial \gamma_2}\right)(k_1,k;a_1,b_1,\boldsymbol{\theta}_1,\boldsymbol{c}_1,
\boldsymbol{c}),\\
g_2(k_2,-k)=&\left(-\frac1{S_0^1}\frac{\partial G}{\partial x}\right)(k_2,-k;a_2,b_2,\boldsymbol{\theta}_2,\boldsymbol{c}_2,
-\boldsymbol{c}).\end{align*}
The Fourier transform of $g_1$ has an explicit expression as
$$\frac{ia_1}{S_0^1}(\frac1{\lambda}+\frac1{\lambda_1})
\Phi_{\boldsymbol{M}_{\tau}}(\lambda_1\boldsymbol{c}_1
+\lambda\boldsymbol{c})+(\frac{ib_1}{S_0^1\lambda_1}
-\frac{e^{(r-\frac12\sigma_1^2)\tau}}{\lambda\lambda_1})
\Phi_{\boldsymbol{M}_{\tau}}(\lambda_1\boldsymbol{c}_1+
\lambda\boldsymbol{c}-i\boldsymbol{\theta}_1),$$
and the expression of Fourier transform of $g_2$ is
$$-\frac{ia_2}{S_0^1\lambda_2}\Phi_{\boldsymbol{M}_{\tau}}
(\lambda_2\boldsymbol{c}_2-\lambda\boldsymbol{c})-\frac{ib_2}
{S_0^1\lambda_2}\Phi_{\boldsymbol{M}_{\tau}}(\lambda_2
\boldsymbol{c}_2-\lambda\boldsymbol{c}-i\boldsymbol{\theta}_2).$$
$\Delta(S_1)$ can be obtained by the inverse Fourier transform formula. Other Greeks can be derived  along the same procedure.
\end{Corollary}

From the foregoing content of this section, we know that,  thanks to the common decomposition method, in order to calculate the price and Greeks of a rainbow option, we only need to find out the Laplace transform of $T_t$. We consider some specific models of $T_t$ in the following examples to give the readers more intuitive insights.


\begin{Example}\label{example of regime switch}Consider the regime switching model given in Example \ref{Q process},
%
by Lemma A.1 in \cite{JOHN2002AMERICAN}, the Laplace transform of $T_t$ is
$$L_t(z)=Ee^{zT_t}=\boldsymbol{1}^\top e^{(A+z{\rm diag}\boldsymbol{\alpha})t}Q_0,$$
where ${\rm diag}\boldsymbol{\alpha}=\begin{bmatrix}\alpha_1 & 0 &\cdots&0\\0&\alpha_2&\cdots&0\\ \vdots&\vdots&\ddots&\vdots\\ 0&0&\cdots&\alpha_n\end{bmatrix}$, $A=(a_{ij})_{n\times n}$ is the generator of $Q_t$. Then by Proposition \ref{fourier transform}, we can get the option price from $L_t(z)$. For example, if the option style is Call-on-max, then
\begin{align*}\hat G(\lambda_1,\lambda;a_1,b_1,\boldsymbol{\theta}_1,\boldsymbol{c}_1,
\boldsymbol{c})=&\frac{K}{\lambda\lambda_1}e^{-\frac12(\lambda_1\sigma_1-\lambda\sigma_2)^2\tau}\boldsymbol{1}^\top e^{(A-2\lambda_1\lambda\sigma_1\sigma_2{\rm diag}\boldsymbol{\alpha})\tau}Q_0\\
&-\frac{S_0^1}{\lambda\lambda_1}e^{(r-\frac12\sigma_1^2-\frac12(\lambda_1\sigma_1+i\sigma_1-\lambda\sigma_2)^2)\tau}\boldsymbol{1}^\top e^{(A-2(\lambda_1+i)\lambda\sigma_1\sigma_2{\rm diag}\boldsymbol{\alpha})\tau}Q_0.\end{align*} 
\end{Example}

In the next example, $\{T_t\}_{t\ge0}$ has a specific modelling through a bounded function of some stochastic processes and the Laplace transform of $T_t$  is given by a PDE.
\begin{Example}
Suppose that $f$ is a bounded function with values in $(0,1)$ and $\nu$ is a diffusion process satisfying the following SDE $$d\nu_t=\mu(t,\nu_t)dt+\sigma(t,\nu_t)dZ_t,$$
where $\{Z_t\}_{t\ge0}$ is a Brownian motion and $\mu(t,x),\sigma(t,x)$ are determined functions such that the SDE have an unique solution.

Let $T_t=\int_0^tf(\nu_s)ds$. By Feynman-Kac formula, the Laplace transform of $T_t-T_s$ for fixed $t$ under the condition $\nu_s$, which is denoted by $L(s,\nu_s;t,z)$, satisfies the following PDE:
\begin{equation}\label{PDE2}\frac{\partial L}{\partial s}+\mu(t,\nu)\frac{\partial L}{\partial \nu}+\frac12\sigma(t,\nu)^2\frac{\partial^2 L}{\partial \nu^2}+zf(\nu)L=0,\end{equation}
with terminal condition $L(t,\nu_t;t,z)=1.$ The solution of \eqref{PDE2} are related with Sturm-Liouville problem, see \cite{Polyanin2002Handbook} 1.8.6.5 and 1.8.9 for more details.

Particularly, the stochastic correlation model considered in \cite{Teng2016Modelling} is equivalent to the special case $f(x)=\frac{1+\tanh(x)}2$. The model discussed in \cite{Ma2009Pricing} is equivalent to $f(x)=\frac{1+x}2$ and $\nu$ is a bounded Jacobi process
$$d\nu_t=\kappa(\theta-\nu_t)dt+\sigma_{\nu}\sqrt{(h-\nu_t)(\nu_t-l)}dZ_t.$$
the boundary for bounded Jacobi process is $[l,h]$ when
$$\kappa(\theta-l)>\frac12\sigma_{\nu}^2(h-l),\kappa(h-\theta)>\frac12\sigma_{\nu}^2(h-l).$$

\end{Example}

Sometimes, there is no closed-form solution of financial derivatives, so Monte Carlo method is needed. The simulation method through common decomposition have been illustrated in Section \ref{Construction of Two Correlated Brownian Motions}.

\section{Numerical Results}\label{Numerical Results}

In literatures that study the pricing problem of two-assets derivatives with models driven by two Brownian motions, $B\ \mbox{and}\ W$,  it is a commonly used assumption that the local correlation of $B\, \mbox{and}\, W$ is a constant, i.e., $d[B,W]_t=\rho dt$ for some $\rho\in[-1,1]$. However, as we have mentioned before, this assumption is inconsistent with empirical studies. For example, based on data from different markets around the world, \cite{Chiang2007Dynamic}, \cite{syllignakis2011dynamic} and \cite{junior2012correlation} all found that the correlation coefficients changed as time and economic situations changed.
Then it is natural to ask, when the actual correlation coefficient is dynamic and stochastic, how much it would influent the pricing error if we still applied the constant-correlation model?

In this part, we consider the price of two-color rainbow options as an example. We investigate the difference of option prices under constant and dynamic correlations by numerical experiments and try to summarize when this difference is negligible or nonnegligible.

Since our concern is in the correlation of underlying assets, we assume for simplicity that all coefficients of the underlying assets, except for the local correlation,  are constants. Thus the underlying prices are assumed to satisfy (under the risk neutral probability)
$$\frac{dS_t^1}{S_t^1}=rdt+\sigma_1dB_t,\frac{dS_t^2}{S_t^2}
=rdt+\sigma_2dW_t.$$
For the dependency structure of $(B,W)$, we apply the regime switching model in this section which has been introduced in Example \ref{Q process} and Example \ref{example of regime switch}. Suppose that the market has three different states described by a finite-state-space Markov process $\{Q_t\}_{t\geq 0}$ with an initial value $Q_0$ and a transition rate matrix $A$. Thus the local correlation process of $B$ and $W$ is as follows,
$$\rho_t=2\alpha^\top Q_t-1.$$
Note that $\frac{d\left[\log S^1,\,\log S^2\right]_t/dt}{\sigma_1\sigma_2}=\rho_t$, hence $\alpha\in(0,1)^3$ indicates the switching states for local correlation coefficient of log prices. For example, if $\alpha=[0.3,0.6,0.9]^\top$, at any time $t$, $\rho_t$ switches among $-0.4, 0.2$ and $0.8$ according to the market conditions. In the rest of this section, parameters are taken as follow unless otherwise specified,
\begin{equation}\label{chap5 generator}
r=0.05,S_0^1=100,S_0^2=120,\sigma_1=0.2,\sigma_2=0.3 A=
\begin{bmatrix}-1 & 0.8 & 0.2\\0.4&-1&0.6\\0.3&0.7&-1
\end{bmatrix}.
\end{equation}

Consider the two-color rainbow options as in Section \ref{subsection_rainbow}, note that, under the above model, if $\rho$ is considered as a constant, the option prices can be given in closed form as in \cite{stulz1982options}. While for the actual case with a regime-switching $\rho$, we can apply Proposition \ref{fourier transform} to derive the true prices. Following the notations in Proposition \ref{fourier transform}, by the inversion fourier formula, we have
\begin{equation}G(x_1,x)=\int_{-\infty+i\lambda_{i}}^{\infty+i\lambda_i}\int_{-\infty+i\lambda_{1i}}^{\infty+i\lambda_{1i}}e^{-i\lambda_1x_1-i\lambda x}\hat G(\lambda_1,\lambda)d\lambda_1d\lambda,\label{inversion fourier formula}\end{equation}
where $\lambda_{i},\lambda_{1i}$ denote the imaginary part of $\lambda$ and $\lambda_1$.

Since $\hat G(\lambda_1,\lambda)$ is well defined only for $\lambda_1,\lambda $ with strictly positive imaginary, we choose $\lambda_{1i}=\lambda_i=1$ in the subsequent numerical experiment. Note that \eqref{inversion fourier formula} remains valid for any $\lambda_{1i},\lambda_i>0$. And we approximate \eqref{inversion fourier formula} by
$$G(x_1,x)\approx\sum_{j=-N_1}^{N_1}\sum_{k=-N}^Ne^{\lambda_{1i}x_1+\lambda_ix-i(j\eta_1x_1+k\eta x)}\hat G(j\eta_1+i\lambda_{1i},k\eta+i\lambda_i)\eta_1\eta,$$
where we set $N_1=N=1000$ and $\eta_1=\eta=0.1$.

Suppose that the contract life of the option is $\tau=0.25$ and the strike is $K=90$. Let $\alpha=[0.6,0.6,0.6]^\top$, then the regime switching model degenerates to the constant correlation model. We verified the group of parameters are accurate enough and the difference of option price obtained from \cite{stulz1982options} and Proposition \ref{fourier transform} is smaller than $10^{-13}$.

%
%

In the following subsection, we compare the option prices induced by the constant-$\rho$ models in \cite{stulz1982options} to the prices given by the regime-switching-$\rho$ models through (\ref{prices by G}). Since we have assumed the regime-switching case to be actual, the latter could be regarded as the ``true" prices. And thus the comparison results will indicate how large the pricing error would be when we substituted a constant for the original nonconstant $\rho$. For clarity, we make comparison in an ideal situation that the investor knows exactly the other coefficients except for $\rho$.\footnote{In empirical, the risk free interest $r$ can be observed and $\sigma_1,\sigma_2$ can be calibrated precisely from vanilla options.}

\subsection{Numerical Experiments of Pricing Rainbow Options}\label{Numerical Experiments of Pricing Rainbow Options}
In Section \ref{Comparing Constant and Dynamic Correlation Model in Option Pricing}, we compare the constant correlation model and dynamic correlation model in a more theoretical way. We assume that the investor estimates $\rho$ historically from the observed stock prices. The numerical results in this section show that there may be big differences between the prices of two models. In Section \ref{Calibrating Constant Correlation Model into Dynamic Correlation Model}, we adopt an approach more close to the practical procedure. We suppose the investor calibrate the constant correlation model to option prices he observed (which were calculated from the regime-switching model). And then the calibrated model is used for pricing. And it shows that there will be a big pricing error by using constant correlation model, especially for those options deep out of the money. This is in line with the results given in \cite{costin2016expectations} for CDS options.
\subsubsection{Numerical Analysis of Constant and Nonconstant Correlation in Pricing Rainbow Options}\label{Comparing Constant and Dynamic Correlation Model in Option Pricing}

In this section, we estimate a constant correlation coefficient $\hat{\rho}$ from the historical data which are given by the regime switching model, and then calculate the option prices derived from this $\hat{\rho}$ \footnote{We have illustrated in Remark \ref{estimate rho from real probability measure} that it is feasible to apply directly the estimated $\hat\rho$ from historical data into option pricing.}. By comparing these option prices with those deriving directly from the regime switching model, we can get a general idea of the error we would make when applying constant correlation model in the situations where the actual correlation coefficients are dynamic and stochastic. For the robustness of the results, we consider the comparisons in different cases with different vector $\alpha$s.

Since we have assumed that all the  other parameters can be obtained precisely, the investor actually could  get the data of $(B,W)$ by observing prices of the underlying assets. Suppose that the investor has got these historical data of a long term and with a relatively high frequency as $(B_{t_i},W_{t_i}),i=0,1,\dots,n$, where $0=t_0<t_1<\dots<t_n=t$. According to definition, the estimated constant correlation based on data till time $t$ is
$$\hat{\rho}\triangleq\frac{\sum_{i=0}^{n-1}\Delta B_{t_i}\Delta W_{t_i}}t.$$
Note that, setting $\Delta t=\max\{t_{i+1}-t_i|i=0,\dots,n\}$, we have 
$$\frac{\sum_{i=0}^{n-1}\Delta B_{t_i}\Delta W_{t_i}}t\xrightarrow[\Delta t\rightarrow 0]{P}\frac{[B,W]_t}{t}=\frac{T_t-S_t}{t}
=\frac1t\int_0^t(2\boldsymbol{\alpha}-\boldsymbol{1})^{\top}Q_sds,$$
and according to the Ergodic Theorem of Markov processes,
$$\lim_{t\to\infty}\frac1t\int_0^tQ_sds=\boldsymbol{\pi},$$
where $\boldsymbol{\pi}$ denotes the stationary distribution of the Markov process $Q_t$.

Therefore, as long as we assume these data to be long-term and with a relatively high frequency, we always have
\begin{equation}\label{estimate rho}\hat{\rho}\approx\frac1t\int_0^t(2\boldsymbol{\alpha}-\boldsymbol{1})^{\top}Q_sds\approx2\boldsymbol{\alpha}^{\top}\boldsymbol{\pi}-1 \footnote{Note that the stationary distribution $\boldsymbol{\pi}$ satisfies the following equations
$$A^{\top}\boldsymbol{\pi}=0,\ \boldsymbol{\pi}\boldsymbol{1}^{\top}=1,$$ where $A$ denotes the generator of $Q$. In our numerical experiments, $A=\begin{bmatrix}-1 & 0.8 & 0.2\\0.4&-1&0.6\\0.3&0.7&-1\end{bmatrix}$, and then $\boldsymbol{\pi}=[0.2636,0.4273,0.3091]^{\top}.$}.\end{equation}

In this case, no matter how violently the correlation coefficient switches over time, the investors may have similar estimates from long-term historical data. And thus the option prices calculated along these estimates may deviate a lot from the ``true" prices. We will show these prices' deviations by the relative error defined as
\begin{equation}
\mbox{Relative error}=\frac{\mbox{Price with constant $\hat \rho$}-\mbox{Price with regime switching $\rho$}}{\mbox{Price with regime switching $\rho$}}.\label{relative error}
\end{equation}

In the numerical experiments, for each case, we simulate a path of $(B,W)$ to present the historical data, where we choose $t=20$ and $\Delta t_i=0.05,\forall i$. In order to make consistent comparison, we randomly choose $5$ different $\boldsymbol{\alpha}$, which all satisfy the condition $2\boldsymbol{\alpha}^{\top}\boldsymbol{\pi}-1=0.2$. That is to say, by \eqref{estimate rho}, the option prices calculated from the estimated coefficients are similar since in all cases $\hat{\rho}\approx0.2$. While on the contrary, we shall see that the prices calculated from original model are quite different from each other.

We list the numerical results in Table \ref{numerical result Q_0}, in which the second column shows the ``true" prices calculated from the original regime switching model, the third column shows the $\hat\rho$ estimated from the "historical data", the forth column shows the prices obtained by constant correlation model with $\hat\rho$, while the last column shows the relative errors defined as in (\ref{relative error}).
\begin{table}[htbp]
\centering
\small
\caption{Option pricing with all history data}
\label{numerical result Q_0}
\begin{tabular}{ccccc}
  \hline
  $\boldsymbol{\alpha}$ &\quad True Prices \quad & $\hat{\rho}$ & Prices with $\hat{\rho}$ & Relative errors \\
  \hline
  $[0.7665,0.7551,0.2436]^{\top}$&37.2642&0.2377&35.2623&-5.37\%\\
  $[0.8068,0.8772,0.0404]^{\top}$&38.2361&0.2103&35.3671&-7.50\%\\
  $[0.6824,0.6178,0.5051]^{\top}$&35.9230&0.2436&35.2398&-1.90\%\\
  $[0.5559,0.4063,0.9054]^{\top}$&33.8134&0.1911&35.4403&4.81\%\\
  $[0.6,0.6,0.6]^{\top}$&35.4064&0.2177&35.3388&0.19\%\\
  \hline
\end{tabular}
\end{table}

It is obviously from Table \ref{numerical result Q_0} that there may be big pricing errors when using constant correlation coefficient estimated from historical data. In this numerical example, although all the other coefficients were assumed to induce zero error, the relative errors for pricing can mount to unacceptable levels. It is almost certain that these high errors come from the substitution of $\hat \rho$s for the real dynamic stochastic $\rho$s. As a verification, we consider the case of $\boldsymbol{\alpha}=[0.6,0.6,0.6]^{\top}$, where the regime switching model degenerates to the constant correlation model. The results are shown in the last row of the table. We can see that there is only a small relative error, $0.19\%$, which presents the technical error other than substitution of constant correlations to dynamic ones.

More specifically, we can see that in all cases the estimated $\hat\rho$s are around $0.2$, and thus the resulting option prices are around $35.3$, while the true prices deviate from as high as $38.2$ to as low as $33.8$. There would be a big unexpected loss if the investor applied the constant correlation model to value these options and used these prices as a guidance of his investments.

\subsubsection{Calibrating a Constant Correlation Model from Data Given by the Dynamic Correlation Model}\label{Calibrating Constant Correlation Model into Dynamic Correlation Model}

In this section, we investigate the difference between option prices under constant correlation model and dynamic stochastic correlation model through a more practical way. First, in practice, when considering derivatives' pricing, investors do not use coefficients estimated from historical data commonly. More often, they observe the market prices of a class of derivatives, and calibrate the theoretical model to the observed prices. In our case, the "market prices" are supposed to be given by the regime switching model, and the " theoretical model" held by investors is supposed to be the constant correlation model. And "calibration of the theoretical model" reduces to " finding the optimum $\rho$ to fit the market prices" since this is assumed to be the only unknown parameter for the theoretical model. On the other hand, just like the idea of "implied volatility", each observed option price can deduce an "\emph{implied correlation}",  $\rho_{imp}$. The change of $\rho_{imp}$ with strikes can also indicate the deviation of option prices given by constant correlation model from actual prices based on dynamic correlation.

The numerical simulations are carried out along the procedure in the following.

First, we give the prices for options with a maturity $\tau=0.25$ and strikes $K=80,90,\dots,140$ under regime switching model by the Fourier transform method. These will play the part of "initial market data" in our numerical experiment.

Then based on these data, we calibrate the constant correlation model to a proper $\rho$.\footnote{Just as before, all the other coefficients are supposed to be known exactly.} This is done by minimizing the following cumulative square error function by Gradient Descent method, \footnote{The initial value is taken as $\rho=0$. The step size is set as $|0.01/L^{\prime}(0)|$ where $L^{\prime}$ denotes the first derivative of $L$. The gradient descent method terminates when $|L^{\prime}(\rho)|$ is smaller than $10^{-4}$.}
$$L(\rho)=\sum_{n}\left(Price^{\mbox{constant}}_n
(\rho)-Price_n^{\mbox{dynamic}}\right)^2.$$

And then, the calibrated correlation coefficients are applied to the constant correlation model for pricing options with strikes $K=82,84,\dots,88$, $92,94,\dots,98$, $\dots$, $132,134,\dots,138$. The resulting prices will be compared with the prices under regime switching model.

To see the variations of implied correlation, we apply the definition of $\rho_{imp}$ given by \cite{da2007option} which satisfies
$$Price=Price^{\mbox{constant}}(\rho_{imp}),$$ to the prices given by regime switching models with more strikes $K=80,82,84,\dots,140$.

In the following, we run through the calibrating-pricing procedure for Call on Min, Call on Max, Put on Max and Put on Min options, consider their relative errors defined as in \eqref{relative error}, and calculate the implied correlations respectively. We show the results in Figures \ref{call on min calibration}-\ref{put on min calibration}. In each figure, the dotted line separates the curve into two parts, the out-of-the-money case (in figures, the left part for puts or the right for calls) and the in-the-money case. The intersection is at-the-money case.

\begin{figure}[htbp]
\centering
\subfigure[Relative error (Calibrated $\rho=-0.3190$)]{
\label{call on min RE}
\begin{minipage}{7cm}
\centering
\includegraphics[width=7cm,height=6cm]{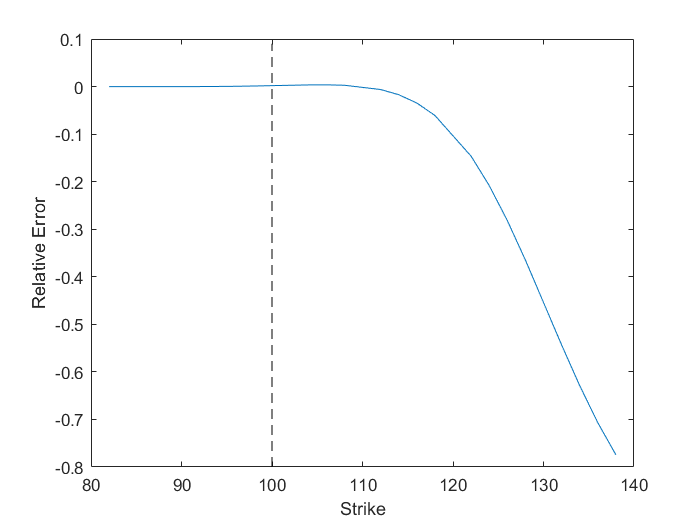}
\end{minipage}
}
\subfigure[Implied correlation]{
\label{call on min impcorr}
\begin{minipage}{7cm}
\centering
\includegraphics[width=7cm,height=6cm]{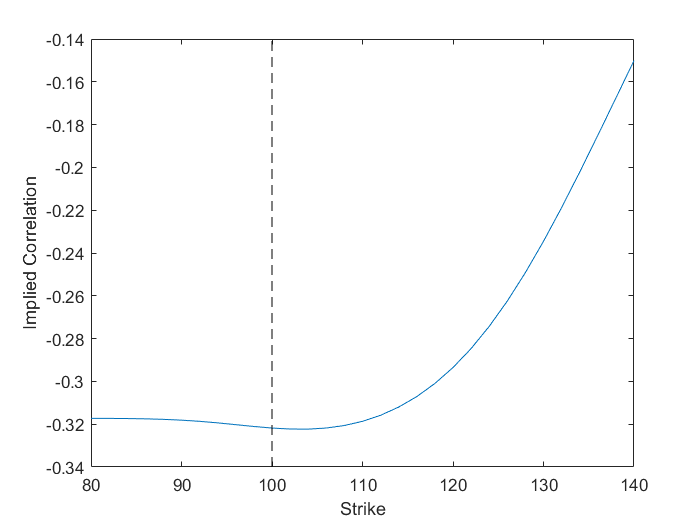}
\end{minipage}
}
\caption{Call on Min option with $Q_0=[1,0,0]^{\top}$, $\boldsymbol{\alpha} = [0.3,0.6,0.9]^{\top}$}
\label{call on min calibration}
\end{figure}

\begin{figure}[htbp]
\centering
\subfigure[Relative error (Calibrated $\rho=-0.3190$)]{
\label{call on max RE}
\begin{minipage}{7cm}
\centering
\includegraphics[width=7cm,height=6cm]{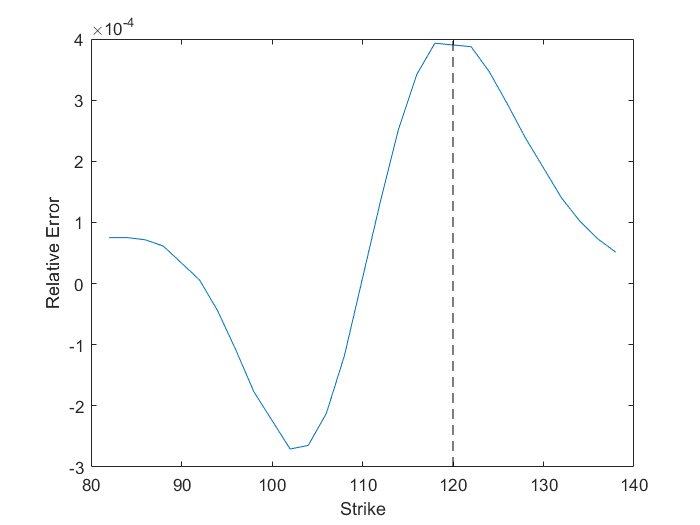}
\end{minipage}
}
\subfigure[Implied correlation]{
\label{call on max impcorr}
\begin{minipage}{7cm}
\centering
\includegraphics[width=7cm,height=6cm]{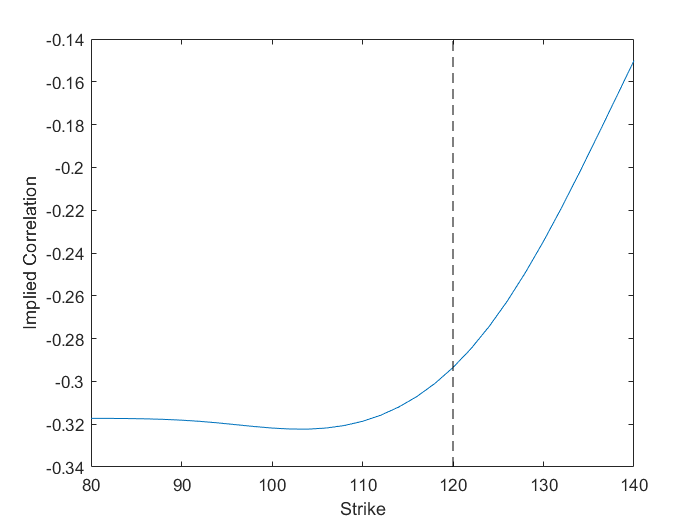}
\end{minipage}
}
\caption{Call on Max option with $Q_0=[1,0,0]^{\top}$, $\boldsymbol{\omega} = [0.3,0.6,0.9]^{\top}$}
\label{call on max calibration}
\end{figure}

\begin{figure}[htbp]
\centering
\subfigure[Relative error (Calibrated $\rho=-0.3177$)]{
\label{put on max RE}
\begin{minipage}{7cm}
\centering
\includegraphics[width=7cm,height=6cm]{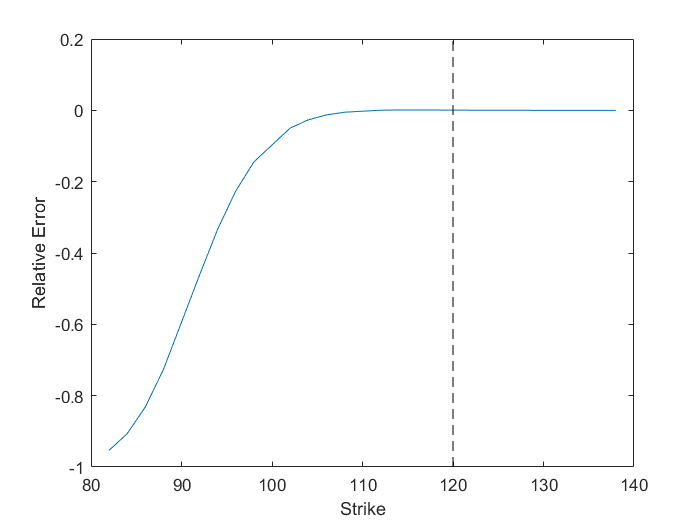}
\end{minipage}
}
\subfigure[Implied correlation]{
\label{put on max impcorr}
\begin{minipage}{7cm}
\centering
\includegraphics[width=7cm,height=6cm]{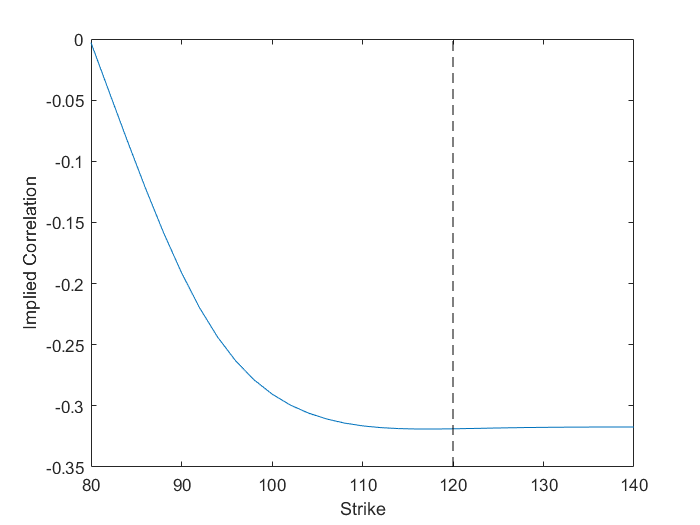}
\end{minipage}
}
\caption{Put on Max option with $Q_0=[1,0,0]^{\top}$, $\boldsymbol{\alpha} = [0.3,0.6,0.9]^{\top}$}
\label{put on max calibration}
\end{figure}

\begin{figure}[htbp]
\centering
\subfigure[Relative error (Calibrated $\rho=0.4940$)]{
\label{put on min RE}
\begin{minipage}{7cm}
\centering
\includegraphics[width=7cm,height=6cm]{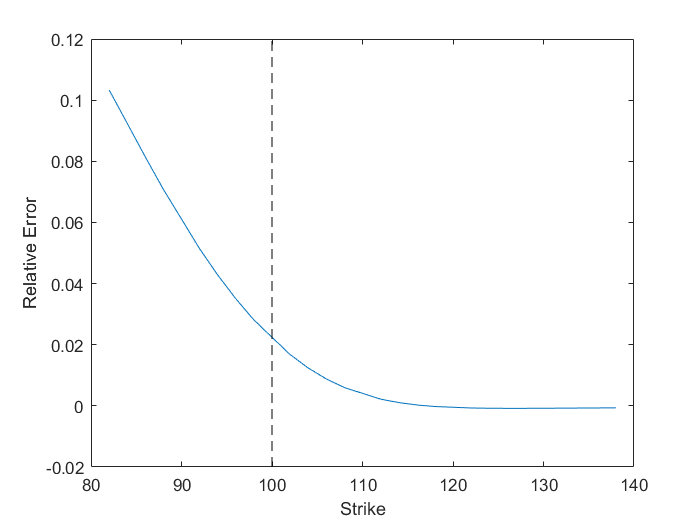}
\end{minipage}
}
\subfigure[Implied correlation]{
\label{put on min impcorr}
\begin{minipage}{7cm}
\centering
\includegraphics[width=7cm,height=6cm]{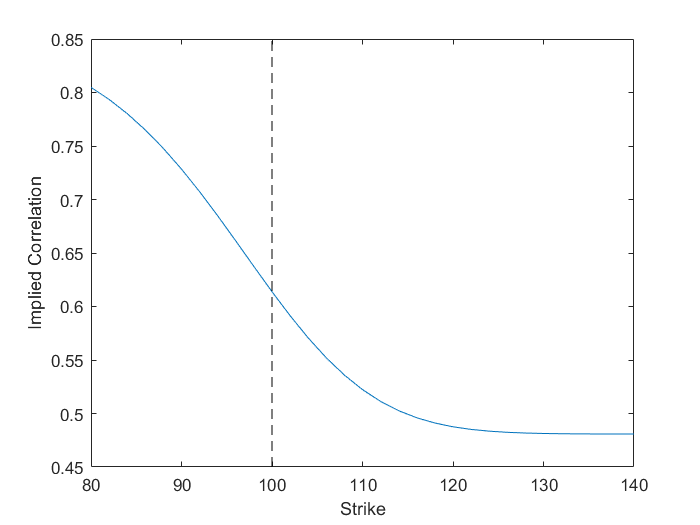}
\end{minipage}
}
\caption{Put on Min option with $Q_0=[0.2,0,0.8]^{\top}$, $\boldsymbol{\alpha} = [0.3,0.6,0.95]^{\top}$}
\label{put on min calibration}
\end{figure}

On the first try, we choose parameters $Q_0=[1,0,0]^{\top}$ and $\boldsymbol{\alpha} = [0.3,0.6,0.9]^{\top}$ to generate the regime switching model. The immediate observation is the huge pricing error for deep-out-of-the-money options of Put on Max and Call on Min. The relative error reaches more than $70\%$, which is shown in Figure \ref{call on min RE} and \ref{put on max RE}.  While for Call on Max option, the relative error is no more than $0.1\%$, as shown in Figure \ref{call on max RE}. And it is also small for Put on Min option whose figure is omitted here since the relative error always lies below the level $0.5\%$.

To see whether this is a common property or not, we change the initial regime switching model to a new one with parameters $Q_0=[0.2,0,0.8]^{\top}$ and $\boldsymbol{\alpha} = [0.3,0.6,0.95]^{\top}$, and repeat the calibrating-pricing procedure. For Call on Max, Call on Min and Put on Max, the results are really similar with the previous group of parameters and we omit the figures. But for Put on Min, the result is different from former one, relative error could be more than $10\%$ for out-of-the-money options as shown in Figure \ref{put on min calibration} which is also nonnegligible.

On the other side, for implied correlation, we can see in Figure \ref{call on min impcorr}-\ref{put on min impcorr}, the implied correlation always changes sharply for out-of-the money cases and mildly for in-the-money cases, which is similar with the calibrated $\rho$. For Call on Max options, though there are only tiny pricing errors, the implied correlations change a lot with different strikes.

\begin{figure}[htbp]
\centering
\subfigure[Relative error (Calibrated $\rho=-0.2431$)]{
\label{put on max RE t=0.5}
\begin{minipage}{7cm}
\centering
\includegraphics[width=7cm,height=6cm]{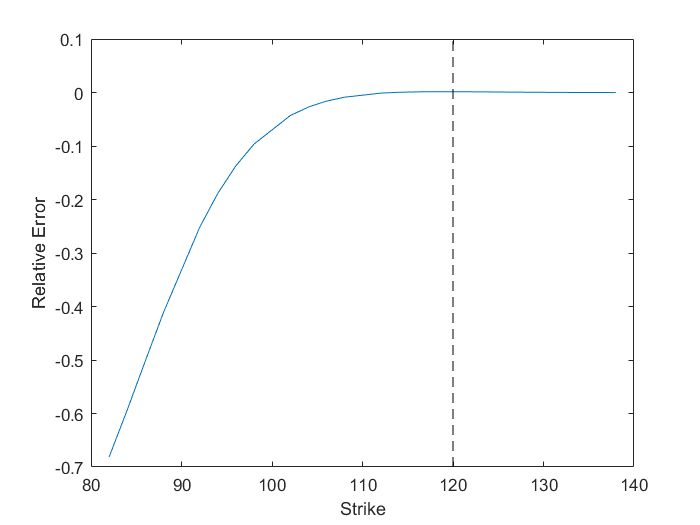}
\end{minipage}
}
\subfigure[Implied correlation]{
\label{put on max impcorr t=0.5}
\begin{minipage}{7cm}
\centering
\includegraphics[width=7cm,height=6cm]{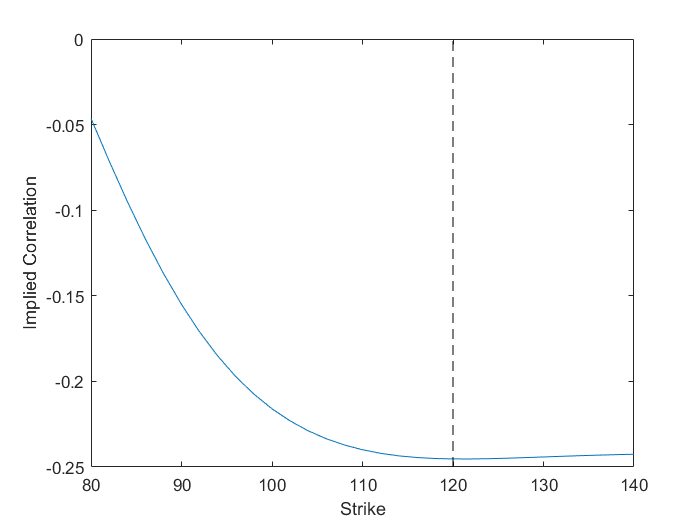}
\end{minipage}
}
\caption{Put on Max option with $\tau=0.5$}
\label{put on max calibration t=0.5}
\end{figure}

Figure \ref{put on max calibration t=0.5} investigates Put on Max options again and the maturity considered as $\tau=0.5$. Comparing with Figure \ref{put on max calibration}, we can find in Figure \ref{put on max calibration t=0.5}, the calibrated error is a little smaller and the implied correlation changes a little milder. But the main features of them are similar, this implies the maturity has little effect on our discoveries.

We will try to give a reasonable explanation for different performances of 4 kinds of rainbow options in the next section, Section \ref{Error Analysis}.  And also there in addition, we will explain why the calibrated option prices perform well for in-the-money and at-the-money options but terribly bad for deep-out-of-the-money options and why the Call on Max seems different from the other options.

\subsection{Error Analysis}\label{Error Analysis}
The pricing errors coming from setting the dynamic stochastic correlation of underlying log prices to be constant are further analyzed in this part. This analysis is from a theoretical view but with the help of numerical simulations. Through this analysis, we try to explain the phenomenon discovered in Section \ref{Numerical Experiments of Pricing Rainbow Options}.


Now we consider options with payoffs $V(S_{\tau}^1,S_{\tau}^2,\tau,K)$,\footnote{Note that all the payoffs considered in previous numerical simulations are in this way} then the price of the option is
$$E^Q\left[e^{-r\tau}V(S_{\tau}^1,S_{\tau}^2,
\tau,K)\right]=E^Q\left[e^{-r\tau}V\left(S_{0}^1e^{(r-\frac12\sigma_1^2)\tau
+\sigma_1B_\tau},S_{0}^2e^{(r-\frac12\sigma_2^2)\tau
+\sigma_2W_\tau},
\tau,K\right)\right],$$
where $Q$ denotes the risk neutral probability measure.

When the local correlation process is a constant $\rho$, since $(B_\tau,W_\tau)\sim N\left((0,0),\tau\begin{pmatrix}1&\rho\\ \rho&1\end{pmatrix}\right)$, the option price is a function of $\rho$ which will be denoted as $Price^c(\rho)$ in the following.

For more general case where $\rho$ is a stochastic process, we first recall the term of average correlation coefficient $\bar{\rho}_t=\frac1t\int_0^{t}\rho_udu$, which by the common decomposition, can be rewritten as
\begin{equation}\bar{\rho}_t=\frac{T_t-S_t}t.\label{bar rho}\end{equation}
Since under the condition of $\mathcal F^T_\tau$, $(B_\tau,W_\tau)=(X_{T_\tau}+Y_{S_\tau},X_{T_\tau}-Y_{S_\tau})\sim N\left((0,0),\tau\begin{pmatrix}1&\bar{\rho}_\tau\\ \bar{\rho}_\tau&1\end{pmatrix}\right)$, following the discussions in the constant-$\rho$ case, the option price ( denoted by $Price^d$ ) equals
$$Price^{d}=E^Q\left[E^Q[e^{-r\tau}V(S_{\tau}^1,S_{\tau}^2,\tau,K)|\mathcal F^T_{\tau}]\right]=E^Q\left[Price^c(\bar\rho_\tau)\right].$$

If $Price^c$ is an affine function of $\rho$, i.e., $\exists\, a, b\in \mathbb R$,
$Price^c(\rho)=a\rho+b,$
we have
\begin{equation}
Price^{d}=E^Q\left[Price^c(\bar\rho_{\tau})\right]
=E^Q[a\bar\rho_{\tau}+b]=Price^c(E^Q[\bar\rho_{\tau}]).\label{aaa3}
\end{equation}
In other words, when the option price under constant-$\rho$ model is linear in $\rho$, the price under a general dynamic correlation model is exactly the same as that with a constant correlation coeeficient $E^Q[\bar{\rho}_{\tau}]$.

Otherwise, for general $Price^c$, by Taylor's expansion, we can get the following approximation formula,
\begin{equation}\label{error approximation}Price^{d}=E^Q\left[Price^c(\bar\rho_{\tau})\right]\approx Price^c(E^Q[\bar\rho_{\tau}])+\frac12
\mbox{Var}^Q(\bar\rho_\tau)\frac{\partial^2Price^c}
{\partial\rho^2}(E^Q[\bar\rho_{\tau}]).\end{equation}
\eqref{aaa3} and \eqref{error approximation} indicate that the main cause of pricing errors between constant correlation model and dynamic correlation model is nonlinear property of $Price^c(\rho)$.



In the following, based on the above analysis, we try to explore causes for the big pricing errors in Section \ref{Comparing Constant and Dynamic Correlation Model in Option Pricing} and the two phenomena found in Section \ref{Calibrating Constant Correlation Model into Dynamic Correlation Model}: (\romannumeral1) the pricing errors seem more remarkable for out-the-money options when applying constant correlation model; (\romannumeral2) the pricing errors for Call-on-Max options seem relatively small than other kind of options.

We first consider relations between $Price^c(\rho)$ and $\rho$ in the cases of in-the-money, at-the-money and out-of-the-money for Put-on-Max options.
\begin{Example}\label{option price and rho}Choosing parameters as $r=0.05,\tau=0.25,S_0^1=100,S_0^2=120,\sigma_1=0.2,\sigma_2=0.3$, we draw diagrams for $Price^c(\rho)$ when $\mbox{Strike}=150$ (in the money), $\mbox{Strike}=120$ (at the money) and $\mbox{Strike}=90$ (out of the money) and list them in Figure \ref{put on max option price in constant correlation model}.

\begin{figure}[htbp]
\centering
\subfigure[Strike$=150$]{
\label{put on max option price in constant correlation model K=150}
\begin{minipage}{5cm}
\centering
\includegraphics[width=5cm,height=4cm]{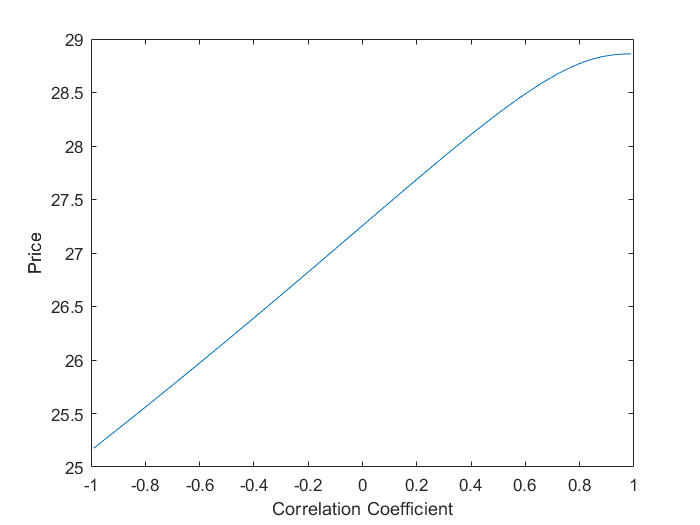}
\end{minipage}
}
\subfigure[Strike$=120$]{
\begin{minipage}{5cm}
\centering
\includegraphics[width=5cm,height=4cm]{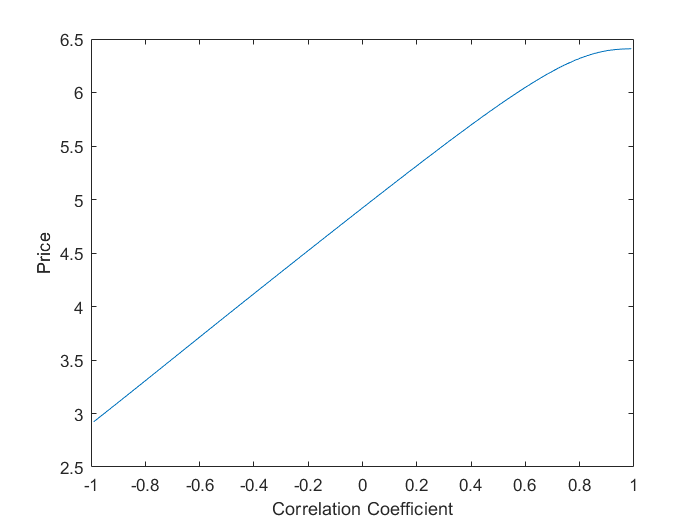}
\end{minipage}
}
\subfigure[Strike$=90$]{
\begin{minipage}{5cm}
\centering
\includegraphics[width=5cm,height=4cm]{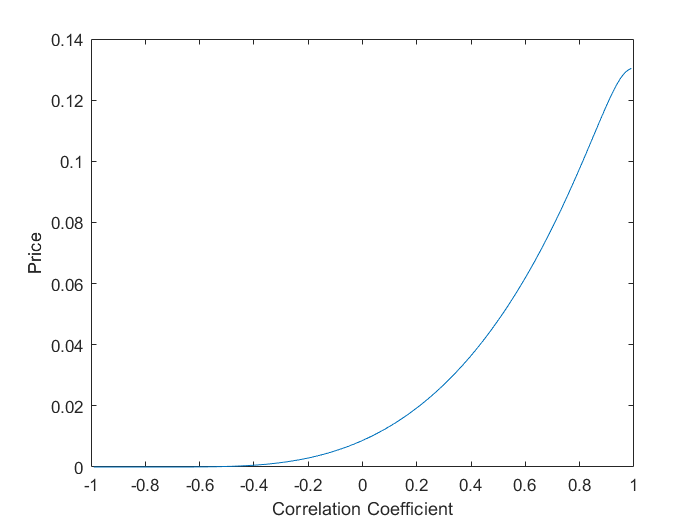}
\end{minipage}
}
\caption{Put on max option price in constant correlation model}
\label{put on max option price in constant correlation model}
\end{figure}
\end{Example}


Example \ref{option price and rho} show that, for in-the-money and at-the-money cases, $Price^c(\rho)$ reveals a strong linearity on $\rho$ except when $\rho$ is near to $1$. But it is quite nonlinear for out-of-the-money case. We conduct similar diagraming with different parameters for Put-on-Max option as well as Put-on-Min, Call-on-Min and Call-on-Max options, and get similar results. Recall the approximations \eqref{aaa3} and \eqref{error approximation}, the above results give an explanation for why constant correlation model performs well on the whole for in-the-money and at-the-money options but poorly for out-of-the-money options. We can find in Figure \ref{call on min impcorr},\ref{call on max impcorr},\ref{put on max impcorr},\ref{put on min impcorr},\ref{put on max impcorr t=0.5} and Table \ref{E & Var rho}, when strike is in-the-money and at-the-money, the implied correlation of each option is very close to $E\bar\rho_{\tau}$; on the contrary, when strike is out-of-the-money, the implied correlation changes sharply and far away from $E\bar\rho_{\tau}$. This is coincident with the conclusion in previous.
\begin{table}[htbp]
\centering
\caption{Expectation of $\bar\rho_{\tau}$}
\label{E & Var rho}
\begin{tabular}{ccc}
  \hline
   & $\tau=0.25$ & $\tau=0.5$\\
  \hline
  $\boldsymbol{\alpha}=[0.3,0.6,0.9],Q_0=[1,0,0]$    &-0.3177&   -0.2488\\
  $\boldsymbol{\alpha}=[0.3,0.6,0.95],Q_0=[0.2,0,0.8]$    &0.5784&   0.5298\\
  \hline
\end{tabular}
\end{table}

Comparing the numerical experiments in Section \ref{Comparing Constant and Dynamic Correlation Model in Option Pricing} and the data in Table \ref{E & Var rho}, we find that there are big differences between the historical local correlation coefficient and the expectation of correlation coefficient in the future, which explains the pricing errors in Section \ref{Comparing Constant and Dynamic Correlation Model in Option Pricing}.

We now turn to the Call-on-Max option whose performance in calibration in Section \ref{Calibrating Constant Correlation Model into Dynamic Correlation Model} seemed quite different from the others that the calibrated constant correlation model always performs well, even for out-of-the-money case. Note that as mentioned before, we have already got diagrams for this kind of option which have similar linear or nonlinear shapes like other options and we did not include them in the main text. A interesting question is, now that the shape of $Price^c(\rho)$ for out-of-the-money case looks apparently nonlinear, why does it still approximate the true price well? We choose the same parameters as before except for $\tau=0.25$ and draw the diagram of $Price^c(\rho)$ for Call-on-Max option for the case $\mbox{Strike}=130$ (out-of-the-money) in Figure \ref{call on max option price in constant correlation model}.
The diagram looks still quite nonlinear, but it is worth noting that in Figure \ref{call on max option price in constant correlation model} $Price^c(\rho)$ just changes from $3.97$ to $3.995$. In other words, when $\rho$ changes in its full range, the price changes only about $0.6\%$ which implies that, for Call-on-max option, the correlation between underlying assets has only a small, almost negligible, impact on the option price. While on the contrary, think about calibrating $\rho$ from option prices, a small deviation in the price may cause great changes in the implied $\rho$. This result on one hand explains why the implied correlation of Call-on-Max option is volatile but the calibrated constant correlation model always performs well and on the other hand indicates that when the data are from out-of-the-money Call-on-Max options, correlation-coefficient calibrating may be unsuitable since the implied correlation is too sensitive with the price.

\begin{figure}
  \centering
  \includegraphics[width=10cm]{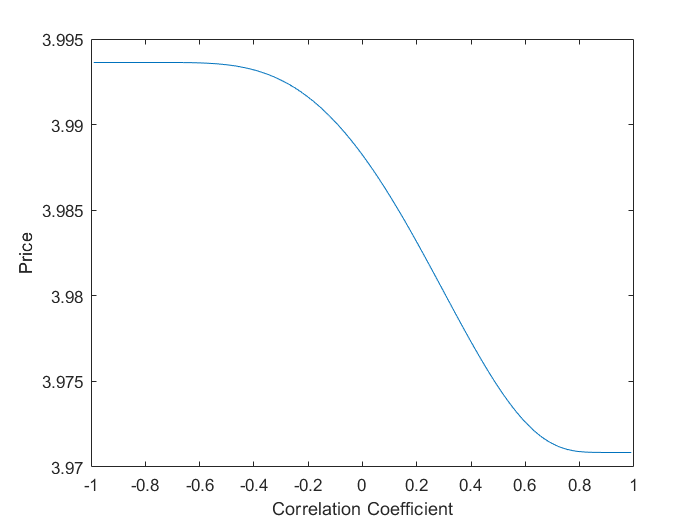}
  \caption{Call on max option price in constant correlation model with $K=130,\tau=0.25$}\label{call on max option price in constant correlation model}
\end{figure}

\section{Proofs}\label{Some Proofs}
\subsection{Preparation Works}
In the first place, we give some lemmas as preparations.

The following lemma which will be often used in Section \ref{proofs of section 2} gives a sufficient condition for a special kind of stochastic process to be a martingale.
\begin{Lemma}\label{martingale lemma}
Suppose $\{M_t\}_{t\ge0}$ is a continuous local martingale with respect to $\{\mathcal F_t\}_{t\ge0}$. If $\{\phi_t\}_{t\ge0}$ is a $\mathcal F$-progressively measurable process such that
\begin{equation}E\left[\exp\left(\frac12\int_0^t\phi_u^2d[M]_u\right)\right]<\infty,\forall t\ge0,\label{martingale lemma condition}\end{equation}
then
$$Z_t\triangleq\exp\left(\int_0^t\phi_udM_u-\frac12\int_0^t\phi_u^2d[M]_u\right),\forall t\ge0,$$
is a martingale with respect to $\mathcal F$.
\end{Lemma}
\begin{proof}
First note that $$E\left[\frac12\int_0^t\phi_u^2d[M]_u\right]<E\left[\exp\left(\frac12\int_0^t\phi_u^2d[M]_u\right)\right]<\infty,\forall t\ge0,$$
therefore $\int_0^t\phi_udM_u,\forall t\ge0$ is well-defined. By It\^o's lemma, $\{Z_t\}_{t\ge0}$ is a local martingale obviously. Hence there is a sequence of stopping times $\{\tau_n\}_{n\ge1}$ satisfy $\tau_1<\tau_2<\cdots<\tau_n<\cdots$, $\lim_{n\to\infty}\tau_n=\infty$ and
$$Z^n_t\triangleq Z_{t\wedge\tau_n},\forall t\ge0,$$
is a martingale. Consequently, $E\left[Z_t^n|\mathcal F_s\right]=Z_s^n,\forall t\ge s$. Observe that $Z^n,\forall n\ge1$ are always positive, then $E\left[Z_t|\mathcal F_s\right]\le Z_s,\forall t\ge s$ according to Fatou's lemma, i.e., $Z$ is a supermartingale. From \eqref{martingale lemma condition} and \cite{karatzas2012brownian}[Chapter 3, Proposition 5.12], we have
$$E[Z_t]=1,\forall t\ge0,$$
which implies $Z$ is a martingale immediately.
\end{proof}

Before going further, we first introduce the condition (E) as follows:
\begin{description}
\item[(E)] For any $\mathcal F^T$-progressively measurable processes $\{\phi_t^1\}_{t\ge0}$ and $\{\phi_t^2\}_{t\ge0}$ that guarantee \begin{equation}\label{condition in E}E\left[\exp\left(\frac12\int_0^{\infty}(\phi_u^1)^2dT_u+\frac12\int_0^{\infty}(\phi_u^2)^2dS_u\right)\right]<\infty,\end{equation}
    we have
$$E\left[\left.\exp\left(\int_0^{\infty}\phi_u^1dX_{T_u}+\int_0^{\infty}\phi_u^2dY_{S_u}\right)\right|\mathcal F^T_{\infty}\right]=\exp\left(\frac12\int_0^{\infty}(\phi_u^1)^2dT_u+\frac12\int_0^{\infty}(\phi_u^2)^2dS_u\right).$$
\end{description}
Next, we establish an equivalence relation between the condition (E) and the independency of $X,Y$ and $T$ by Lemma \ref{lemma of XYT independent}. Then in Section \ref{proofs of section 2} we complete the proofs of Theorem \ref{condition c1} and Proposition \ref{girsanov} through the condition (E). Besides, Lemma \ref{lemma of XYT independent} also give other two necessary conditions for the independency of $X,Y$ and $T$, which will be used in the proof of Corollary \ref{XYT independent} and Theorem \ref{condition c1} respectively.

\begin{Lemma}\label{lemma of XYT independent} Suppose $(X_t,Y_t)_{t\geq 0}$
 is a 2-dimensional standard Brownian motion and $\{T_t\}_{t\ge0}$, $\{S_t\}_{t\ge0}$
 are two increasing processes with $T_t+S_t=t,\forall t\ge0$.
 If $X$, $Y$ and $T$ are mutually independent, then we have the following consequences:
\begin{enumerate}[(\romannumeral1)]
\item the condition (E) holds;
\item $\mathcal F^{X_T}_{\infty}\perp\mathcal F^{Y_S}_{t}|\mathcal F^{X_T}_{t}$ and $\mathcal F^{X_T}_{t}\perp\mathcal F^{Y_S}_{\infty}|\mathcal F^{Y_S}_{t}$;
\item $\{X_{T_t}\}_{t\ge0}$ and $\{Y_{S_t}\}_{t\ge0}$ are martingales with respect to $\{\mathcal F^{B,W}_t\bigvee\mathcal F^T_{\infty}\}_{t\geq 0}$.
\end{enumerate}
Moreover, if $(X,Y,T)$ is a triplet of common decomposition, i.e. the conditions in Theorem \ref{decompose BM} hold, then the statement (\romannumeral1) is also sufficient for the independency property of $X$, $Y$ and $T$.
\end{Lemma}
\begin{proof}
We first prove the statements (\romannumeral1), (\romannumeral2) and (\romannumeral3).

\begin{enumerate}[(\romannumeral1)]
\item Let $\tau$ and $\varsigma$
be the inverse of $T$ and $S$ as defined in (\ref{tau,varsigma}), and $\phi^1,\phi^2$ be any progressively measurable processes satisfy \eqref{condition in E}. Define $\Phi_u^1$ and $\Phi_u^2$ as follows,
$$\Phi_u^1\triangleq\phi^1_{{\tau_u}}1_{\{u\le T_{\infty}\}},\
\Phi_u^2\triangleq\phi^2_{{\varsigma_u}}1_{\{u\le S_{\infty}\}}.$$
Then
\begin{equation}\label{integral transform 1}
\int_0^{\infty}(\Phi_u^1)^2du=\int_0^{T_{\infty}}
(\phi^1_{{\tau_u}})^2du=\int_0^{\infty}(\phi_u^1)^2dT_u,\
\int_0^{\infty}(\Phi_u^2)^2du=\int_0^{S_{\infty}}
(\phi^2_{{\varsigma_u}})^2du=\int_0^{\infty}(\phi_u^2)^2dS_u,
\end{equation}
and by \eqref{condition in E},
\begin{align*}E\left[\frac12\int_0^{\infty}(\Phi_u^1)^2du+\frac12\int_0^{\infty}(\Phi_u^2)^2du\right]=&E\left[\frac12\int_0^{\infty}(\phi_u^1)^2dT_u+\frac12\int_0^{\infty}(\phi_u^2)^2dS_u\right]\\
\le&E\exp\left(\frac12\int_0^{\infty}(\phi_u^1)^2dT_u+\frac12\int_0^{\infty}(\phi_u^2)^2dS_u\right)<\infty.
\end{align*}
Hence $\int_0^{\infty}\Phi_u^1dX_u$ and
$\int_0^{\infty}\Phi_u^2dY_u$ are well defined and
\begin{equation}\label{integral transform 2}
\int_0^{\infty}\Phi_u^1dX_u=\int_0^{T_{\infty}}\phi^1_{{\tau_u}}dX_u
=\int_0^{\infty}\phi^1_udX_{T_u},\int_0^{\infty}\Phi_u^2dY_u
=\int_0^{S_{\infty}}\phi^2_{{\varsigma_u}}dY_u
=\int_0^{\infty}\phi^2_udY_{S_u}.
\end{equation}
Observe that $X$ and $Y$ are martingales with respect to $\{\mathcal F_t^X\bigvee\mathcal F_t^Y\bigvee\mathcal F_{\infty}^T\}_{t\geq 0}$ by the independency of $X,Y,T$, and $E\exp\left(\frac12\int_0^{\infty}(\Phi_u^1)^2du+\frac12
\int_0^{\infty}(\Phi_u^2)^2du\right)<\infty$ according to \eqref{condition in E} and \eqref{integral transform 1}, thus by Lemma \ref{martingale lemma}, $$\left\{\exp\left(\int_0^{t}\Phi_u^1dX_u+\int_0^{t}
\Phi_u^2dY_u-\frac12\int_0^{t}(\Phi_u^1)^2du-\frac12
\int_0^{t}(\Phi_u^2)^2du\right)\right\}_{t\ge0}$$ is a martingale. Consequently
\begin{equation}\label{martingale}E\left[\left.\exp\left(\int_0^{\infty}\Phi_u^1dX_u+\int_0^{\infty}
\Phi_u^2dY_u-\frac12\int_0^{\infty}(\Phi_u^1)^2du-\frac12
\int_0^{\infty}(\Phi_u^2)^2du\right)\right|\mathcal F^T_{\infty}\right]=1.\end{equation}
Substituting \eqref{integral transform 1} and \eqref{integral transform 2} into \eqref{martingale}, we have
$$E\left[\left.\exp\left(\int_0^{\infty}\phi_u^1dX_{T_u}
+\int_0^{\infty}\phi_u^2dY_{S_u}-\frac12\int_0^{\infty}(\phi_u^1)^2dT_u
-\frac12\int_0^{\infty}(\phi_u^2)^2dS_u\right)\right|\mathcal F^T_{\infty}\right]=1.$$
Note that $\exp\left(\frac12\int_0^{\infty}(\phi_u^1)^2dT_u+\frac12
\int_0^{\infty}(\phi_u^2)^2dS_u\right)$ is measurable with
$\mathcal F^T_{\infty}$, and the desired result holds immediately.
\item First note that, when $X,\,Y,\,T$ are independent, by the former result, the condition (E) is true. As a direct consequence of the condition (E), $\mathcal{F}^{X_T}_t$ and $\mathcal{F}^{Y_S}_t$ is conditional independent given $\mathcal F^T_t$, $\forall t\in [0,+\infty]$. Thus for every $\mathcal{F}^{Y_S}_t$-measurable random variable $\eta$, $E[\eta|\mathcal{F}^{X_T}_t\bigvee \mathcal{F}^{T}_t]=E[\eta| \mathcal{F}^{T}_t]$. Furthermore, by the truth $\mathcal F^T_t\subset \mathcal F^{X_T}_t$, $\forall t\in [0,+\infty]$, we have
\begin{equation}\label{eta}
 E[\eta|\mathcal{F}^{T}_t]=E[\eta|\mathcal{F}^{X_T}_t],\ \forall t\in [0,+\infty].
\end{equation}

To prove the result of this part, i.e., $\mathcal{F}_t^{Y_S}$
and $\mathcal{F}^{X_T}_{\infty}$ are conditional independent given
 $\mathcal{F}_t^{X_T}$, it is sufficient to prove that 
for any $\mathcal F^T$-progressively measurable process $\phi$ satisfying \eqref{condition in E}, the following equation holds
$$E\left[\left.\exp\left(\int_0^t\phi_udY_{S_u}\right)\right|
\mathcal{F}_\infty^{X_T}\right]=E\left[\left.\exp\left(\int_0^t\phi_udY_{S_u}\right)\right|
\mathcal{F}_t^{X_T}\right].$$
By (\ref{eta}),
\begin{equation*}
E\left[\left.\exp\left(\int_0^t\phi_udY_{S_u}\right)\right|
\mathcal{F}_\infty^{X_T}\right]=E\left[\left.\exp\left(\int_0^t\phi_udY_{S_u}\right)\right|
\mathcal{F}_\infty^{T}\right]
=\exp\left(\frac12\int_0^t(\phi_u)^2d{S_u}
\right).
\end{equation*}
where the second equality comes from the condition (E) immediately. Since $\exp\left(\frac12\int_0^t(\phi_u)^2d{S_u}
\right)\in \mathcal{F}_t^{T}$, applying (\ref{eta}) again, we have
\begin{align*}
E\left[\left.\exp\left(\int_0^t\phi_udY_{S_u}\right)\right|
\mathcal{F}_t^{X_T}\right]=&E\left[\left.\exp\left(\int_0^t\phi_udY_{S_u}\right)\right|
\mathcal{F}_t^{T}\right]
=\exp\left(\frac12\int_0^t(\phi_u)^2d{S_u}
\right)\\
=&E\left[\left.\exp\left(\int_0^t\phi_udY_{S_u}\right)\right|
\mathcal{F}_\infty^{X_T}\right],
\end{align*}
which is the desired conclusion.
By similar proofs, we have $\mathcal F^{X_T}_{t}\perp\mathcal F^{Y_S}_{\infty}|\mathcal F^{Y_S}_{t}$.

\item Given $\mathcal F^T_{\infty}$, for any $n,m\in\mathbb N$ and $0\le t_1\le\dots\le t_n\le t,0\le s_1\le\dots\le s_m\le t$, we can obtain the characteristic functions of $X_{T_{u+t}}-X_{T_t}$, $\{X_{T_{t_1}},\dots,X_{T_{t_n}}\}$ and $\{Y_{S_{s_1}},\dots,Y_{S_{s_m}}\}$ respectively according to the condition (E) by some special $\phi^1$ and $\phi^2$. Besides, the condition (E) also gives the joint characteristic function of them, which implies the mutual independency of $X_{T_{u+t}}-X_{T_t}$, $\{X_{T_{t_1}},\dots,X_{T_{t_n}}\}$ and $\{Y_{S_{s_1}},\dots,Y_{S_{s_m}}\}$. By the arbitrary chosen for $t_i,1\le i\le n$ and $s_j,1\le j\le m$, we have $X_{T_{u+t}}-X_{T_t}$, $\mathcal F^{X_T}_t$ and $\mathcal F^{Y_S}_t$ are mutually independent. Hence,
    $$E\left[X_{T_{u+t}}-X_{T_t}|\mathcal F^{X_T}_t\bigvee\mathcal F^{Y_S}_{t}\bigvee\mathcal F^T_{\infty}\right]=E\left[X_{T_{u+t}}-X_{T_t}|\mathcal F^T_{\infty}\right]=0.$$
    Observe that $\mathcal F^{X_T}_t\bigvee\mathcal F^{Y_S}_{t}=\mathcal F^{B,W}$, thus $X_{T_t}$ is a martingale with $\mathcal F^{B,W}_t\bigvee \mathcal F^T_\infty$.
    The same arguments hold for $Y_{S}$.
\end{enumerate}

In the following, we prove that if the conditions in Theorem \ref{decompose BM} hold, then the condition (E) is a sufficient condition for the independency of $X$, $Y$ and $T$.

For $\forall n,m\in \mathbb N$, and $0=t_0<t_1<\dots<t_n,\ 0=s_0<s_1<\dots<s_m,$
we consider the joint distribution of $\{X_{t_1},\dots,X_{t_n},Y_{s_1},\dots,Y_{t_m}\}$
conditional on $\mathcal F^T_\infty$ by calculating
\begin{equation}\label{lapl func}
 E\left[\left.\exp(\sum_{i=1}^n\theta_i^1(X_{t_i}-X_{t_{i-1}})
 +\sum_{j=1}^m\theta_j^2(Y_{s_j}-Y_{s_{j-1}}))\right|\mathcal F^T_\infty\right],
\end{equation}
 where $\theta^1_i,\,\theta^2_j\in R,\, i=1,2,\dots,n,\ j=1,2,\dots,m$.

Define
$$\Phi_u^1=\sum_{i=1}^n\theta_i^11_{\{t_{i-1}\le u<t_i\}},\ \Phi_u^2=\sum_{j=1}^m\theta_j^21_{\{s_{j-1}\le u<s_j\}}.$$
It is easy to verify
$\int_0^{\infty}\Phi_u^1dX_u=\sum_{i=1}^n\theta_i^1(X_{t_i}-X_{t_{i-1}})$,
 $\int_0^{\infty}\Phi_u^2dY_u=\sum_{j=1}^m\theta_j^2(Y_{s_j}-Y_{s_{j-1}})$ and $$E\left[\exp\left(\int_0^{\infty}(\Phi_{T_u}^1)^2dT_u+\int_0^{\infty}(\Phi_{S_u}^2)^2dS_u\right)\right]<\infty, E\left[\exp\left(\int_0^{\infty}(\Phi_{u}^1)^2du+\int_0^{\infty}(\Phi_{u}^2)^2du\right)\right]<\infty.$$

By definitions of $X$ and $Y$ (for simplicity, we set $\int_{T_{\infty}}^{\infty}\Phi_u^1dX_u=0$ when $T_{\infty}=\infty$), we have\footnote{In the proofs of this section and Section \ref{proof of section 3}, the time-change formula for stochastic integral such as \eqref{time-change formula 1} and \eqref{time-change formula 2} will be often used. If the stochastic integral is well-defined and the integrand is progressively measurable, then the time-change formula for stochastic integral is available, in which the conditions are quite relaxed. For more details, please refer to \cite{karatzas2012brownian}[Chapter 3, Proposition 4.8] or \cite{revuz2013continuous}[Chapter V, Proposition 1.5].}
\begin{equation}\int_0^{\infty}\Phi_u^1dX_u=\int_0^{T_{\infty}}\Phi_u^1dX_u+\int_{T_{\infty}}^{\infty}\Phi_u^1dX_u=\int_0^{\infty}\Phi_{T_u}^1dX_{T_u}+\int_{0}^{\infty}\Phi_{u+T_{\infty}}^1d\tilde X_u,\label{time-change formula 1}\end{equation}
\begin{equation}\int_0^{\infty}\Phi_u^2dY_u=\int_0^{S_{\infty}}\Phi_u^2dY_u+\int_{S_{\infty}}^{\infty}\Phi_u^2dY_u=\int_0^{\infty}\Phi_{S_u}^2dY_{S_u}+\int_{0}^{\infty}\Phi_{u+S_{\infty}}^2d\tilde Y_u.\label{time-change formula 2}\end{equation}
Thus
\begin{align}
\nonumber&E\left[\left.\exp\left(\int_0^{\infty}\Phi_u^1dX_u
+\int_0^{\infty}\Phi_u^2dY_u\right)\right|\mathcal F^T_\infty\right]\\
\nonumber=&E\left[\left.\exp\left(\int_0^{\infty}\Phi_{T_u}^1
dX_{T_u}+\int_{0}^{\infty}\Phi_{u+T_{\infty}}^1d\tilde X_u+\int_0^{\infty}\Phi_{S_u}^2dY_{S_u}+\int_{0}^{\infty}
\Phi_{u+S_{\infty}}^2d\tilde Y_u\right)\right|\mathcal F^T_{\infty}\right]\\
\nonumber=&E\left[\left.E\Big{[}\exp\left(\int_0^{\infty}\Phi_{T_u}^1dX_{T_u}
+\int_{0}^{\infty}\Phi_{u+T_{\infty}}^1d\tilde X_u+\int_0^{\infty}\Phi_{S_u}^2dY_{S_u}+\int_{0}^{\infty}
\Phi_{u+S_{\infty}}^2d\tilde Y_u\right)\Big{|}
\mathcal F_{\infty}^{ X_T}\bigvee \mathcal F_{\infty}^{Y_S}\Big{]}\right|\mathcal F^T_{\infty}\right]\\
=&E\left[\left.\exp\left(\int_0^{\infty}\Phi_{T_u}^1
dX_{T_u}+\int_0^{\infty}\Phi_{S_u}^2dY_{S_u}\right)E\Big{[}\exp(\int_{0}^{\infty}\Phi_{u+T_{\infty}}^1d\tilde X_u
+\int_{0}^{\infty}\Phi_{u+S_{\infty}}^2d\tilde Y_u)\Big{|}
\mathcal F_{\infty}^{X_T}\bigvee
\mathcal F_{\infty}^{Y_S}\Big{]}\right|\mathcal F^T_{\infty}\right]
.\label{aaa}
\end{align}
It is not difficult to verify that $\left\{\int_{0}^{t}\Phi_{u+T_{\infty}}^1d\tilde X_u
+\int_{0}^{t}\Phi_{u+S_{\infty}}^2d\tilde Y_u\right\}_{t\ge0}$ is a continuous local martingale with respect to $\left\{\mathcal F^T_{\infty}\bigvee\mathcal F_{\infty}^{X_T}\bigvee\mathcal F_{\infty}^{Y_S}\bigvee\mathcal F^{\tilde X}_t\bigvee\mathcal F^{\tilde Y}_t\right\}_{t\ge0}$ and $E\left[\exp\left(\frac12\int_{0}^{\infty}(\Phi_{u+T_{\infty}}^1)^2du+\frac12\int_{0}^{\infty}(\Phi_{u+S_{\infty}}^2)^2du\right)\right]<\infty$.
Then according to Lemma \ref{martingale lemma}, $\left\{\exp\left(\int_{0}^{t}\Phi_{u+T_{\infty}}^1d\tilde X_u
+\int_{0}^{t}\Phi_{u+S_{\infty}}^2d\tilde Y_u-\frac12\int_{0}^{t}(\Phi_{u+T_{\infty}}^1)^2du
-\frac12\int_{0}^{t}(\Phi_{u+S_{\infty}}^2)^2du\right)\right\}_{t\ge0}$ is a martingale. Hence,
\begin{align}\label{ccc}
\nonumber E\Big{[}\exp(\int_{0}^{\infty}\Phi_{u+T_{\infty}}^1d\tilde X_u
&+\int_{0}^{\infty}\Phi_{u+S_{\infty}}^2d\tilde Y_u)\Big{|}
\mathcal F^T_{\infty}\bigvee\mathcal F_{\infty}^{X_T}\bigvee
\mathcal F_{\infty}^{Y_S}\Big{]}\\
&=\exp\left(\frac12\int_{0}^{\infty}(\Phi_{u+T_{\infty}}^1)^2du
+\frac12\int_{0}^{\infty}(\Phi_{u+S_{\infty}}^2)^2du\right).
\end{align}
Substituting \eqref{ccc} into \eqref{aaa}, we have
\begin{align*}
&E\left[\left.\exp\left(\int_0^{\infty}\Phi_u^1dX_u
+\int_0^{\infty}\Phi_u^2dY_u\right)\right|\mathcal F^T_\infty\right]\\
\nonumber=&\exp\left(\frac12\int_{0}^{\infty}(\Phi_{u+T_{\infty}}^1)^2du
+\frac12\int_{0}^{\infty}(\Phi_{u+S_{\infty}}^2)^2du\right)
E\left[\left.\exp\left(\int_0^{\infty}\Phi_{T_u}^1
dX_{T_u}+\int_0^{\infty}\Phi_{S_u}^2dY_{S_u}\right)
\right|\mathcal F^T_{\infty}\right].
\end{align*}
Then from the condition (E), we obtain
\begin{align*}
&E\left[\left.\exp\left(\int_0^{\infty}\Phi_u^1dX_u
+\int_0^{\infty}\Phi_u^2dY_u\right)\right|\mathcal F^T_\infty\right]\\
\nonumber=&\exp\left(\frac12\int_{0}^{\infty}(\Phi_{u+T_{\infty}}^1)^2du
+\frac12\int_{0}^{\infty}(\Phi_{u+S_{\infty}}^2)^2du\right)
\exp\left(\frac12\int_0^{\infty}(\Phi_{T_u}^1)^2dT_u+
\frac12\int_0^{\infty}(\Phi_{S_u}^2)^2dS_u\right)\\
\nonumber=& \exp\left(\frac12\{\int_0^{T_{\infty}}(\Phi_{u}^1)^2d{u}+
\int_0^{S_{\infty}}(\Phi_{u}^2)^2d{u}+\int_{T_{\infty}}^{\infty}
(\Phi_{u}^1)^2du+\int_{S_{\infty}}^{\infty}(\Phi_{u}^2)^2du\}\right)\\
=&\exp\left(\frac12\int_0^{{\infty}}(\Phi_{u}^1)^2d{u}
+\frac12\int_0^{{\infty}}(\Phi_{u}^2)^2d{u}\right).
\end{align*}
By the definition of $\Phi^1$ and $\Phi^2$, the previous equation comes to
$$E\left[\exp(\sum_{i=1}^n\theta_i^1(X_{t_i}-X_{t_{i-1}})+\sum_{j=1}^m\theta_j^2(Y_{s_j}-Y_{s_{j-1}}))|\mathcal F_{\infty}^T\right]=\exp(\sum_{i=1}^n\frac12(\theta_i^1)^2({t_k}-{t_{k-1}})+\sum_{j=1}^m\frac12(\theta_j^2)^2({s_j^2}-{s_{j-1}^2})),$$
which implies $X$ and $Y$ are independent and $\mathcal F^T_{\infty}$ does not affect the distribution of $\{X_t,Y_t\}_{t\ge0}$. Hence, $\{X_t\}_{t\ge0}$, $\{Y_t\}_{t\ge0}$ and $\{T_t\}_{t\ge0}$ are mutually independent. 
\end{proof}
Note that the condition (E) is actually equivalence with the independency of $X$, $Y$ and $T$ under the conditions in Theorem \ref{decompose BM}. 

Lemma \ref{generalized girsanov} is a generalization of Girsanov Theorem, and it may be useful in the proof of Proposition \ref{girsanov}.
\begin{Lemma}\label{generalized girsanov}Suppose $\{X_t\}_{t\ge0}$ is a Brownian motion and $\{T_t\}_{t\ge0}$ is a nondecreasing stochastic process independent with $\{X_t\}_{t\ge0}$. Given $\{\phi_t\}_{t\ge0}$ and $\{\theta_t\}_{t\ge0}$, which are progressively measurable with $\{\mathcal F^T_t\}_{t\ge0}$ and
$$E\left[\exp\left(\frac12\int_0^t(\phi_u)^2dT_u\right)\right]<\infty,E\left[\exp\left(\frac12\int_0^t(\theta_u)^2dT_u\right)\right]<\infty,\forall t\ge0,$$
let
$$X^{\phi}_t=X_t-\int_0^t\phi_{\tau_u}du,\tau_t=\inf\{u:T_t\ge u\}.$$
Then we have
$$E\left[\exp\left(\int_0^{t\wedge T_t}\theta_{\tau_u}dX^{\phi}_u\right)\exp\left(\int_0^{t\wedge T_t}\phi_{\tau_u}dX_u-\frac12\int_0^{t\wedge T_t}(\phi_{\tau_u})^2du\right)|\mathcal F^T_{\infty}\right]=\exp\left(\frac12\int_0^{t\wedge T_t}(\theta_{\tau_u})^2du\right).$$
\end{Lemma}
\begin{proof}
Given $t$, from
$$E\exp\left(\frac12\int_0^{s\wedge T_t}(\phi_{{\tau_u}})^2du\right)\le E\exp\left(\frac12\int_0^{T_t}(\phi_{\tau_u})^2du\right)
=E\exp\left(\frac12\int_0^{t}(\phi_{u})^2dT_u\right)<\infty,
$$
and Lemma \ref{martingale lemma} we have $\left\{\exp\left(\int_0^{s\wedge T_t}\phi_{{\tau_u}}dX_{u}-\frac12\int_0^{s\wedge T_t}(\phi_{{\tau_u}})^2du\right)\right\}_{s\ge0}$ is a martingale with respect to $\{\mathcal F^X_{s}\bigvee\mathcal F^T_{\infty}\}_{s\ge0}$. Let
$$\frac{d\tilde Q}{dP}\Big|\mathcal F^X_{s}\bigvee\mathcal F^T_{\infty}=\exp\left(\int_0^{s\wedge T_t}\phi_{{\tau_u}}dX_{u}-\frac12\int_0^{s\wedge T_t}(\phi_{{\tau_u}})^2du\right).$$
Note that $X$ is a Brownian motion with respect to $\{\mathcal F^X_{s}\bigvee\mathcal F^T_{\infty}\}_{s\ge0}$, then by Girsanov theorem,
$$\tilde X^{\phi}_s\triangleq X_s-\int_0^{s\wedge T_t}\phi_{{\tau_u}}du, 0\le s\le t,$$
is a Brownian motion with $\{\mathcal F^X_{s}\bigvee\mathcal F^T_{\infty}\}_{s\ge0}$ under probability measure $\tilde Q$.
Hence $\left\{\exp\left(\int_0^{s}\theta_{\tau_u}d\tilde X_u^{\phi}-\frac12\int_0^s(\theta_{\tau_u})^2du\right)\right\}_{0\le s\le t}$ is a martingale under $\tilde Q$, then by optional stopping theorem\footnote{In Girsanov theorem, we need to determine an upper bound $t$ in advance, then $\tilde X^{\phi}$ is a Brownian motion with $\{\mathcal F^X_{s}\bigvee\mathcal F^T_{\infty}\}_{0\le s\le t}$ in $[0,t]$. Thanks to $0\le t\wedge T_t\le t$, optional stopping theorem for $t\wedge T_t$ remains valid.}, we obtain
$$E^{\tilde Q}\left[\exp\left(\int_0^{t\wedge T_t}\theta_{\tau_u}d\tilde X_u^{\phi}-\frac12\int_0^{t\wedge T_t}(\theta_{\tau_u})^2du\right)|\mathcal F^T_{\infty}\right]=1,$$
i.e.,
$$E^P\left[\exp\left(\int_0^{t\wedge T_t}\theta_{\tau_u}d\tilde X_u^{\phi}\right)\exp\left(\int_0^{t\wedge T_t}\phi_{{\tau_u}}dX_{u}-\frac12\int_0^{t\wedge T_t}(\phi_{{\tau_u}})^2du\right)|\mathcal F^T_{\infty}\right]=\exp\left(\frac12\int_0^{t\wedge T_t}(\theta_{\tau_u})^2du\right).$$
Note that $\tilde X_s^{\phi}=X_s^{\phi},\forall s\in[0,t\wedge T_t]$, thus we get desired result immediately.
\end{proof}
\subsection{Proofs of Results in Section \ref{Dependency Structure of Two Correlated Brownian Motions}}\label{proofs of section 2}
First we prove Theorem \ref{decompose BM}.

\begin{proofoftheorem}{\ref{decompose BM}}
We prove (\romannumeral1) first. Note that $\frac{B+W}2$ and $\frac{B-W}2$ are continuous martingales and
$$\left[X_T,Y_S\right]_t=\left[\frac{B+W}2,\frac{B-W}2\right]_t=\frac1{16}([2B,2B]_t-[2W,2W]_t)=0.$$
By the definitions of $\tau$ and $\varsigma$,
$${\tau_t}=\inf\left\{u:\left[\frac{B+W}2,\frac{B+W}2\right]_u>t\right\},\ {\varsigma_t}=\inf\left\{u:\left[\frac{B-W}2,\frac{B-W}2\right]_u>t\right\}.$$
Then according to \cite{revuz2013continuous}[Chapter V, Theorem 1.10], $\{X_t\}_{t\ge0}$ and $\{Y_t\}_{t\ge0}$ are two independent Brownian motions.

As for (\romannumeral2), \eqref{[B,W]_t-[B,W]_s} implies $[B,W]$ is absolutely continuous with respect to $t$, hence is derivable. Then \eqref{definition of T} leads to the result immediately.
\end{proofoftheorem}

Next, we prove Theorem \ref{condition c1} through Lemma \ref{lemma of XYT independent}.

\begin{proofoftheorem}{\ref{condition c1}}
For the ``if" part: since $\mathcal F^B_{\infty}\perp\mathcal F^T_{\infty}|\mathcal F_t^{B,W}$ and $\mathcal F_t^{B,W}\subset \mathcal F_t$,
\begin{equation}E\left[B_t-B_s|\mathcal F^T_{\infty}\bigvee\mathcal F^{B,W}_s\right]=E\left[B_t-B_s|\mathcal F^{B,W}_s\right]=0,\label{B_t-B_s}
\end{equation}
therefore the process $B$ is a martingale with respect to $\{\mathcal F_t^{B,W}\bigvee\mathcal F^T_{\infty}\}_{t\geq 0}$, so is the process $W$ by similar analysis.

As a consequence,
$X_{T_t}=\frac{B_t+W_t}2$ and $Y_{S_t}=\frac{B_t-W_t}2$ are martingales with respect to the same filtration. So for any $\mathcal F^T$-progressively measurable processes $\phi^1$, $\phi^2$ satisfying \eqref{condition in E},
$$D_t^{\phi}\triangleq\exp\left(\int_0^{t}\phi_u^1dX_{T_u}+\int_0^{t}
\phi_u^2dY_{S_u}-\frac12\int_0^{t}(\phi_u^1)^2dT_u-\frac12\int_0^{t}
(\phi_u^2)^2dS_u\right),\ t\in [0,+\infty)$$
is a martingale with respect to $\{\mathcal F_t^{B,W}\bigvee\mathcal F^T_{\infty}\}_{t\geq 0}$ by Lemma \ref{martingale lemma}. 
Moreover,
\begin{align*}E\left[\left(\frac12\int_0^{\infty}(\phi_u^1)^2dT_u+\frac12\int_0^{\infty}(\phi_u^2)^2dS_u\right)\right]<E\left[\exp\left(\frac12\int_0^{\infty}(\phi_u^1)^2dT_u+\frac12\int_0^{\infty}(\phi_u^2)^2dS_u\right)\right]<\infty\end{align*}
implies $D^{\phi}_{\infty}$ exists. Thus
$$E[D^{\phi}_{\infty}|\mathcal F_0^{B,W}\bigvee\mathcal F_{\infty}^T]=D^{\phi}_0=1,$$
i.e.
$$E\left[\exp\left(\int_0^{\infty}\phi_u^1dX_{T_u}+\int_0^{\infty}\phi_u^2dY_{S_u}\right)|\mathcal F^T_{\infty}\right]=\exp\left(\frac12\int_0^{\infty}(\phi_u^1)^2dT_u+\frac12\int_0^{\infty}(\phi_u^2)^2dS_u\right).$$
According to Lemma \ref{lemma of XYT independent}, the desired result is obtained.


For the``only if" part: if $\{X_t\}_{t\ge0}$, $\{Y_t\}_{t\ge0}$ and $\{T_t\}_{t\ge0}$ are independent, by Lemma \ref{lemma of XYT independent}, $\{X_{T_t}\}_{t\ge0}$ and $\{Y_{S_t}\}_{t\ge0}$ are martingales with respect to $\mathcal F^{B,W}_t\bigvee\mathcal F^T_{\infty}$.

Consequently, $B_t=X_{T_t}+Y_{S_t},W_t=X_{T_t}-Y_{S_t},$ are martingales with respect to $\mathcal F^{B,W}_t\bigvee\mathcal F^T_{\infty}$. Since $[B,B]_t=[W,W]_t=t$, $B_t$ and $W_t$ are Brownian motions with respect to $\mathcal F^{B,W}_t\bigvee\mathcal F^T_{\infty}$ according to L\'evy characterisation.

On the other hand, $B_t$ and $W_t$ are Brownian motions with respect to $\mathcal F^{B,W}_t$ as well. That is to say, for any $t\geq 0$, the conditional distribution of the process $B$ given $\mathcal F_t^{B,W}\bigvee\mathcal F^T_{\infty}$ is coincident with its conditional distribution given $\mathcal F_t^{B,W}$. Then we can conclude that $\mathcal F^B_{\infty}\perp\mathcal F^T_{\infty}|\mathcal F_t^{B,W}$. Similarly, $\mathcal F^W_{\infty}\perp\mathcal F^T_{\infty}|\mathcal F_t^{B,W}$.
%
%
\end{proofoftheorem}

In the following, we complete the proof of Proposition \ref{girsanov}.

\begin{proofofproposition}{\ref{girsanov}}We prove that the independency of $X$, $Y$ and $T$ is equivalent with the condition (C2), then from Theorem \ref{condition c1}, we have the condition (C1) is equivalent with the condition (C2).


%

For the "$\Rightarrow$" part: It is obvious that $D_t^{\phi}$ is a martingale from Lemma \ref{martingale lemma}.

Suppose $\theta^i_t,i=1,2$ are bounded determined processes, then
\begin{align}
\nonumber &E^Q\left[\exp\left(\int_0^t\theta^1_udX^{\phi}_{T_u}+\int_0^t\theta^2_udY^{\phi}_{S_u}\right)\right]=E^P\left[\exp\left(\int_0^t\theta^1_udX^{\phi}_{T_u}+\int_0^t\theta^2_udY^{\phi}_{S_u}\right)D_t^{\phi}\right]\\
\nonumber=&E^P\left[\exp\left(\int_0^t\theta^1_udX^{\phi}_{T_u}\right)\exp\left(\int_0^t\phi_u^1dX_{T_u}-\frac12\int_0^t(\phi_u^1)^2dT_u\right)\right.\\
&\left.E^P\left[\exp\left(\int_0^t\theta^2_udY^{\phi}_{S_u}\right)\exp\left(\int_0^t\phi_u^2dY_{S_u}-\frac12\int_0^t(\phi_u^2)^2dS_u\right)|\mathcal F^T_{\infty}\bigvee\mathcal F_{\infty}^X\right]\right].\label{aaa1}
\end{align}
According to the independency of $X,Y,T$, we have
\begin{align}\nonumber&E^P\left[\exp\left(\int_0^t\theta^2_udY^{\phi}_{S_u}\right)\exp\left(\int_0^t\phi_u^2dY_{S_u}-\frac12\int_0^t(\phi_u^2)^2dS_u\right)|\mathcal F^T_{\infty}\bigvee\mathcal F_{\infty}^X\right]\\
\nonumber=&E^P\left[\exp\left(\int_0^t\theta^2_udY^{\phi}_{S_u}\right)\exp\left(\int_0^t\phi_u^2dY_{S_u}-\frac12\int_0^t(\phi_u^2)^2dS_u\right)|\mathcal F^T_{\infty}\right]\\
=&E^P\left[\exp\left(\int_0^{S_t}\theta^2_{\varsigma_u}dY_u^{\phi}\right)\exp\left(\int_0^{S_t}\phi_{{\varsigma_u}}^2dY_{u}-\frac12\int_0^{S_t}(\phi_{{\varsigma_u}}^2)^2du\right)|\mathcal F^T_{\infty}\right],
\label{bbb1}\end{align}
where $Y^{\phi}_t=Y_{t}-\int_0^{{\varsigma_t}}\phi_u^2dS_u=Y_t-\int_0^t\phi_{{\varsigma_u}}^2du$.
Observe that $t\wedge S_t=S_t$, then from Lemma \ref{generalized girsanov} we have
\begin{align}
\nonumber&E^P\left[\exp\left(\int_0^{S_t}\theta^2_{\varsigma_u}dY_u^{\phi}\right)\exp\left(\int_0^{S_t}\phi_{{\varsigma_u}}^2dY_{u}-\frac12\int_0^{S_t}(\phi_{{\varsigma_u}}^2)^2du\right)|\mathcal F^T_{\infty}\right]\\
\nonumber=&\exp\left(\frac12\int_0^{S_t}(\theta^2_{\varsigma_u})^2du\right)
=E^P\left[\exp\left(\int_0^{S_t}\theta^2_{\varsigma_u}dY_u\right)|\mathcal F^T_{\infty}\right]
=E^P\left[\exp\left(\int_0^{S_t}\theta^2_{\varsigma_u}dY_u\right)|\mathcal F^T_{\infty}\bigvee\mathcal F^X_{\infty}\right]\\
=&E^P\left[\exp\left(\int_0^{t}\theta^2_{u}dY_{S_u}\right)|\mathcal F^T_{\infty}\bigvee\mathcal F^X_{\infty}\right].
\label{ccc1}\end{align}
Substituting \eqref{bbb1} and \eqref{ccc1} into \eqref{aaa1},
\begin{align*}
&E^Q\left[\exp\left(\int_0^t\theta^1_udX^{\phi}_{T_u}+\int_0^t\theta^2_udY^{\phi}_{S_u}\right)\right]\\
=&E^P\left[\exp\left(\int_0^t\theta^1_udX^{\phi}_{T_u}\right)\exp\left(\int_0^t\phi_u^1dX_{T_u}-\frac12\int_0^t(\phi_u^1)^2dT_u\right)E^P[\exp\left(\int_0^{t}\theta^2_{u}dY_{S_u}\right)|\mathcal F^T_{\infty}\bigvee\mathcal F^X_{\infty}]\right]\\
=&E^P\left[\exp\left(\int_0^t\theta^1_udX^{\phi}_{T_u}\right)\exp\left(\int_0^t\phi_u^1dX_{T_u}-\frac12\int_0^t(\phi_u^1)^2dT_u\right)\exp\left(\int_0^{t}\theta^2_{u}dY_{S_u}\right)\right]\\
=&E^P\left[\exp\left(\int_0^{t}\theta^2_{u}dY_{S_u}\right)E^P[\exp\left(\int_0^t\theta^1_udX^{\phi}_{T_u}\right)\exp\left(\int_0^t\phi_u^1dX_{T_u}-\frac12\int_0^t(\phi_u^1)^2dT_u\right)|\mathcal F^T_{\infty}\bigvee\mathcal F_{\infty}^Y]\right].
\end{align*}
Applying Lemma \ref{generalized girsanov} to the former equation again, we obtain
\begin{align*}E^Q\left[\exp\left(\int_0^t\theta^1_udX^{\phi}_{T_u}+\int_0^t\theta^2_udY^{\phi}_{S_u}\right)\right]
=&E^P\left[\exp\left(\int_0^{t}\theta^2_{u}dY_{S_u}\right)E^P[\exp\left(\int_0^t\theta^1_udX_{T_u}\right)|\mathcal F^T_{\infty}\bigvee\mathcal F_{\infty}^Y]\right]\\
=&E^P\left[\exp\left(\int_0^t\theta^1_udX_{T_u}+\int_0^{t}\theta^2_{u}dY_{S_u}\right)\right].
\end{align*}
If $\theta_t^i,i=1,2$ are complex, the proof remains valid, hence we have $(\tilde X^{\phi},\tilde Y^{\phi})_{Q}\overset{d}{=}(X_{T},Y_{S})_P$ immediately.

For the "$\Leftarrow$" part: Suppose $\{\phi^1_t\}_{t\ge0}$ and $\{\phi^2_t\}_{t\ge0}$ satisfy \eqref{condition in E}. Note that the range of $\{\phi^1_t\}_{t\ge0}$ and $\{\phi^2_t\}_{t\ge0}$ in \eqref{condition in E} is smaller than the condition (C2), then
\begin{align*}E\left[\left(\frac12\int_0^{\infty}(\phi_u^1)^2dT_u+\frac12\int_0^{\infty}(\phi_u^2)^2dS_u\right)\right]
 \le E\left[\exp\left(\frac12\int_0^{\infty}(\phi_u^1)^2dT_u+\frac12\int_0^{\infty}(\phi_u^2)^2dS_u\right)\right]<\infty,\end{align*}
accordingly $D^{\phi}_{\infty}$ exists. We first claim that
$$E^P[D^{\phi}_{\infty}|\mathcal F^T_{\infty}]=1\quad a.s..$$
To see this, we only need to prove for any $A\in\mathcal F^T_{\infty}$,
\begin{equation}E^P[D^{\phi}_{\infty}1_A]=P(A).\label{sufficient condition for conditional expectation}\end{equation}
Let
$$\mathcal D\triangleq\{A\in\mathcal F|E^P[D^{\phi}_{\infty}1_A]=P(A)\},\quad \mathcal P\triangleq\{\bigcap_{i=1}^nA_{t_i}|A_{t_i}\in\sigma(T_{t_i}),n\ge1,\forall t_1<t_2<\dots<t_n\},$$
note that $E^P[D^{\phi}_{\infty}]=1$, so $\mathcal D$ is a $\lambda$-system and obviously $\mathcal P$ is a $\pi$-system, moreover, $\sigma(\mathcal P)=\mathcal F^T_{\infty}$. Suppose $A_{t_i}=\{T_{t_i}\in\mathcal B_i\}$, where $\mathcal B_i$ is a Borel set, then for any $A=\bigcap_{i=1}^nA_{t_i}\in\mathcal P$ we have
\begin{equation}\label{D_infty^phi}E^P[D^{\phi}_{\infty}1_A]=E^P[1_AE^P[D^{\phi}_{\infty}|\mathcal{F}_{t_n}]]=E^P[D^{\phi}_{t_n}1_A].\end{equation}
Since $(\tilde X^{\phi},\tilde Y^{\phi})_{Q}\overset{d}{=}(X_{T},Y_{S})_P,$ and
$$[X_T]_t=X_{T_t}^2-\int_0^tX_{T_u}dX_{T_u},$$
so we have $([\tilde X^{\phi}],[\tilde Y^{\phi}])_{Q}\overset{d}{=}([X_{T}],[Y_{S}])_P,$ i.e., $(T,S)_{Q}\overset{d}{=}(T,S)_P$. Consequently,
\begin{equation}\label{P(A)}P(A)=P(T_{t_i}\in\mathcal B_i,i=1,2,\dots,n)=Q(T_{t_i}\in\mathcal B_i,i=1,2,\dots,n)=E^Q[1_A]=E^P[D^{\phi}_{t_n}1_A].\end{equation}
From \eqref{D_infty^phi} and \eqref{P(A)} we know that $\mathcal P\subset\mathcal D$. According to $\pi-\lambda$ theorem we can conclude
$$\mathcal F^T_{\infty}=\sigma(\mathcal P)\subset\mathcal D,$$
hence, we have proved our claim \eqref{sufficient condition for conditional expectation}. $E^P[D^{\phi}_{\infty}|\mathcal F^T_{\infty}]=1$ implies
$$E\left[\exp\left(\int_0^{\infty}\phi_u^1dX_{T_u}+\int_0^{\infty}\phi_u^2dY_{S_u}\right)|\mathcal F^T_{\infty}\right]=\exp\left(\frac12\int_0^{\infty}(\phi_u^1)^2dT_u+\frac12\int_0^{\infty}(\phi_u^2)^2dS_u\right),$$
we complete proof by Lemma \ref{lemma of XYT independent}.

\end{proofofproposition}

We prove Proposition \ref{independent local correlation imply independent time change} by the equivalence of the condition (C3) and the condition (C1).

\begin{proofofproposition}{\ref{independent local correlation imply independent time change}}

"(C1)$\Rightarrow$(C3)": According to $\mathcal F^B_{\infty}\perp\mathcal F^T_{\infty}|\mathcal F_t^{B,W}$ and \eqref{B_t-B_s}, we have $\{B_t\}_{t\ge0}$ is a martingale with respect to $\mathcal F^{B,W}_t\bigvee\mathcal F^T_{\infty}$. Because $\mathcal F_{\infty}^{\tilde Z}\perp\mathcal F^{B,W}_{\infty}\bigvee\mathcal F^T_{\infty}$ (actually, $\mathcal F^T_{\infty}\subset\mathcal F^{B,W}_{\infty}$), then for any $\xi\in\mathcal F_t^{\tilde Z}$,
$$E\left[\xi\Big|\mathcal F^{B,W}_{\infty}\bigvee\mathcal F^T_{\infty}\right]=E[\xi]=E\left[\xi\Big|\mathcal F^{B,W}_{t}\bigvee\mathcal F^T_{\infty}\right],$$
which is equivalent with $\mathcal F_t^{\tilde Z}\perp\mathcal F_{\infty}^{B,W}|\mathcal F^{B,W}_t\bigvee\mathcal F^T_{\infty}$. Hence,
$$E\left[B_t-B_s\Big|\mathcal F^{B,W}_s\bigvee\mathcal F^T_{\infty}\bigvee\mathcal F_s^{\tilde Z}\right]=E\left[B_t-B_s\Big|\mathcal F^{B,W}_s\bigvee\mathcal F^T_{\infty}\right]=0,$$
and equivalently, $\{B_t\}_{t\ge0}$ is a martingale with respect to $\{\mathcal F^{B,W}_t\bigvee\mathcal F^T_{\infty}\bigvee\mathcal F_t^{\tilde Z}\}_{t\ge0}$. With the same arguments, $\{W_t\}_{t\ge0}$ is a martingale with respect to $\mathcal F^{B,W}_t\bigvee\mathcal F^T_{\infty}\bigvee\mathcal F_t^{\tilde Z}$ as well. Obviously, $\{\tilde Z_t\}_{t\ge0}$ is a martingale with respect to $\mathcal F^{B,W}_t\bigvee\mathcal F^T_{\infty}\bigvee\mathcal F_t^{\tilde Z}$, so from the definition of $Z_t$, we know $\{Z_t\}_{t\ge0}$ is a martingale with respect to $\mathcal F^{B,W}_t\bigvee\mathcal F^T_{\infty}\bigvee\mathcal F_t^{\tilde Z}$ and $[Z]_t=t,[B,Z]_t=0$. According to L\'evy characterisation (see \cite{Shreve2004Stochastic}[Theorem 4.6.4]), $\{B_t\}_{t\ge0}$ and $\{Z_t\}_{t\ge0}$ are two independent Brownian motions with respect to $\mathcal F^{B,W}_t\bigvee\mathcal F^T_{\infty}\bigvee\mathcal F_t^{\tilde Z}$. Since $\{B_t\}_{t\ge0}$ and $\{Z_t\}_{t\ge0}$ are adapted with $\mathcal F^{B,W}_t\bigvee\mathcal F_t^{\tilde Z}\subset\mathcal F^{B,W}_t\bigvee\mathcal F^T_{\infty}\bigvee\mathcal F_t^{\tilde Z}$, so $\{B_t\}_{t\ge0}$ and $\{Z_t\}_{t\ge0}$ are also two independent Brownian motions with respect to $\mathcal F^{B,W}_t\bigvee\mathcal F_t^{\tilde Z}$. Consequently, the joint distribution of $\{B_t\}_{t\ge0}$ and $\{Z_t\}_{t\ge0}$ is the same under the condition of $\mathcal F_t^{B,W}\bigvee\mathcal F^T_{\infty}\bigvee\mathcal F^{\tilde Z}_{t}$ and $\mathcal F_t^{B,W}\bigvee\mathcal F^{\tilde Z}_{t}$, which implies
\begin{equation}\mathcal F^Z_{\infty}\bigvee\mathcal F^B_{\infty}\perp\mathcal F^T_{\infty}|\mathcal F^{B,W}_t\bigvee\mathcal F^{\tilde Z}_{t}.\label{conditional independent}\end{equation}
In \eqref{conditional independent}, let $t=0$ we obtain $\mathcal F^Z_{\infty}\bigvee\mathcal F^B_{\infty}\perp\mathcal F^T_{\infty}$. Note that $\{B_t\}_{t\ge0}$ is also independent with $\{Z_t\}_{t\ge0}$, hence we can conclude that $\{B_t\}_{t\ge0}$, $\{Z_t\}_{t\ge0}$ and $\{\rho_t\}_{t\ge0}$ are mutually independent.

"(C3)$\Rightarrow$(C1)": Note that $\mathcal F^{B,W}_t\subset\mathcal F_t^B\bigvee\mathcal F_t^Z\bigvee\mathcal F_t^T$, and obviously $\mathcal F^B_{\infty}$, $\mathcal F^T_{\infty}$ and $\mathcal F^Z_t$ are mutually independent given $\mathcal F^B_t$ by the condition (C3). Then for any $\xi\in\mathcal F^B_{\infty}$,
$$E\left[\xi\Big|\mathcal F_t^{B,W}\right]=E\left[E[\xi|\mathcal F_t^B\bigvee\mathcal F_t^Z\bigvee\mathcal F^T_t]\Big|\mathcal F_t^{B,W}\right]=E\left[E[\xi|\mathcal F_t^B]\Big|\mathcal F_t^{B,W}\right]=E\left[\xi\Big|\mathcal F_t^B\right],$$
with similar approach we can prove $E\left[\xi\Big|\mathcal F^{B,W}_t\bigvee\mathcal F^T_{\infty}\right]=E\left[\xi\big|\mathcal F_t^B\right]$ as well, immediately
$$E\left[\xi\Big|\mathcal F^{B,W}_t\right]=E\left[\xi\Big|\mathcal F^{B,W}_t\bigvee\mathcal F^T_{\infty}\right],\forall\xi\in\mathcal F^B_{\infty},$$
which is equivalent to $\mathcal F^B_{\infty}\perp\mathcal F^T_{\infty}|\mathcal F_t^{B,W}$.

As for $\mathcal F^W_{\infty}\perp\mathcal F^T_{\infty}|\mathcal F^{B,W}_t$, we first observe that $\{B_t\}_{t\ge0}$ and $\{Z_t\}_{t\ge0}$ are martingales with respect to $\mathcal F^{B,Z}_t\bigvee\mathcal F^T_{\infty}$ by the independecy of $\{\rho_t\}_{t\ge0}$, $\{B_t\}_{t\ge0}$ and $\{Z_t\}_{t\ge0}$. So according to
$$W_t=\int_0^t\rho_sdB_s+\int_0^t\sqrt{1-\rho_s^2}dZ_s,$$
$\{W_t\}_{t\ge0}$ is a martingale with respect to $\mathcal F^{B,Z}_t\bigvee\mathcal F^T_{\infty}$. Since $\mathcal F_t^{B,W}\subset\mathcal F^{B,Z}_t\bigvee\mathcal F^T_{\infty}$ and $\mathcal F_t^{B,W}\bigvee\mathcal F^T_{\infty}\subset\mathcal F^{B,Z}_t\bigvee\mathcal F^T_{\infty}$, and note that $\{W_t\}_{t\ge0}$ is adapted to $\mathcal F_t^{B,W}$ and $\mathcal F_t^{B,W}\bigvee\mathcal F^T_{\infty}$ respectively, so $\{W_t\}_{t\ge0}$ is a martingale with respect to $\mathcal F_t^{B,W}$ and $\mathcal F_t^{B,W}\bigvee\mathcal F^T_{\infty}$ respectively. Hence, by L\'evy characterisation, $\{W_t\}_{t\ge0}$ is a Brownian motion with respect to $\mathcal F_t^{B,W}$ and $\mathcal F_t^{B,W}\bigvee\mathcal F^T_{\infty}$ respectively. Thus, the distribution of $\{W_t\}_{t\ge0}$ is the same under the condition of $\mathcal F_t^{B,W}\bigvee\mathcal F^T_{\infty}$ and $\mathcal F_t^{B,W}$, which result in $\mathcal F^W_{\infty}\perp\mathcal F^T_{\infty}|\mathcal F_t^{B,W}$.

\end{proofofproposition}

We prove Proposition \ref{convergence in distribution} by comparing the distribution of discrete local correlation model and discrete common decomposition model.

\begin{proofofproposition}{\ref{convergence in distribution}}Given the partition $\Pi$, let
\begin{align*}\rho_u^\Pi\triangleq&\rho_{t_i},\ t_i\le u<t_{i+1},\\
W_s^{\Pi}\triangleq&\int_0^s\rho_u^{\Pi}dB_u+\int_0^s\sqrt{1-(\rho_u^{\Pi})^2}dZ_u.\end{align*}
Then
$$W_s^{\Pi}=\sum_{k=0}^i(\rho_{t_k}\Delta B_{t_k}+\sqrt{1-\rho_{t_k}^2}\Delta Z_{t_k})+\rho_{t_i}(B_s-B_{t_i})+\sqrt{1-\rho_{t_i}^2}(Z_s-Z_{t_i}),t_i\le s<t_{i+1},$$
where $\Delta B_{t_k}=B_{t_{k+1}}-B_{t_k},
\Delta Z_{t_k}=Z_{t_{k+1}}-Z_{t_k}$.

Observe that given $\mathcal F^T_{\infty}$, the conditional distribution of $(\Delta B_{t_i},
\Delta W^\Pi_{t_i})$ is
$$(\Delta B_{t_i},\Delta W^\Pi_{t_i})\sim N\left(\begin{pmatrix}0\\0\end{pmatrix},\begin{pmatrix}\Delta t_i&\rho_{t_i}\Delta t_i\\
 \rho_{t_i}\Delta t_i&\Delta t_i\end{pmatrix}\right).$$
 which is just the same as the conditional distribution of
  $(\Delta B_{t_i},\Delta\tilde W^{\Pi}_{t_{i}})$.
If the condition (C3) holds, $\{B_t\}_{t\ge0}$ and $\{Z_t\}_{t\ge0}$ are independent Brownian motions with respect to $\mathcal F^{B,Z}_t\bigvee\mathcal F^T_{\infty}$. Hence, by the independent property of increments, given $\mathcal F^T_{\infty}$, we have
$$(\Delta B_{t_0},\Delta B_{t_1},\dots,\Delta B_{t_{n-1}},\Delta W_{t_0}^\Pi,\Delta W_{t_1}^\Pi,\dots,\Delta W_{t_{n-1}}^\Pi)\overset{d}{=}(\Delta B_{t_0},\Delta B_{t_1},\dots,\Delta B_{t_{n-1}},\Delta\tilde W_{t_0}^\Pi,\Delta\tilde W_{t_1}^\Pi,\dots,\Delta\tilde W_{t_{n-1}}^\Pi).$$
Consequently,
\begin{equation}(B_{t_0},B_{t_1},\dots,B_{t_{n-1}},W_{t_0}^\Pi,W_{t_1}^\Pi,\dots,W_{t_{n-1}}^\Pi)\overset{d}{=}(B_{t_0},B_{t_1},\dots,B_{t_{n-1}},\tilde W_{t_0}^\Pi,\tilde W_{t_1}^\Pi,\dots,\tilde W_{t_{n-1}}^\Pi).\label{identical distribution of W^Pi}\end{equation}
Next, for any $K,L\in\mathbb N$ given $u_k,v_l,k=1,2,\dots,K,l=1,2,\dots,L$, we consider the difference between the distribution of $(B_{u_1},B_{u_2},\dots,B_{u_K},W^\Pi_{v_1},W^\Pi_{v_2},\dots,W^\Pi_{v_L})$ and $(B_{u_1},B_{u_2},\dots,B_{u_K},\tilde W^\Pi_{v_1},\tilde W^\Pi_{v_2},\dots,\tilde W^\Pi_{v_L})$. Let
$$i_k=\sup\{z\in\mathbb Z:t_z<u_k\},j_l=\sup\{z\in\mathbb Z:t_z<v_l\},k=1,2,\dots,K,l=1,2,\dots,L.$$
For any $\epsilon>0$, we first give a $\delta$ small enough\footnote{Since we focus on the properties when $||\Pi||\to0$, we can only consider the case that $||\Pi||<\min(u_1,v_1)/2$. Then $t_{i_k}>u_1/2$, $t_{j_l}>v_1/2$, $\forall k,l$, hence there always exists a $\delta$ satisfy the condition.} such that for any $a_k,k=1,2,\dots,K$ and $b_l,l=1,2,\dots,L$,
\begin{equation}\sum_{k=1}^KP(|B_{t_{i_k}}-a_k|\le\delta)+\sum_{l=1}^LP(|W^\Pi_{t_{j_l}}-b_l|\le\delta)<\frac{\epsilon}2.\label{delta small enough}\end{equation}
Observe that
\begin{align*}P(B_{t_{i_1}}\le a_1&-\delta,\dots,B_{t_{i_K}}\le a_K-\delta,W^\Pi_{t_{j_1}}\le b_1-\delta,\dots,W^\Pi_{t_{j_L}}\le b_L-\delta,\\
&B_{u_k}-B_{t_{i_k}}\le\delta,W^\Pi_{v_l}-W^\Pi_{t_{j_l}}\le\delta,k=1,\dots,K,l=1,\dots,L)\\
\le P(B_{u_1}\le a_1&,\dots,B_{u_K}\le a_K,W^\Pi_{v_1}\le b_1,\dots,W^\Pi_{v_L}\le b_L)\\
\le P(B_{t_{i_1}}\le a_1&+\delta,\dots,B_{t_{i_K}}\le a_K+\delta,W^\Pi_{t_{j_1}}\le b_1+\delta,\dots,W^\Pi_{t_{j_L}}\le b_L+\delta)\\
&+\sum_{k=1}^KP(B_{u_k}-B_{t_{i_k}}\le-\delta)+\sum_{l=1}^LP(W^\Pi_{v_l}-W^\Pi_{t_{j_l}}\le-\delta),\end{align*}
similar inequality holds for $P(B_{u_1}\le a_1,\dots,B_{u_K}\le a_K,\tilde W^\Pi_{v_1}\le b_1,\dots,\tilde W^\Pi_{v_L}\le b_L)$. Let $$H=\{B_{u_k}-B_{t_{i_k}}\le\delta,\tilde W^\Pi_{v_l}-\tilde W^\Pi_{t_{j_l}}\le\delta,k=1,\dots,K,l=1,\dots,L\},$$ then
\begin{align}\nonumber&P(B_{u_1}\le a_1,\dots,B_{u_K}\le a_K,W^\Pi_{v_1}\le b_1,\dots,W^\Pi_{v_L}\le b_L)-P(B_{u_1}\le a_1,\dots,B_{u_K}\le a_K,\tilde W^\Pi_{v_1}\le b_1,\dots,\tilde W^\Pi_{v_L}\le b_L)\\
\nonumber\le&P(B_{t_{i_1}}\le a_1+\delta,\dots,B_{t_{i_K}}\le a_K+\delta,W^\Pi_{t_{j_1}}\le b_1+\delta,\dots,W^\Pi_{t_{j_L}}\le b_L+\delta)+\sum_{k=1}^KP(B_{u_k}-B_{t_{i_k}}\le-\delta)\\
&+\sum_{l=1}^LP(W^\Pi_{v_l}-W^\Pi_{t_{j_l}}\le-\delta)-P(B_{t_{i_1}}\le a_1-\delta,\dots,B_{t_{i_K}}\le a_K-\delta,\tilde W^\Pi_{t_{j_1}}\le b_1-\delta,\dots,\tilde W^\Pi_{t_{j_L}}\le b_L-\delta,H)\label{diff prob}.
\end{align}
Note that \eqref{identical distribution of W^Pi} implies
\begin{align*}&P(B_{t_{i_1}}\le a_1-\delta,\dots,B_{t_{i_K}}\le a_K-\delta,\tilde W^\Pi_{t_{j_1}}\le b_1-\delta,\dots,\tilde W^\Pi_{t_{j_L}}\le b_L-\delta)\\
=&P(B_{t_{i_1}}\le a_1-\delta,\dots,B_{t_{i_K}}\le a_K-\delta,W^\Pi_{t_{j_1}}\le b_1-\delta,\dots,W^\Pi_{t_{j_L}}\le b_L-\delta),\end{align*}
and compared the first term and last term in the right hand of \eqref{diff prob}, we have
\begin{align}
\nonumber&P(B_{t_{i_1}}\le a_1+\delta,\dots,B_{t_{i_K}}\le a_K+\delta,W^\Pi_{t_{j_1}}\le b_1+\delta,\dots,W^\Pi_{t_{j_L}}\le b_L+\delta)\\
\nonumber&-P(B_{t_{i_1}}\le a_1-\delta,\dots,B_{t_{i_K}}\le a_K-\delta,\tilde W^\Pi_{t_{j_1}}\le b_1-\delta,\dots,\tilde W^\Pi_{t_{j_L}}\le b_L-\delta,H)\\
\nonumber=&P(B_{t_{i_1}}\le a_1+\delta,\dots,B_{t_{i_K}}\le a_K+\delta,W^\Pi_{t_{j_1}}\le b_1+\delta,\dots,W^\Pi_{t_{j_L}}\le b_L+\delta)\\
\nonumber&-P(B_{t_{i_1}}\le a_1-\delta,\dots,B_{t_{i_K}}\le a_K-\delta,\tilde W^\Pi_{t_{j_1}}\le b_1-\delta,\dots,\tilde W^\Pi_{t_{j_L}}\le b_L-\delta)\\
\nonumber&+P(B_{t_{i_1}}\le a_1-\delta,\dots,B_{t_{i_K}}\le a_K-\delta,\tilde W^\Pi_{t_{j_1}}\le b_1-\delta,\dots,\tilde W^\Pi_{t_{j_L}}\le b_L-\delta,H^c)\\
\nonumber\le&P(B_{t_{i_1}}\le a_1+\delta,\dots,B_{t_{i_K}}\le a_K+\delta,W^\Pi_{t_{j_1}}\le b_1+\delta,\dots,W^\Pi_{t_{j_L}}\le b_L+\delta)\\
\nonumber&-P(B_{t_{i_1}}\le a_1-\delta,\dots,B_{t_{i_K}}\le a_K-\delta,W^\Pi_{t_{j_1}}\le b_1-\delta,\dots,W^\Pi_{t_{j_L}}\le b_L-\delta)+P(H^c)\\
\le&\sum_{k=1}^KP(|B_{t_{i_k}}-a_k|\le\delta)+\sum_{l=1}^LP(|W^\Pi_{t_{j_l}}-b_l|\le\delta)+P(H^c).\label{inequality of distribution}
\end{align}
Substituting \eqref{inequality of distribution} into \eqref{diff prob}, we obtain
\begin{align}\nonumber&P(B_{u_1}\le a_1,\dots,B_{u_K}\le a_K,W^\Pi_{v_1}\le b_1,\dots,W^\Pi_{v_L}\le b_L)-P(B_{u_1}\le a_1,\dots,B_{u_K}\le a_K,\tilde W^\Pi_{v_1}\le b_1,\dots,\tilde W^\Pi_{v_L}\le b_L)\\
\nonumber\le&\sum_{k=1}^KP(B_{u_k}-B_{t_{i_k}}\le-\delta)+\sum_{l=1}^LP(W^\Pi_{v_l}-W^\Pi_{t_{j_l}}\le-\delta)+\sum_{k=1}^KP(|B_{t_{i_k}}-a_k|\le\delta)+\sum_{l=1}^LP(|W^\Pi_{t_{j_l}}-b_l|\le\delta)+P(H^c)\\
\nonumber\le&\sum_{k=1}^KP(B_{u_k}-B_{t_{i_k}}\le-\delta)+\sum_{l=1}^LP(W^\Pi_{v_l}-W^\Pi_{t_{j_l}}\le-\delta)+\sum_{k=1}^KP(|B_{t_{i_k}}-a_k|\le\delta)+\sum_{l=1}^LP(|W^\Pi_{t_{j_l}}-b_l|\le\delta)\\
\nonumber&+\sum_{k=1}^KP(B_{u_k}-B_{t_{i_k}}\ge\delta)+\sum_{l=1}^LP(\tilde W^\Pi_{v_l}-\tilde W^\Pi_{t_{j_l}}\ge\delta)\\
=&2\sum_{k=1}^K\Phi(-\frac{\delta}{\sqrt{u_k-t_{i_k}}})+2\sum_{l=1}^L\Phi(-\frac{\delta}{\sqrt{v_l-t_{j_l}}})+\sum_{k=1}^KP(|B_{t_{i_k}}-a_k|\le\delta)+\sum_{l=1}^LP(|W^\Pi_{t_{j_l}}-b_l|\le\delta),
\label{difference of two distributions}\end{align}
where $\Phi$ denotes the standard normal distribution. For given $K,L,\delta,\epsilon$, it is not difficult to verify that, if 
$$||\Pi||<\frac{\delta^2}{\left(\Phi^{-1}(\frac{\epsilon}{4(K+L)})\right)^2},$$
we have
$$\Phi(-\frac{\delta}{\sqrt{||\Pi||}})<\frac{\epsilon}{4(K+L)},$$
thus
 \begin{equation}2\sum_{k=1}^K\Phi(-\frac{\delta}{\sqrt{u_k-t_{i_k}}})+2\sum_{l=1}^L\Phi(-\frac{\delta}{\sqrt{v_l-t_{j_l}}})\le2(K+L)\Phi(-\frac{\delta}{\sqrt{||\Pi||}})<\frac{\epsilon}2.\label{Pi small enough}\end{equation}
 As a consequence of \eqref{delta small enough}, \eqref{difference of two distributions} and \eqref{Pi small enough},
 \begin{align*}
 P(B_{u_1}\le a_1,\dots,B_{u_K}\le a_K,W^\Pi_{v_1}\le b_1,\dots,W^\Pi_{v_L}\le b_L)-P(B_{u_1}\le a_1,\dots,B_{u_K}\le a_K,\tilde W^\Pi_{v_1}\le b_1,\dots,\tilde W^\Pi_{v_L}\le b_L)\le\epsilon,
 \end{align*}
 similarly,
 \begin{align*}
P(B_{u_1}\le a_1,\dots,B_{u_K}\le a_K,\tilde W^\Pi_{v_1}\le b_1,\dots,\tilde W^\Pi_{v_L}\le b_L)-P(B_{u_1}\le a_1,\dots,B_{u_K}\le a_K,W^\Pi_{v_1}\le b_1,\dots,W^\Pi_{v_L}\le b_L)\le\epsilon,
 \end{align*}
 i.e.
 \begin{align}
 \nonumber|P(B_{u_1}\le a_1,\dots,B_{u_K}\le a_K,&W^\Pi_{v_1}\le b_1,\dots,W^\Pi_{v_L}\le b_L)\\
 &-P(B_{u_1}\le a_1,\dots,B_{u_K}\le a_K,\tilde W^\Pi_{v_1}\le b_1,\dots,\tilde W^\Pi_{v_L}\le b_L)|\le\epsilon,
 \label{tilde W^Pi=W^Pi}\end{align}
 From the definition of It\^{o}'s integral, we have
 \begin{equation}\label{W^Pi convergence}(B_{u_1},B_{u_2},\dots,B_{u_K},W^\Pi_{v_1},W^\Pi_{v_2},\dots,W^\Pi_{v_L})\xrightarrow[\quad]{d}(B_{u_1},B_{u_2},\dots,B_{u_K},W_{v_1},W_{v_2},\dots,W_{v_L}).\end{equation}
 Combining \eqref{tilde W^Pi=W^Pi} and \eqref{W^Pi convergence}, as $||\Pi||\to0$, we have
 \begin{align*}(B_{u_1},B_{u_2},\dots,B_{u_K},\tilde W^\Pi_{v_1},\tilde W^\Pi_{v_2},\dots,\tilde W^\Pi_{v_L})\xrightarrow[\quad]{d}(B_{u_1},B_{u_2},\dots,B_{u_K},W_{v_1},W_{v_2},\dots,W_{v_L}).
 \end{align*}

\end{proofofproposition}

\subsection{Proofs of Results in Section \ref{Construction of Two Correlated Brownian Motions}}\label{proof of section 3}
\begin{proofoftheorem}{\ref{combine BM}}
By optional stopping theorem, $X_T$, $Y_S$ are martingales under $\mathcal F^{X_T}$, $\mathcal F^{Y_S}$ respectively and $\mathcal{F}_t^{Y_S}\perp\mathcal{F}^{X_T}_{\infty}|\mathcal{F}_t^{X_T}$, $\mathcal{F}_t^{X_T}\perp\mathcal{F}^{Y_S}_{\infty}|\mathcal{F}_t^{Y_S}$ guarantee that
$$E\left[\left.X_{T_u}\right|\mathcal F^{X_T,Y_S}_t\right]=E\left[\left.X_{T_u}\right|\mathcal F^{X_T}_t\right]=X_{T_t},\ E\left[\left.Y_{S_u}\right|\mathcal F^{X_T,Y_S}_t\right]=E\left[\left.Y_{S_u}\right|\mathcal F^{Y_S}_t\right]=Y_{S_t},\ u\ge t,$$
which give the martingale properties of $X_T$ and $Y_S$.

If $T$ and $S$ are strictly increasing and $T_t+S_t=t,\forall t$, then \eqref{T_t-T_s} holds. With the same discussion in the beginning of Section \ref{model setup}, $T$ and $S$ are derivable with respect to $t$, let
$$\lambda_t\triangleq\frac{dT_t}{dt},\mu_t\triangleq\frac{dS_t}{dt},$$
and $\tau$, $\varsigma$ be defined as in \eqref{tau,varsigma}. Then, $\lambda_t+\mu_t=1,\forall t$ and $\tau$, $\varsigma$ are continuous and strictly increasing processes.

Next, we claim that $[X_T,Y_S]_t=0$.
We first consider the case that $E{\tau_t},E{\varsigma_t}<\infty$ for any $t>0$. Observe that
\begin{equation}\label{tau_t}\int_0^t\frac{1_{\{\lambda_{\tau_s}\neq0\}}}{\lambda_{\tau_s}}ds+\int_0^t1_{\{\lambda_{\tau_s}=0\}}d\tau_s=\int_0^{\tau_t}\frac{1_{\{\lambda_{s}\neq0\}}}{\lambda_{s}}dT_s+\int_0^{\tau_t}1_{\{\lambda_{s}=0\}}ds=\int_0^{\tau_t}1_{\{\lambda_{s}\neq0\}}ds+\int_0^{\tau_t}1_{\{\lambda_{s}=0\}}ds=\tau_t,\end{equation}
\begin{equation}\label{varsigma_t}\int_0^t\frac{1_{\{\mu_{\varsigma_s}\neq0\}}}{\mu_{\varsigma_s}}ds+\int_0^t1_{\{\mu_{\varsigma_s}=0\}}d\varsigma_s=\int_0^{\varsigma_t}\frac{1_{\{\mu_{s}\neq0\}}}{\mu_{s}}dT_s+\int_0^{\varsigma_t}1_{\{\mu_{s}=0\}}ds=\int_0^{\varsigma_t}1_{\{\mu_{s}\neq0\}}ds+\int_0^{\varsigma_t}1_{\{\mu_{s}=0\}}ds=\varsigma_t,\end{equation}
then
\begin{equation}E\left[\int_0^t\frac{1_{\{\lambda_{\tau_s}\neq0\}}}{\lambda_{\tau_s}}ds+\int_0^t1_{\{\lambda_{\tau_s}=0\}}d\tau_s\right]<\infty,E\left[\int_0^t\frac{1_{\{\mu_{\varsigma_s}\neq0\}}}{\mu_{\varsigma_s}}ds+\int_0^t1_{\{\mu_{\varsigma_s}=0\}}d\varsigma_s\right]<\infty,\forall t>0.\label{Ealpha_t}\end{equation}
Since $\{T_t\}_{t\ge0}$ is a time change of $\mathcal F_t$, so $T_t$ is adapted to $\mathcal{F}_{T_t}$, and $\lambda_t$ is adapted to $\mathcal{F}_{T_t}$, consequently $\lambda_{{\tau_t}}$ is adapted to $\mathcal{F}_{t}$. Similarly, $\mu_{{\varsigma_t}}$ is adapted to $\mathcal{F}_{t}$ as well. According to \eqref{Ealpha_t}, the stochastic processes 
$$M_t\triangleq\int_0^t\frac{1_{\{\lambda_{\tau_s}\neq0\}}}{\sqrt{\lambda_{\tau_s}}}dX_s+\int_0^t1_{\{\lambda_{\tau_s}=0\}}d\tilde X_{\tau_s},N_t\triangleq\int_0^t\frac{1_{\{\mu_{\varsigma_s}\neq0\}}}{\sqrt{\mu_{\varsigma_s}}}dY_s+\int_0^t1_{\{\mu_{\varsigma_s}=0\}}d\tilde Y_{\varsigma_s}$$
are well-defined, where $(\tilde X,\tilde Y)$ is a 2-dimension Brownian motion independent with $\mathcal F_{\infty}$. From \eqref{tau_t} and \eqref{varsigma_t}, we have
\begin{equation}[M]_t={\tau_t},[N]_t={\varsigma_t}.\label{[M] and [N]}\end{equation}
By the independency, it is not difficult to verify that$\{X_tY_t\}_{t\ge0}$, $\{X_t\tilde Y_{\varsigma_t}\}_{t\ge0}$, $\{\tilde X_{\tau_t}Y_t\}_{t\ge0}$ and $\{\tilde X_{\tau_t}\tilde Y_{\varsigma_t}\}_{t\ge0}$ are continuous martingales respectively\footnote{To speak specifically, for example, we could first prove that $\{X_{S_t}\tilde Y_t\}_{t\ge0}$ is a martingale by independency. Note that $\{X_t\tilde Y_{\varsigma_t}\}_{t\ge0}$ can be seen as the time-changed process of $\{X_{S_t}\tilde Y_t\}_{t\ge0}$, and thus $\{X_t\tilde Y_{\varsigma_t}\}_{t\ge0}$ is a martingale. The arguments for $\{\tilde X_{\tau_t}Y_t\}_{t\ge0}$ and $\{\tilde X_{\tau_t}\tilde Y_{\varsigma_t}\}_{t\ge0}$ are similar. The continuity of these processes come from the continuity of $X$, $Y$, $\tilde X$, $\tilde Y$, $\tau$ and $\varsigma$.}, thus
\begin{equation}\label{[X,Y]=0}[X,Y]_t=<X,Y>_t=[X,\tilde Y_{\varsigma}]_t=<X,\tilde Y_{\varsigma}>_t=[\tilde X_{\tau},Y]_t=<\tilde X_{\tau},Y>_t=[\tilde X_{\tau},\tilde Y_{\varsigma}]_t=<\tilde X_{\tau},\tilde Y_{\varsigma}>_t=0.\end{equation}
Consequently,
\begin{align}\nonumber[M,N]_t=&\int_0^t\frac{1_{\{\lambda_{\tau_s}\neq0,\mu_{\varsigma_s}\neq0\}}}{\sqrt{\lambda_{\tau_s}\mu_{\varsigma_s}}}d[X,Y]_s+\int_0^t\frac{1_{\{\lambda_{\tau_s}\neq0,\mu_{\varsigma_s}=0\}}}{\sqrt{\lambda_{\tau_s}}}d[X,\tilde Y_{\varsigma}]_s\\
&+\int_0^t\frac{1_{\{\lambda_{\tau_s}=0,\mu_{\varsigma_s}\neq0\}}}{\sqrt{\mu_{\varsigma_s}}}d[\tilde X_{\tau},Y]_s+\int_0^t1_{\{\lambda_{\tau_s}=0,\mu_{\varsigma_s}=0\}}d[\tilde X_{\tau},\tilde Y_{\varsigma}]_s=0.\label{construct zero quadratic covariation}\end{align}
By the continuity of $\tau$ and $\varsigma$, $M$ and $N$ are continuous as well. Hence according to \cite{revuz2013continuous}[Chapter V, Theorem 1.10] and \eqref{[M] and [N]}, 
$M_T$ and $N_S$ are two independent Brownian motions. As a consequence,
\begin{equation}[M_T,N_S]_t=0,\forall t.\label{[M_T,N_S]=0}\end{equation}
On the other hand, by the definition of $M$ and $N$,
$$M_{T_t}=\int_0^{T_t}\frac{1_{\{\lambda_{\tau_s}\neq0\}}}{\sqrt{\lambda_{\tau_s}}}dX_s+\int_0^{T_t}1_{\{\lambda_{\tau_s}=0\}}d\tilde X_{\tau_s}=\int_0^{t}\frac{1_{\{\lambda_{s}\neq0\}}}{\sqrt{\lambda_{s}}}dX_{T_s}+\int_0^{t}1_{\{\lambda_{s}=0\}}d\tilde X_{s},$$
$$N_{S_t}=\int_0^{S_t}\frac{1_{\{\mu_{\varsigma_s}\neq0\}}}{\sqrt{\mu_{\varsigma_s}}}dY_s+\int_0^{S_t}1_{\{\mu_{\varsigma_s}=0\}}d\tilde Y_{\varsigma_s}=\int_0^{t}\frac{1_{\{\mu_{s}\neq0\}}}{\sqrt{\mu_{s}}}dY_{S_s}+\int_0^{t}1_{\{\mu_{s}=0\}}d\tilde Y_{s},$$
thus
\begin{align}\nonumber[M_T,N_S]_t=&\int_0^t\frac{1_{\{\lambda_{s}\neq0,\mu_{s}\neq0\}}}{\sqrt{\lambda_{s}\mu_{s}}}d[X_T,Y_S]_s+\int_0^t\frac{1_{\{\lambda_{s}\neq0,\mu_{s}=0\}}}{\sqrt{\lambda_{s}}}d[X_T,\tilde Y]_s\\
&+\int_0^t\frac{1_{\{\lambda_{s}=0,\mu_{s}\neq0\}}}{\sqrt{\mu_{s}}}d[\tilde X,Y_S]_s+\int_0^t1_{\{\lambda_{s}=0,\mu_{s}=0\}}d[\tilde X,\tilde Y]_s.\label{[M_T,N_S]}\end{align}
With the similar discussions of \eqref{[X,Y]=0}, we have $[X_T,\tilde Y]_t=[\tilde X,Y_S]_t=0$. Comparing \eqref{[M_T,N_S]=0} and \eqref{[M_T,N_S]}, we obtain
$$\int_0^t\frac{1_{\{\lambda_{s}\neq0,\mu_{s}\neq0\}}}{\sqrt{\lambda_{s}\mu_{s}}}d[X_T,Y_S]_s=0,\forall t\ge0,$$
and immediately,
$$\int_0^t1_{\{\lambda_{s}\neq0,\mu_{s}\neq0\}}d[X_T,Y_S]_s=0,\forall t\ge0.$$
Note that $\int_0^t1_{\{\lambda_s=0\}}dT_s=\int_0^t1_{\{\mu_s=0\}}dS_s=0,$ which implies (\cite{revuz2013continuous}[Chapter IV, Proposition 1.12])
$$\int_0^t1_{\{\lambda_s=0\}}dX_{T_s}=\int_0^t1_{\{\mu_s=0\}}dY_{S_s}=0.$$
Therefore,
$$X_{T_t}=\int_0^t1_{\{\lambda_s\neq0\}}dX_{T_s},Y_{S_t}=\int_0^t1_{\{\mu_s\neq0\}}dY_{S_s},$$
and
$$[X_{T},Y_S]_t=\int_0^t1_{\{\lambda_{s}\neq0,\mu_{s}\neq0\}}d[X_T,Y_S]_s=0.$$

%
%
%
%

If there is a $t>0$ subject to $E{\tau_t}=\infty$ or $E{\varsigma_t}=\infty$, then we define
\begin{alignat*}{2}
\lambda_t^n&\triangleq\begin{cases}\lambda_t,\quad &t\le n\\ \frac12,\quad &t>n\end{cases},&\mu_t^n&\triangleq\begin{cases}\mu_t,\quad &t\le n\\ \frac12,\quad &t>n\end{cases},\\
T_t^n&\triangleq\int_0^t\lambda_u^ndu,&S_t^n&\triangleq\int_0^t\mu_u^ndu,\\
\tau_t^n&\triangleq\inf\{u:T_u^n>t\},&\varsigma_t^n&\triangleq\inf\{u:S_u^n>t\},
\end{alignat*}
hence $E{\tau_t}^n<\infty,E{\varsigma_t}^n<\infty$ and according to the previous proof, we have $[X_{T^n},Y_{S^n}]_t=0$. When $t<n$, $(T_t^n,S_t^n)=(T_t,S_t)$, so we have $[X_T,Y_S]_t=0$. Let $n\to\infty$, we complete the proof of our claim.


Since $[X_T,Y_S]_t=0$, we have
$$[B]_t=[X_T+Y_S]_t=[X_T]_t+[Y_S]_t+2[X_T,Y_S]_t=T_t+S_t=t,$$
similarly,
$$[W]_t=[X_T-Y_S]_t=t.$$
Hence $B$ and $W$ are Brownian motions with respect to $\mathcal F^{B,W}$ (which is equal to $\mathcal F^{X_T,Y_S}$). And $[B,W]_t=[X_T+Y_S,X_T-Y_S]_t=T_t-S_t, t\geq 0$.
%

\end{proofoftheorem}

Through Theorem \ref{combine BM}, the proof of Corollary \ref{XYT independent} is straightforward.

\begin{proofofcorollary}{\ref{XYT independent}}
Let
$$\tilde{\mathcal{F}}_t\triangleq\sigma\{X_u,Y_u,\{T_v\le u\},\{S_v\le u\}:u\le t,\forall v\}.$$
Then from the independency of $\{X_t\}_{t\ge0},\{Y_t\}_{t\ge0},\{T_t\}_{t\ge0}$, we know that $\{X_t\}_{t\ge0},\{Y_t\}_{t\ge0}$ are two standard Brownian motions with respect to $\tilde{\mathcal{F}}_t$. By definition of $\tilde{\mathcal F}$, $\{T_u\le t\},\{S_u\le t\}\in\tilde{\mathcal F}_t$ for any $u>0$, hence $T_u,S_u$ are stopping times, and $\{T_t\}_{t\ge0},\{S_t\}_{t\ge0}$ are time changes of $\tilde{\mathcal F}$.

Then by Lemma \ref{lemma of XYT independent}, the conditions in Theorem \ref{combine BM} are satisfied, and we get the desired result.
\end{proofofcorollary}
\subsection{Proof of Results in Section \ref{Pricing Financial Derivatives by Decomposition of Two Correlated Brownian Motions}}
\begin{proofofproposition}{\ref{fourier transform}}
By the definition of $\hat G$,
\begin{equation}\hat{G}(\lambda_1,\lambda)=\int_{-\infty}^{\infty}\int_{-\infty}^{\infty}
e^{i\lambda_1x_1+i\lambda x}G(x_1,x)dx_1dx.\label{fourier transform for G}\end{equation}
According to Fubini theorem,
\begin{align}\nonumber\int_{-\infty}^{\infty}e^{i\lambda_1x_1}G(x_1,x)dx_1=&E\left[\int_{-\infty}^{\infty}e^{i\lambda_1x_1}(\gamma_1+\gamma_2e^{\boldsymbol{\gamma}_3^\top\boldsymbol{M}_{\tau}})1_{\{\boldsymbol{\gamma}_4^\top\boldsymbol{M}_{\tau}\le x_1\}}1_{\{\boldsymbol{\gamma}_5^\top\boldsymbol{M}_{\tau}\le x\}}dx_1\right]\\
\nonumber=&E\left[(\gamma_1+\gamma_2e^{\boldsymbol{\gamma}_3^\top\boldsymbol{M}_{\tau}})1_{\{\boldsymbol{\gamma}_5^\top\boldsymbol{M}_{\tau}\le x\}}\int_{-\infty}^{\infty}e^{i\lambda_1x_1}1_{\{\boldsymbol{\gamma}_4^\top\boldsymbol{M}_{\tau}\le x_1\}}dx_1\right]\\
=&\frac1{i\lambda_1}E\left[e^{i\lambda_1\boldsymbol{\gamma}_4^\top\boldsymbol{M}_{\tau}}(\gamma_1+\gamma_2e^{\boldsymbol{\gamma}_3^\top\boldsymbol{M}_{\tau}})1_{\{\boldsymbol{\gamma}_5^\top\boldsymbol{M}_{\tau}\le x\}}\right],\label{calculate fourier transform}\end{align}
where the last equality comes from the fact that the imaginary part of $\lambda_1$ is positive. 
Substituting \eqref{calculate fourier transform} into \eqref{fourier transform for G}, then with similar calculation for $x$, we have 
\begin{align*}\hat{G}(\lambda_1,\lambda)=&\int_{-\infty}^{\infty}\frac1{i\lambda_1}e^{i\lambda x}E\left[e^{i\lambda_1\boldsymbol{\gamma}_4^\top\boldsymbol{M}_{\tau}}(\gamma_1+\gamma_2e^{\boldsymbol{\gamma}_3^\top\boldsymbol{M}_{\tau}})1_{\{\boldsymbol{\gamma}_5^\top\boldsymbol{M}_{\tau}\le x\}}\right]dx\\
=&-\frac1{\lambda\lambda_1}E\left[e^{i\lambda_1\boldsymbol{\gamma}_4^\top\boldsymbol{M}_{\tau}+i\lambda\boldsymbol{\gamma}_5^\top\boldsymbol{M}_{\tau}}(\gamma_1+\gamma_2e^{\boldsymbol{\gamma}_3^\top\boldsymbol{M}_{\tau}})\right]\\
=&-\frac{\gamma_1}{\lambda\lambda_1}\Phi_{\boldsymbol{M}_{\tau}}(\lambda_1\boldsymbol{\gamma}_4+\lambda\boldsymbol{\gamma}_5)-\frac{\gamma_2}{\lambda\lambda_1}\Phi_{\boldsymbol{M}_{\tau}}(\lambda_1\boldsymbol{\gamma}_4+\lambda\boldsymbol{\gamma}_5-i\boldsymbol{\gamma}_3),
\end{align*}
where $\Phi_{\boldsymbol{M}_{\tau}}$ denotes the characteristic function of $\boldsymbol{M}_{\tau}$. Thus the proof of \eqref{hat G} is completed.

As for \eqref{Phi M}, note that $X,Y$ and $T$ are mutually independent, it can be calculated by conditional expectation
\begin{align*}
\Phi_{\boldsymbol{M}_{\tau}}(z_1,z_2)=&Ee^{iz_1X_{T_\tau}+iz_2Y_{S_\tau}}=E\left[E[e^{iz_1X_{T_\tau}+iz_2Y_{S_\tau}}|T_\tau,S_\tau]\right]=Ee^{-\frac12T_\tau z_1^2-\frac12S_\tau z_2^2}\\
=&e^{-\frac12\tau z_2^2}Ee^{-\frac12(z_1^2-z_2^2)T_\tau}=e^{-\frac12\tau z_2^2}L_\tau(-\frac12(z_1^2-z_2^2)),
\end{align*}
where $L_t$ represents the generalized fourier transform of $T_t$ at time $t$.
\end{proofofproposition}

\section{Conclusion}
By applying time-change technique, we propose a new method so called common decomposition to study dependency structure for two correlated Brownian motions $(B,W)$. The common decomposition triplet of $(B,W)$ is denoted by $(X,Y,T)$.

We find that $X$ and $Y$ are two independent Brownian motions, $T$ is a time-change process, and we give three equivalence conditions (C1), (C2) and (C3) for the mutual independency of $X$, $Y$ and $T$. The condition (C1) is given from the aspect of filtration. The condition (C2) gives a generalization of Girsanov theorem and we give an example to show that the invariance property of $T$ under the change of measure by applying the condition (C2). The condition (C3) give connections between common decomposition and local correlation.

Conversely, we construct two correlated Brownian motions based on the common decomposition. Furthermore, the simulation method is given from the common decomposition and may have some advantages compared with the Euler-Maruyama scheme under some conditions.

Pricing covariance swap, covariance option and Quanto option show the direct usage of the common decomposition. Moreover, the price and Greeks of 2-color rainbow options is given by combining common decomposition and Fourier transform.

Finally, a numerical experiment is designed to show the difference between stochastic correlation and constant correlation for the price of rainbow options. We find that the results are truly different for Call on Min, Put on Min and Put on Max options in the out-of-the-money case but have little differences for in-the-money case. As for the Call on Max option, the results are always similar for stochastic correlation and constant correlation. We also analyze the pricing error in theoretical and interpret the phenomenon discovered in previous.

\bibliographystyle{elsarticle-harv}

\bibliography{deco-BM_reference}

\end{document}